\documentclass[preprint,11pt]{elsarticle}

\usepackage{amssymb}
\usepackage{amsthm}
\theoremstyle{definition}
\newtheorem{remark}{Remark}
\newtheorem{proposition}{Proposition}

\usepackage{geometry}
\usepackage{amsmath}
\usepackage{mathabx}
\usepackage[ruled,linesnumbered]{algorithm2e}
\usepackage{graphicx}
\usepackage{float}
\usepackage{upgreek}
\usepackage{hyperref}
\hypersetup{colorlinks=true,linkcolor=black}
\SetArgSty{textnormal}

\title{Multimaterial topology optimization for finite strain elastoplasticity: theory, methods, and applications}

\author[1]{Yingqi Jia}

\author[1,2,3]{Xiaojia Shelly Zhang\texorpdfstring{\corref{corr-author}}}
\cortext[corr-author]{Corresponding author.}
\ead{zhangxs@illinois.edu}

\address[1]{Department of Civil and Environmental Engineering, University of Illinois Urbana-Champaign, Urbana, IL 61801, USA}
\address[2]{Department of Mechanical Science and Engineering, University of Illinois Urbana-Champaign, Urbana, IL 61801, USA}
\address[3]{National Center for Supercomputing Applications, Urbana, IL 61801, USA}

\begin{document}

\begin{abstract}
Plasticity is inherent to many engineering materials such as metals. While it can degrade the load-carrying capacity of structures via material yielding, it can also protect structures through plastic energy dissipation. To fully harness plasticity, here we present the theory, method, and application of a topology optimization framework that simultaneously optimizes structural geometries and material phases to customize the stiffness, strength, and structural toughness of designs experiencing finite strain elastoplasticity. The framework accurately predicts structural responses by employing a rigorous, mechanics-based elastoplasticity theory that ensures isochoric plastic flow. It also effectively identifies optimal material phase distributions using a gradient-based optimizer, where gradient information is obtained via a reversed adjoint method to address history dependence, along with automatic differentiation to compute the complex partial derivatives. We demonstrate the framework's capabilities by optimizing a range of 2D and 3D multimaterial elastoplastic structures with real-world applications, including energy-dissipating dampers, load-carrying beams, impact-resisting bumpers, and cold working profiled sheets. These optimized multimaterial structures reveal important mechanisms for improving design performance under large deformation, such as the transition from kinematic to isotropic hardening with increasing displacement amplitudes and the formation of twisted regions that concentrate stress, enhancing plastic energy dissipation. Through the superior performance of these optimized designs, we demonstrate the framework’s effectiveness in tailoring elastoplastic responses across various spatial configurations, material types, hardening behaviors, and combinations of candidate materials. This work offers a systematic approach for optimizing next-generation multimaterial structures with elastoplastic behaviors under large deformations.

\endkeyword
\end{abstract}

\maketitle

\section{Introduction}

Plasticity is a fundamental physical phenomenon describing the permanent deformation of solids \citep{chaboche_review_2008} --- a trait that transcends elasticity and is inherent to critical materials such as metals, woods, foams, and soils. These materials form the backbone of engineering structures. The irreversible nature of plasticity holds two faces: it can cause undesirable deformations and a loss of structural load-carrying capacity or enable the formation of desired geometries and the dissipation of energy. In addition, plasticity is an inherently complex behavior, involving history-dependent stress-strain responses, material hardening, and large deformations that interact across multiple scales. Thus, there is a need for innovative methods to design and control plasticity at will.

To harness material plasticity, topology optimization \citep{bendsoe_generating_1988, bendsoe_topology_2003, wang_comprehensive_2021} as a form of computational morphogenesis offers great promise. This approach strategically distributes limited materials to optimize structural responses under given constraints. Therefore, it is naturally suited for controlling mechanical responses such as displacement fields \citep{jia_topology_2024, jia_unstructured_2024}, stress distributions \citep{jia_modulate_2024}, and fracture nucleation and propagation \citep{jia_controlling_2023}.

In light of its potential for controlling mechanical responses, topology optimization has recently achieved notable progress in tailoring plastic behaviors under the assumption of infinitesimal strains. For instance, employing a single material type with restricted usage, researchers have shown success in limiting plastic yielding \citep{amir_stress-constrained_2017, zhang_framework_2023, li_three-dimensional_2023}, maximizing end compliance/force \citep{boissier_elastoplastic_2021, desai_topology_2021}, and maximizing energy \citep{maute_adaptive_1998, zhang_topology_2017, li_topology_2017, li_design_2017, li_design_2017-1, alberdi_topology_2017, alberdi_design_2019}. Furthermore, studies have explored maximizing both compliance and energy \citep{abueidda_topology_2021} and fitting target stress--strain responses \citep{kim_microstructure_2020}. Beyond single-material strategies, \citet{nakshatrala_topology_2015} minimized energy propagation by filling distinct materials into fixed geometries, while \citet{li_topology_2021} optimized end force in similar setups.

Despite these significant advances in controlling structural plasticity under small deformations, optimizing finite strain plastic responses remains in its infancy. Only a few pioneering studies have tackled this challenge in the past decade. For example, \citet{wallin_topology_2016} and \citet{ivarsson_topology_2018, ivarsson_topology_2020} focused on optimizing structural geometries to improve system energy by specifying a single candidate material. Similarly, \citet{ivarsson_plastic_2021} maximized compliance while constraining plastic work through single-material topology optimization. In addition, \citet{alberdi_bi-material_2019} employed two candidate materials in fixed geometries to maximize plastic work. More recently, \citet{han_topology_2024} separately optimized structural geometries and material phases in different design cases to enhance stiffness and plastic hardening.

A review of past efforts highlights a key limitation: most studies either optimize structural geometries while fixing material types or optimize material phases while fixing structural geometries. To fully leverage topology optimization’s potential and further improve structural performance, simultaneous optimization of structural geometries and material phases is desired.
Such an approach would enable the creation of free-form, multimaterial structures with optimized plastic responses. Recently, \citet{jia_multimaterial_2025} pursued this goal and demonstrated preliminary success for infinitesimal strain plasticity; however, simultaneous optimization for finite strain plasticity remains under-explored.

In this study, we propose a multimaterial topology optimization framework for finite strain elastoplasticity. As illustrated in Fig. \ref{Fig: Idea Figure}(a), our goal is to \textit{concurrently} optimize structural geometries and material phases to maximize stiffness, strength (end force), and total energy (effective structural toughness) of structures undergoing finite strain plasticity. To achieve this goal, we present a multimaterial optimization approach in Fig. \ref{Fig: Idea Figure}(b), which addresses two major challenges --- accurate prediction of finite strain elastoplastic responses and path-dependent sensitivity analysis. To precisely predict plastic behaviors, we adopt a rigorous mechanics-based elastoplasticity theory from \citet{simo_framework_1988-1, simo_framework_1988}, which accounts for large deformations. We further adapt it to enforce isochoric plastic flow, accommodating a wide range of plastic materials including metals. To handle the path-dependent sensitivity analysis necessary for updating design variables, we employ the reversed adjoint method \citep{alberdi_unified_2018, jia_multimaterial_2025} to address history dependency and leverage automatic differentiation to effortlessly compute the required partial derivatives. By resolving these challenges, the framework establishes a closed-loop process that enables the creation of freeform, multimaterial optimized elastoplastic structures.

\begin{figure}[!htbp]
    \centering
    \includegraphics[width=18cm]{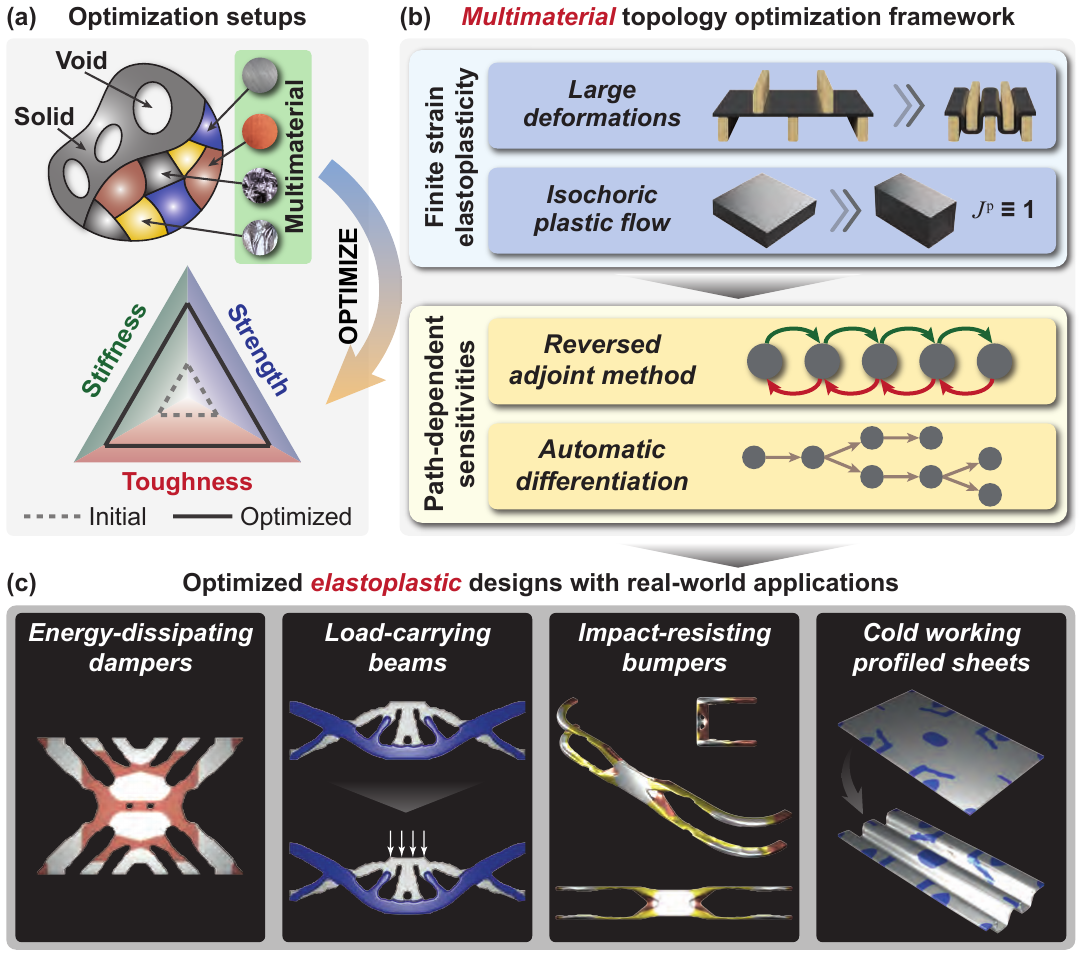}
    \caption{Multimaterial and multiobjective topology optimization for finite strain elastoplasticity. (a) Optimization setups. (b) Multimaterial topology optimization framework. The variable $J^\texttt{p}$ is the determinant of the plastic part of the deformation gradient. (c) Optimized elastoplastic designs with real-world applications.}
    \label{Fig: Idea Figure}
\end{figure}

Based on the proposed framework, we optimize a spectrum of representative multimaterial elastoplastic structures with real-world applications (Fig. \ref{Fig: Idea Figure}(c)) and uncover the mechanisms behind their optimized behaviors. For example, we design two-dimensional (2D) metallic yielding dampers with maximized energy dissipation for vibration mitigation under cyclic loadings, revealing a shift in dominance from kinematic to isotropic hardening as displacement amplitudes increase. We also explore the synergistic use of hyperelastic and elastoplastic materials to achieve various stiffness--strength balances in load-carrying beams. Extending the design to three dimensions (3D), we maximize the crashworthiness of impact-resisting bumpers, uncovering the formation of X-shaped structures in the middle to shorten the load path and enhance elastic energy absorption, as well as twisted regions aside to concentrate stress and increase plastic energy dissipation. Finally, we present optimized cold working profiled sheets undergoing ultra-large plastic deformations, demonstrating the simultaneous optimization of metal-forming and load-carrying stages while incorporating practical constraints such as cost, lightweight, and sustainability \citep{kundu2025sustainability}.

The optimized designs showcase the effectiveness of the proposed framework in freely controlling elastoplastic responses under large deformations in 2D and 3D, across different material types (elastic/plastic/mixed), various hardening behaviors (perfect/isotropic/kinematic/combined and linear/nonlinear), and arbitrary numbers of candidate materials (single-material/bi-material/multimaterial). Ultimately, the framework can provide a systematic tool for designing the next generation of multimaterial elastoplastic structures.

The remaining parts of this study are organized as follows. Section \ref{Sec: Finite Strain Elastoplasticity} recaps the finite strain elastoplasticity theory presented in \citet{simo_framework_1988-1, simo_framework_1988}, with additional emphasis on enforcing isochoric plastic flow. Building on this theory, Section \ref{Sec: Topology Optimization Framework} introduces the proposed multimaterial topology optimization framework for tailoring structural elastoplastic responses. To demonstrate the framework’s effectiveness, Section \ref{Sec: Sample Examples} presents a broad spectrum of optimized elastoplastic designs with real-world applications. Finally, Section \ref{Sec: Conclusions} provides concluding remarks.

This study is further supplemented by six appendices. \ref{Sec: Derivation of Updating Formula of be_bar} derives the proposed updating formulae for the internal variable to ensure isochoric plastic flow. \ref{Sec: Second Elastoplastic Moduli} presents the derivation of the second elastoplastic moduli, a key component of the elastoplasticity theory for a quadratic convergence rate. \ref{Sec: FEA Verification} verifies the accuracy of the theory and implementation by comparing it to analytical solutions. \ref{Sec: FEA Convergence} investigates the convergence, precision, and computational time of the finite element analysis (FEA). \ref{Sec: Sensitivity Analysis and Verification} details the derivation and verification of the path-dependent sensitivity analysis. Finally, \ref{Sec: Comparison of Meshes} examines the consistency of elastoplastic responses across various mesh and finite element combinations for the optimized designs.

\section{Finite strain elastoplasticity: theory and implementation}
\label{Sec: Finite Strain Elastoplasticity}

In this section, we provide a recap of the finite strain elastoplasticity theory from \citet{simo_framework_1988-1, simo_framework_1988} and outline its numerical implementation. We begin with an overview of the fundamentals of finite strain deformation and strain tensors, followed by a presentation of the local governing equations. Lastly, we introduce the global governing equations for finite strain elastoplasticity. The detailed explanations are presented below.

\subsection{Prerequisites: finite strain deformation and strain tensors}

We consider an open bounded domain ($\Omega_0 \subset \mathbb{R}^3$) occupied by a piece of undeformed material. For any material point $\mathbf{X} \in \Omega_0$, we define its displacement field as $\mathbf{u}(\mathbf{X})$ and the total deformation gradient as $\mathbf{F}(\mathbf{X}) = \mathbf{I} + \nabla \mathbf{u}(\mathbf{X})$. We also define the left and right Cauchy--Green deformation tensors as
\begin{equation*}
    \mathbf{b}(\mathbf{X}) = \mathbf{F}(\mathbf{X}) \mathbf{F}^\top(\mathbf{X})
    \quad \text{and} \quad
    \mathbf{C}(\mathbf{X}) = \mathbf{F}^\top(\mathbf{X}) \mathbf{F}(\mathbf{X}),
\end{equation*}
respectively. Additionally, we define the Lagrangian strain tensor as
\begin{equation} \label{Finite Strain Tensors}
    \mathbf{E}(\mathbf{X}) = \dfrac{1}{2}(\mathbf{C} - \mathbf{I}).
\end{equation}
Based on the multiplicative decomposition, the total deformation gradient can be split as
\begin{equation*}
    \mathbf{F}(\mathbf{X}) = \mathbf{F}^\texttt{e}(\mathbf{X}) \mathbf{F}^\texttt{p}(\mathbf{X})
\end{equation*}
where $\mathbf{F}^\texttt{e}(\mathbf{X})$ and $\mathbf{F}^\texttt{p}(\mathbf{X})$ are the elastic and plastic parts of $\mathbf{F}(\mathbf{X})$, respectively, with $\mathbf{F}^\texttt{p}$ corresponding to stress-free plastic deformation. For later use, we further define the elastic part of $\mathbf{b}(\mathbf{X})$ as
\begin{equation} \label{Elastic Part of b}
    \mathbf{b}^\texttt{e}(\mathbf{X})
    = \mathbf{F}^\texttt{e}(\mathbf{X}) {\mathbf{F}^\texttt{e}}^\top(\mathbf{X})
\end{equation}
and the plastic part of $\mathbf{C}(\mathbf{X})$ as
\begin{equation} \label{Plastic Part of C}
    \mathbf{C}^\texttt{p}(\mathbf{X})
    = {\mathbf{F}^\texttt{p}}^\top(\mathbf{X}) \mathbf{F}^\texttt{p}(\mathbf{X}).
\end{equation}

The relationships among these deformation tensors are illustrated in Fig. \ref{Fig: Configurations}. At load step $n$, the undeformed configuration ($\Omega_0$) is mapped to the intermediate configuration ($\Xi_n$) through plastic deformation ($\mathbf{F}_n^\texttt{p}$ or $\mathbf{C}_n^\texttt{p}$), and further to the deformed configuration ($\Omega_n$) via elastic deformation ($\mathbf{F}_n^\texttt{e}$ or $\mathbf{b}_n^\texttt{e}$). Alternatively, $\Omega_0$ can be directly mapped to $\Omega_n$ through the total deformation ($\mathbf{F}_n$). Similarly, at the next load step $n+1$, $\Omega_0$ maps to the intermediate configuration ($\Xi_{n+1}$) through $\mathbf{F}_{n+1}^\texttt{p}$ or $\mathbf{C}_{n+1}^\texttt{p}$, and subsequently to the deformed configuration ($\Omega_{n+1}$) via $\mathbf{F}_{n+1}^\texttt{e}$ or $\mathbf{b}_{n+1}^\texttt{e}$. Note that $\Omega_0$ can also directly map to $\Omega_{n+1}$ through $\mathbf{F}_{n+1}$. Additionally, $\Omega_n$ can be mapped to $\Omega_{n+1}$ using the relative deformation gradient expressed as
\begin{equation*}
    \mathbf{f}_{n+1} = \mathbf{F}_{n+1} \mathbf{F}_n^{-1}.
\end{equation*}
Finally, for later use, we define the volume-preserving parts of these deformation and strain tensors as
\begin{equation*}
    \overline{\mathbf{F}} = J^{-1/3} \mathbf{F}, \quad
    \overline{\mathbf{b}}^\texttt{e} = {(J^\texttt{e})}^{-2/3} \mathbf{b}^\texttt{e}, \quad \text{and} \quad
    \overline{\mathbf{f}} = J_f^{-1/3} \mathbf{f},
\end{equation*}
and express the determinants as
\begin{equation*}
    J = \text{det}(\mathbf{F}), \quad
    J^\texttt{e} = \text{det}(\mathbf{F}^\texttt{e}), \quad
    J^\texttt{p} = \text{det}(\mathbf{F}^\texttt{p}), \quad \text{and} \quad
    J_f = \text{det}(\mathbf{f}).
\end{equation*}

\begin{figure}[!htbp]
    \centering
    \includegraphics[width=13cm]{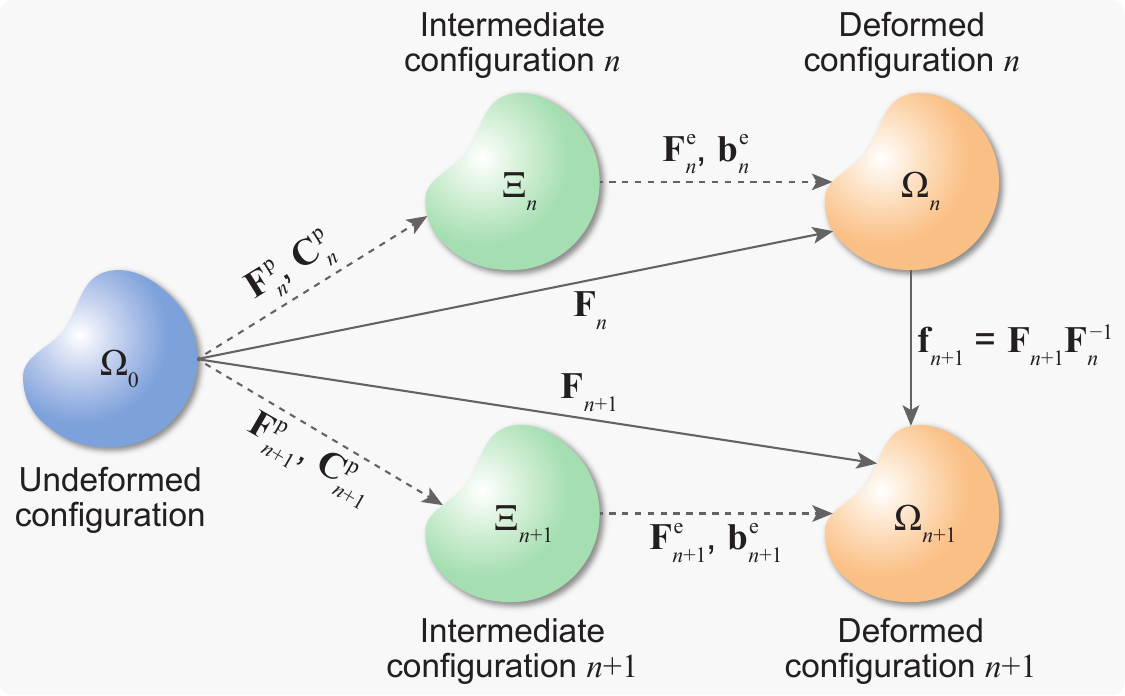}
    \caption{Relationships among the deformation tensors used in the finite strain elastoplasticity theory.}
    \label{Fig: Configurations}
\end{figure}

\subsection{Local governing equations for finite strain elastoplasticity}

In this subsection, we introduce the local governing equations for finite strain elastoplasticity, including the constitutive relationships, flow rules, hardening laws, and yield conditions. Combining these principles, we conclude with the radial return mapping scheme, which is employed to determine the elastoplastic response of a material point.

\subsubsection{Constitutive relationships}

We consider a hyperelastic material with a strain energy density function expressed as
\begin{equation} \label{Strain Energy Density Function}
    W(J^\texttt{e}, \overline{\mathbf{b}}^\texttt{e}) = U(J^\texttt{e}) + \overline{W}(\overline{\mathbf{b}}^\texttt{e}),
\end{equation}
which is characterized by separable volumetric ($U$) and deviatoric ($\overline{W}$) parts. One classical choice of such function reads as \citep{simo_computational_2006}
\begin{equation} \label{Volumetric and Deviatoric Engery}
    U(J^\texttt{e}) = \dfrac{\kappa}{2} \left\{ \frac{1}{2}\left[(J^\texttt{e})^2-1\right] - \ln J^\texttt{e} \right\}
    \quad \text{and} \quad
    \overline{W}(\overline{\mathbf{b}}^\texttt{e})
    = \dfrac{\mu}{2} \left[ \text{tr}(\overline{\mathbf{b}}^\texttt{e}) - 3 \right],
\end{equation}
where $\mu$ and $\kappa$ are the initial shear and bulk moduli, respectively. Based on the strain energy density function in \eqref{Strain Energy Density Function}, we derive the Kirchhoff stress tensor as
\begin{equation} \label{Kirchhoff Stress}
    \widehat{\boldsymbol{\tau}} = \dfrac{\partial W}{\partial \mathbf{F}^\texttt{e}} {\mathbf{F}^\texttt{e}}^\top
    = \boldsymbol{\tau}^\texttt{vol} + \widehat{\mathbf{s}}
    = J^\texttt{e} U'(J^\texttt{e}) \mathbf{I} + \mu\ \text{dev} (\overline{\mathbf{b}}^\texttt{e}),
\end{equation}
where we separate $\widehat{\boldsymbol{\tau}}$ into the volumetric ($\boldsymbol{\tau}^\texttt{vol} = J^\texttt{e} U'(J^\texttt{e}) \mathbf{I}$) and deviatoric ($\widehat{\mathbf{s}} = \mu\ \text{dev} (\overline{\mathbf{b}}^\texttt{e})$) parts. Finally, we define the net stress tensor as
\begin{equation*}
    \boldsymbol{\xi} = \widehat{\mathbf{s}} - \text{dev}(\overline{\boldsymbol{\beta}})
    \quad \text{with} \quad
    \overline{\boldsymbol{\beta}} = J^{-2/3} \boldsymbol{\beta}
\end{equation*}
where $\boldsymbol{\beta}$ is the back stress tensor triggered by the kinematic hardening of materials.

\subsubsection{Flow rules and kinematic hardening laws}

After defining the deformation and strain tensors as well as the constitutive relationships, we proceed with the flow rules and hardening laws. In the context of the finite strain elastoplasticity \citep{simo_framework_1988-1, simo_framework_1988}, we present the flow rules as 
\begin{equation} \label{Flow Rules}
    \mu J^{-2/3} \text{dev}( \mathcal{L}_v \mathbf{b}^\texttt{e} )
    = - 2 \overline{\overline{\mu}} \gamma \mathbf{n}
    \quad \text{and} \quad
    \text{tr}( \mathcal{L}_v \mathbf{b}^\texttt{e} ) = 0,
\end{equation}
and Prager--Ziegler-type kinematic hardening laws as
\begin{equation} \label{Kinematic Hardening Laws}
    \mu J^{-2/3} \text{dev}( \mathcal{L}_v \boldsymbol{\beta} )
    = \dfrac{2}{3} h \overline{\overline{\mu}} \gamma \mathbf{n}
    \quad \text{and} \quad
    \text{tr}( \mathcal{L}_v \boldsymbol{\beta} ) = 0.
\end{equation}
Here, $h$ is the kinematic hardening modulus of materials, and $\gamma \geq 0$ is a consistency parameter that will be determined later. The variable $\mathbf{n}$ is a unit tensor defined as $\mathbf{n} = \boldsymbol{\xi} / \lVert \boldsymbol{\xi} \rVert$ where $\lVert \boldsymbol{\xi} \rVert = \sqrt{\boldsymbol{\xi} : \boldsymbol{\xi}}$. The parameter $\overline{\overline{\mu}}$ is defined as
\begin{equation*}
    \overline{\overline{\mu}} = \overline{\mu} - \dfrac{1}{3} \text{tr}(\overline{\boldsymbol{\beta}})
    \quad \text{with} \quad
    \overline{\mu} = \dfrac{\mu}{3} J^{-2/3} \text{tr}(\mathbf{b}^\texttt{e}).
\end{equation*}
Additionally, the operator $\mathcal{L}_v$ in \eqref{Flow Rules}--\eqref{Kinematic Hardening Laws} is the Lie time derivative defined as
\begin{equation*}
    \mathcal{L}_v \mathbf{T} = \varphi_* \left[ \dfrac{\partial}{\partial t} \varphi^* (\mathbf{T}) \right]
\end{equation*}
for a dummy tensor $\mathbf{T}$ where $t$ represents the time. The symbols $\varphi^*$ and $\varphi_*$ are the pull-back and push-forward operators expressed as
\begin{equation*}
    \varphi^*(\mathbf{T}) = \mathbf{F}^{-1} \mathbf{T} \mathbf{F}^{-\top}
    \quad \text{and} \quad
    \varphi_*(\mathbf{T}) = \mathbf{F} \mathbf{T} \mathbf{F}^\top,
\end{equation*}
respectively.

\subsubsection{Yield criterion and isotropic hardening laws}

To account for the yielding of materials, we adopt the classical von Mises yield criterion defined as
\begin{equation} \label{Yield Criterion}
    f(\boldsymbol{\xi}, \alpha) = \lVert \boldsymbol{\xi} \rVert - \sqrt{\dfrac{2}{3}} k(\alpha) \leq 0
\end{equation}
where $f(\boldsymbol{\xi}, \alpha) < 0$ represents that the material deforms elastically, and $f(\boldsymbol{\xi}, \alpha) = 0$ denotes the onset of yielding. To enforce the yield criterion of $f(\boldsymbol{\xi}, \alpha) \leq 0$, we apply the Kuhn--Tucker conditions expressed as
\begin{equation} \label{Kuhn Tucker Conditions}
    \gamma \geq 0
    \quad \text{and} \quad
    \gamma f(\boldsymbol{\xi}, \alpha) = 0.
\end{equation}

In the above equations, the non-decreasing variable $\alpha$ is the equivalent plastic strain, whose evolution is governed by \citep{simo_framework_1988-1, simo_framework_1988}
\begin{equation} \label{Evolution of Equivalent Plastic Strain}
    \dfrac{\partial \alpha}{\partial t} = \sqrt{\dfrac{2}{3}} \gamma.
\end{equation}
The function $k(\alpha)$ in \eqref{Yield Criterion} determines the radius of the yield surface and characterizes the isotropic hardening of materials, which can follow either a linear or nonlinear form. For linear isotropic hardening, we set
\begin{equation} \label{Linear Isotropic Hardening Law}
    k(\alpha) = \sigma_y + K \alpha
\end{equation}
where $\sigma_y$ is the initial yield strength of the material, and $K$ is the isotropic hardening modulus. For nonlinear isotropic hardening, we adopt a classical expression as \citep{simo_framework_1988}
\begin{equation} \label{Nonlinear Isotropic Hardening Law}
    k(\alpha) = \sigma_y + K \alpha + (\sigma_\infty - \sigma_y)(1 - e^{- \delta \alpha})
\end{equation}
where $\sigma_\infty$ is the residual yield strength, and $\delta>0$ is a saturation exponent.

\subsubsection{Radial return mapping scheme}
\label{Return mapping scheme}

In previous discussions, we have established the constitutive relationship in \eqref{Kirchhoff Stress}, flow rules in \eqref{Flow Rules}, kinematic hardening laws in \eqref{Kinematic Hardening Laws}, yield criterion in \eqref{Yield Criterion}, Kuhn--Tucker conditions in \eqref{Kuhn Tucker Conditions}, the evolution of the equivalent plastic strain in \eqref{Evolution of Equivalent Plastic Strain}, and isotropic hardening laws in \eqref{Linear Isotropic Hardening Law}--\eqref{Nonlinear Isotropic Hardening Law}. A collection of these requirements fully governs the finite strain elastoplastic responses of any material point $\mathbf{X} \in \Omega_0$. To reconcile these governing equations and solve for the material elastoplastic responses, we adopt a canonical radial return mapping scheme \citep{wilkins1969calculation, krieg_accuracies_1977} as follows.

Applying a backward Euler difference scheme to the above local governing equations, we derive their time-discretized versions as
\begin{equation} \label{Local Governing Equations}
    \left\{ \begin{array}{l}
        \overline{\mathbf{b}}_{n+1}^\texttt{e}
        = \overline{\mathbf{f}}_{n+1} \overline{\mathbf{b}}_n^\texttt{e} \overline{\mathbf{f}}_{n+1}^\top
        - \dfrac{2 \overline{\overline{\mu}}_{n+1}}{\mu} \widehat{\gamma}_{n+1} \mathbf{n}_{n+1}, \quad
        \overline{\boldsymbol{\beta}}_{n+1} = \overline{\mathbf{f}}_{n+1} \overline{\boldsymbol{\beta}}_n \overline{\mathbf{f}}_{n+1}^\top
        + \dfrac{2 h \overline{\overline{\mu}}_{n+1}}{3 \mu} \widehat{\gamma}_{n+1} \mathbf{n}_{n+1}, \\[10pt]

        \alpha_{n+1} = \alpha_{n} + \sqrt{\dfrac{2}{3}} \widehat{\gamma}_{n+1}, \quad
        \widehat{\mathbf{s}}_{n+1} = \mu\ \text{dev}(\overline{\mathbf{b}}_{n+1}^\texttt{e}), \quad
        \boldsymbol{\xi}_{n+1} = \widehat{\mathbf{s}}_{n+1} - \text{dev} \left( \overline{\boldsymbol{\beta}}_{n+1} \right),
        \\[10pt]

        \mathbf{n}_{n+1} = \boldsymbol{\xi}_{n+1} / \lVert \boldsymbol{\xi}_{n+1} \rVert, \quad
        \overline{\mu}_{n+1} = \dfrac{\mu}{3} \text{tr}(\overline{\mathbf{b}}_{n+1}^\texttt{e}), \quad
        \overline{\overline{\mu}}_{n+1} = \overline{\mu}_{n+1} - \dfrac{1}{3} \text{tr} (\overline{\boldsymbol{\beta}}_{n+1}), \\[10pt]

        f_{n+1} = \lVert \boldsymbol{\xi}_{n+1} \rVert - \sqrt{\dfrac{2}{3}} k(\alpha_{n+1}) \leq 0, \quad
        \widehat{\gamma}_{n+1} \geq 0, \quad
        \widehat{\gamma}_{n+1} f_{n+1} = 0.
    \end{array} \right.
\end{equation}
In these expressions, we recall that $\mathbf{f}_{n+1} = \mathbf{F}_{n+1} \mathbf{F}_n^{-1}$ is the relative deformation gradient between load steps $n$ and $n+1$, and $\overline{\mathbf{f}}_{n+1} = (J_{n+1}^f)^{-1/3} \mathbf{f}_{n+1}$ is its volume-preserving part with $J_{n+1}^f = \det (\mathbf{f}_{n+1})$. The variable $\widehat{\gamma}_{n+1}$ is shorted for $\Delta t \gamma_{n+1}$, and $\Delta t$ is the time interval.

\begin{remark}
For the derivation of $\overline{\mathbf{b}}_{n+1}^\texttt{e}$ and $\overline{\mu}_{n+1}$ in \eqref{Local Governing Equations} and all subsequent analyses, we adopt the assumption of isochoric plastic flow ($J_n^\texttt{p} = J_{n+1}^\texttt{p} = 1$). This assumption is commonly accepted and even required \citep{weber_finite_1990, simo_algorithms_1992, simo_associative_1992, wang_how_2017} for many elastoplastic materials such as metals, and it implies that the volume of the material does not change during plastic deformation.
\end{remark}

To solve \eqref{Local Governing Equations} with given $\overline{\mathbf{b}}_n^\texttt{e}$, $\overline{\boldsymbol{\beta}}_n$, $\alpha_n$, $\mathbf{F}_n$, and $\mathbf{F}_{n+1}$, we \textit{temporarily} assume no new plastic flow ($\widehat{\gamma}_{n+1}=0$) and define the trial variables as
\begin{equation} \label{Trial Variables}
    \left\{ \begin{array}{llll}
        \overline{\mathbf{b}}_{n+1}^\texttt{e,tr}
        = \overline{\mathbf{f}}_{n+1} \overline{\mathbf{b}}_n^\texttt{e} \overline{\mathbf{f}}_{n+1}^\top,

        & \overline{\boldsymbol{\beta}}_{n+1}^\texttt{tr}
        = \overline{\mathbf{f}}_{n+1} \overline{\boldsymbol{\beta}}_n \overline{\mathbf{f}}_{n+1}^\top,

        & \alpha_{n+1}^\texttt{tr} = \alpha_n, \\[10pt]

        \mathbf{s}_{n+1}^\texttt{tr} = \mu\ \text{dev} \left( \overline{\mathbf{b}}_{n+1}^\texttt{e,tr} \right),

        &\boldsymbol{\xi}_{n+1}^\texttt{tr}
        = \mathbf{s}_{n+1}^\texttt{tr} - \text{dev} \left( \overline{\boldsymbol{\beta}}_{n+1}^\texttt{tr} \right),

        &\mathbf{n}_{n+1}^\texttt{tr} = \boldsymbol{\xi}_{n+1}^\texttt{tr} / \lVert \boldsymbol{\xi}_{n+1}^\texttt{tr} \rVert, \\[10pt]

        \overline{\mu}_{n+1}^\texttt{tr} = \dfrac{\mu}{3} \text{tr}\left( \overline{\mathbf{b}}_{n+1}^\texttt{e,tr} \right),

        & \overline{\overline{\mu}}_{n+1}^\texttt{tr} = \overline{\mu}_{n+1}^\texttt{tr} - \dfrac{1}{3} \text{tr} \left( \overline{\boldsymbol{\beta}}_{n+1}^\texttt{tr} \right),

        & f_{n+1}^\texttt{tr} = \lVert \boldsymbol{\xi}_{n+1}^\texttt{tr} \rVert
        - \sqrt{\dfrac{2}{3}} k \left( \alpha_{n+1}^\texttt{tr} \right).
    \end{array} \right.
\end{equation}
Next, we correct $\overline{\boldsymbol{\beta}}_{n+1}$ and $\alpha_{n+1}$ with
\begin{equation} \label{Updated beta and alpha}
    \overline{\boldsymbol{\beta}}_{n+1} = \overline{\boldsymbol{\beta}}_{n+1}^\texttt{tr}
    + \dfrac{2 h \overline{\overline{\mu}}_{n+1}^\texttt{tr}}{3 \mu} \widehat{\gamma}_{n+1} \mathbf{n}_{n+1}
    \quad \text{and} \quad
    \alpha_{n+1} = \alpha_{n+1}^\texttt{tr} + \sqrt{\dfrac{2}{3}} \widehat{\gamma}_{n+1}
\end{equation}
based on $\eqref{Local Governing Equations}_{2,3}$ and Proposition \ref{Compute mu_bar and mu_bar_bar} and correct $\overline{\mathbf{b}}_{n+1}^\texttt{e}$ with
\begin{equation} \label{Updated be_bar}
    \left\{ \begin{array}{l}
        \overline{\mathbf{b}}_{n+1}^\texttt{e}
        = \text{dev} \left( \overline{\mathbf{b}}_{n+1}^\texttt{e} \right) + \dfrac{1}{3} \mathcal{I}_1 \mathbf{I}, \\[12pt]

        \text{dev} \left( \overline{\mathbf{b}}_{n+1}^\texttt{e} \right) = \text{dev} \left( \overline{\mathbf{b}}_{n+1}^\texttt{e,tr} \right) - \dfrac{2 \overline{\overline{\mu}}_{n+1}^\texttt{tr}}{\mu} \widehat{\gamma}_{n+1} \mathbf{n}_{n+1}, \\[12pt]

        \mathcal{I}_1 = \left\{ \begin{array}{ll}
            \left( -\dfrac{\mathcal{Q}}{2} + \sqrt{-\Delta} \right)^{1/3} + \left( -\dfrac{\mathcal{Q}}{2} - \sqrt{-\Delta} \right)^{1/3}, & \text{if $\Delta < 0$;} \\[12pt]
    
            \max \left\{ \dfrac{3 \mathcal{Q}}{\mathcal{P}}, - \dfrac{3 \mathcal{Q}}{2 \mathcal{P}} \right\}, & \text{if $\Delta = 0$;} \\[12pt]  
    
            \mathcal{S}_1, & \text{if $\Delta > 0$, $\left( \overline{\mathbf{b}}_{n+1}^\texttt{e} \right)_{11} > 0$,
            and $\mathcal{B} > 0$}; \\[12pt]
    
            \mathcal{S}_2, & \text{otherwise.}
        \end{array} \right.
    \end{array} \right.
\end{equation}
Here, the related variables are defined as
\begin{equation*}
    \left\{ \begin{array}{l}
        \mathcal{P} = - \mathcal{J}_2 \leq 0, \quad
        \mathcal{Q} = \mathcal{J}_3 - 1, \quad
        \Delta = - \left( \dfrac{\mathcal{P}^3}{27} + \dfrac{\mathcal{Q}^2}{4} \right), \\[12pt]

        \mathcal{J}_2 = \dfrac{1}{2} \left\lVert \text{dev} \left( \overline{\mathbf{b}}_{n+1}^\texttt{e} \right) \right\rVert^2, \quad 
        \mathcal{J}_3 = \det \left[ \text{dev} \left( \overline{\mathbf{b}}_{n+1}^\texttt{e} \right) \right], \\[12pt]
   
        \mathcal{B} = \left( \overline{\mathbf{b}}_{n+1}^\texttt{e} \right)_{11} \left( \overline{\mathbf{b}}_{n+1}^\texttt{e} \right)_{22} - \left( \overline{\mathbf{b}}_{n+1}^\texttt{e} \right)_{12} \left( \overline{\mathbf{b}}_{n+1}^\texttt{e} \right)_{21}.
    \end{array} \right.
\end{equation*}
Additionally, we define $\mathcal{S}_1$ and $\mathcal{S}_2$ as the largest and second to largest values among $r_1$, $r_2$, and $r_3$ with $r_k$ defined as
\begin{equation} \label{Trigonometric Solutions}
    r_k = 2 \sqrt{-\dfrac{\mathcal{P}}{3}} \cos \left[ \dfrac{1}{3} \arccos \left( \dfrac{3 \mathcal{Q}}{2 \mathcal{P}} \sqrt{-\dfrac{3}{\mathcal{P}}} \right) - \dfrac{2 k \pi}{3} \right]
    \quad \text{for} \quad
    k = 1, 2, 3.
\end{equation}

\begin{proposition} \label{Compute mu_bar and mu_bar_bar}
The variables $\overline{\mu}$ and $\overline{\overline{\mu}}$ satisfy $\overline{\mu}_{n+1} = \overline{\mu}_{n+1}^\texttt{tr}$ and $\overline{\overline{\mu}}_{n+1} = \overline{\overline{\mu}}_{n+1}^\texttt{tr}$, respectively.
\end{proposition}

\begin{proof}
Applying the backward Euler difference on $\eqref{Flow Rules}_2$ derives
\begin{equation*}
    0 = \text{tr} \left[ \mathbf{F}_{n+1} \left( \mathbf{F}_{n+1}^{-1} \mathbf{b}_{n+1}^\texttt{e} \mathbf{F}_{n+1}^{-\top} - \mathbf{F}_n^{-1} \mathbf{b}_n^\texttt{e} \mathbf{F}_n^{-\top} \right) \mathbf{F}_{n+1}^\top \right]
    = \text{tr} \left( \mathbf{b}_{n+1}^\texttt{e} \right) - \text{tr} \left( \mathbf{f}_{n+1} \mathbf{b}_n^\texttt{e} \mathbf{f}_{n+1}^\top \right).
\end{equation*}
Note that
\begin{equation*}
    J_{n+1}^{-2/3} \mathbf{b}_{n+1}^\texttt{e} = (J_{n+1}^\texttt{e})^{-2/3} \mathbf{b}_{n+1}^\texttt{e} = \overline{\mathbf{b}}_{n+1}^\texttt{e}
\end{equation*}
and
\begin{equation*}
    J_{n+1}^{-2/3} \mathbf{f}_{n+1} \mathbf{b}_n^\texttt{e} \mathbf{f}_{n+1}^\top = \left( J_{n+1}^f J_n^\texttt{e} \right)^{-2/3} \mathbf{f}_{n+1} \mathbf{b}_n^\texttt{e} \mathbf{f}_{n+1}^\top = \overline{\mathbf{f}}_{n+1} \overline{\mathbf{b}}_n^\texttt{e} \overline{\mathbf{f}}_{n+1}^\top = \overline{\mathbf{b}}_{n+1}^\texttt{e,tr}.
\end{equation*}
We then prove
\begin{equation*}
    \text{tr} \left( \overline{\mathbf{b}}_{n+1}^\texttt{e} \right)
    = \text{tr} \left( \overline{\mathbf{b}}_{n+1}^\texttt{e,tr} \right)
    \quad \Rightarrow \quad
    \overline{\mu}_{n+1} = \dfrac{\mu}{3} \text{tr} \left( \overline{\mathbf{b}}_{n+1}^\texttt{e} \right) = \dfrac{\mu}{3} \text{tr} \left( \overline{\mathbf{b}}_{n+1}^\texttt{e,tr} \right) = \overline{\mu}_{n+1}^\texttt{tr}.
\end{equation*}
Similarly, applying the backward Euler difference on $\eqref{Kinematic Hardening Laws}_2$ derives
\begin{equation*}
    \text{tr} \left( \boldsymbol{\beta}_{n+1} \right) = \text{tr} \left( \mathbf{f}_{n+1} \boldsymbol{\beta}_n \mathbf{f}_{n+1}^\top \right)
    \quad \Rightarrow \quad
    \text{tr} \left( \overline{\boldsymbol{\beta}}_{n+1} \right) = \text{tr} \left( \overline{\boldsymbol{\beta}}_{n+1}^\texttt{tr} \right),
\end{equation*}
and we prove
\begin{equation*}
    \overline{\overline{\mu}}_{n+1} = \overline{\mu}_{n+1} - \dfrac{1}{3} \text{tr} \left( \overline{\boldsymbol{\beta}}_{n+1} \right) = \overline{\mu}_{n+1}^\texttt{tr} - \dfrac{1}{3} \text{tr} \left( \overline{\boldsymbol{\beta}}_{n+1}^\texttt{tr} \right) = \overline{\overline{\mu}}_{n+1}^\texttt{tr}.
\end{equation*}
\end{proof}

\begin{remark}
The updating formulae for $\overline{\mathbf{b}}_{n+1}^\texttt{e}$ in \eqref{Updated be_bar} is different from $\eqref{Local Governing Equations}_1$ as used in \citet{simo_framework_1988-1, simo_framework_1988}. It is because the expression in $\eqref{Local Governing Equations}_1$ is a necessary but \textit{insufficient} condition of the isochoric plastic flow ($J_{n+1}^\texttt{p} = 1$). As pointed out in \citet{wang_how_2017} and demonstrated in \ref{Sec: Comparison with Analytical Solution}, the violation of the isochoric plastic flow leads to incorrect stress prediction, especially during the unloading stage. Many remedies such as the exponential map algorithm \citep{simo_algorithms_1992, miehe_exponential_1996} can fix this issue. In the current work, we propose the updating formulae in \eqref{Updated be_bar} that require minimum modifications to the original theory \citep{simo_framework_1988-1, simo_framework_1988}, and this formula is essentially a generalized version compared to the one used in \citet{simo_associative_1992}. The derivation of \eqref{Updated be_bar} is presented in \ref{Sec: Derivation of Updating Formula of be_bar}. A comparison between the resultant FEA prediction and the analytical solution is shown in \ref{Sec: FEA Verification}, where we observe that the isochoric plastic flow is enforced precisely.
\end{remark}

We emphasize that the updating formulae in \eqref{Updated beta and alpha} and \eqref{Updated be_bar} are fully explicit up to the unit tensor ($\mathbf{n}_{n+1}$) and consistency parameter ($\widehat{\gamma}_{n+1}$), and we proceed to determine them. To derive the unit tensor ($\mathbf{n}_{n+1}$), we utilize $\eqref{Local Governing Equations}_4$, $\eqref{Trial Variables}_4$, and $\eqref{Updated be_bar}_2$ and yield
\begin{equation} \label{Deviatoric Kirchhoff Stress}
    \mathbf{s}_{n+1} = \mu\ \text{dev}(\overline{\mathbf{b}}_{n+1}^\texttt{e,tr})
    - 2 \overline{\overline{\mu}}_{n+1}^\texttt{tr} \widehat{\gamma}_{n+1} \mathbf{n}_{n+1}
    = \mathbf{s}_{n+1}^\texttt{tr}
    - 2 \overline{\overline{\mu}}_{n+1}^\texttt{tr} \widehat{\gamma}_{n+1} \mathbf{n}_{n+1}
\end{equation}
by noticing $\mathbf{n}_{n+1}$ is deviatoric. Based on $\eqref{Local Governing Equations}_2$, $\eqref{Trial Variables}_{2,5}$, and \eqref{Deviatoric Kirchhoff Stress}, we rewrite $\eqref{Local Governing Equations}_5$ as
\begin{equation*}
    \boldsymbol{\xi}_{n+1} = \mathbf{s}_{n+1} - \text{dev} \left( \overline{\boldsymbol{\beta}}_{n+1} \right)
    = \boldsymbol{\xi}_{n+1}^\texttt{tr}
    - 2 \overline{\overline{\mu}}_{n+1}^\texttt{tr}
    \left( 1 + \dfrac{h}{3\mu} \right) \widehat{\gamma}_{n+1} \mathbf{n}_{n+1},
\end{equation*}
which further renders
\begin{equation*}
    \lVert \boldsymbol{\xi}_{n+1} \rVert \mathbf{n}_{n+1}
    = \lVert \boldsymbol{\xi}_{n+1}^\texttt{tr} \rVert \mathbf{n}_{n+1}^\texttt{tr}
    - 2 \overline{\overline{\mu}}_{n+1}^\texttt{tr}
    \left( 1 + \dfrac{h}{3\mu} \right) \widehat{\gamma}_{n+1} \mathbf{n}_{n+1}
\end{equation*}
and implies
\begin{equation} \label{Relationship between xi and xi_trial}
    \lVert \boldsymbol{\xi}_{n+1} \rVert
    + 2 \overline{\overline{\mu}}_{n+1}^\texttt{tr}
    \left( 1 + \dfrac{h}{3\mu} \right) \widehat{\gamma}_{n+1}
    = \lVert \boldsymbol{\xi}_{n+1}^\texttt{tr} \rVert
    \quad \text{and} \quad
    \mathbf{n}_{n+1} = \mathbf{n}_{n+1}^\texttt{tr}
\end{equation}
under the assumption of $\overline{\overline{\mu}}_{n+1}^\texttt{tr} > 0$ and due to $\lVert \mathbf{n}_{n+1} \rVert = \lVert \mathbf{n}_{n+1}^\texttt{tr} \rVert = 1$.

Finally, to compute the consistency parameter ($\widehat{\gamma}_{n+1}$), we discuss the signs of the trial yield indicator, $f_{n+1}^\texttt{tr}$. In the case of $f_{n+1}^\texttt{tr} \leq 0$, the yield criterion is satisfied ($f_{n+1} < 0$ and $\widehat{\gamma}_{n+1} = 0$) with $\overline{\mathbf{b}}_{n+1}^\texttt{e,tr}$, $\overline{\boldsymbol{\beta}}_{n+1}^\texttt{tr}$, and $\alpha_{n+1}^\texttt{tr}$. We then use them to update the internal variables at load step $n+1$ as
\begin{equation*}
    \overline{\mathbf{b}}_{n+1}^\texttt{e} = \overline{\mathbf{b}}_{n+1}^\texttt{e,tr}, \quad
    \overline{\boldsymbol{\beta}}_{n+1} = \overline{\boldsymbol{\beta}}_{n+1}^\texttt{tr},
    \quad \text{and} \quad
    \alpha_{n+1} = \alpha_{n+1}^\texttt{tr}.
\end{equation*}
In the case of $f_{n+1}^\texttt{tr} > 0$, the yield criterion is violated ($f_{n+1} > 0$) with $\overline{\mathbf{b}}_{n+1}^\texttt{e,tr}$, $\overline{\boldsymbol{\beta}}_{n+1}^\texttt{tr}$, and $\alpha_{n+1}^\texttt{tr}$. We need to make them ``return" along the $\mathbf{n}_{n+1}$ direction with a distance determined by $\widehat{\gamma}_{n+1} > 0$. To figure out $\widehat{\gamma}_{n+1}$, we enforce the yield criterion
\begin{equation} \label{Function of gamma}
    0 = f_{n+1} = \lVert \boldsymbol{\xi}_{n+1} \rVert - \sqrt{\dfrac{2}{3}} k (\alpha_{n+1})
    = \lVert \boldsymbol{\xi}_{n+1}^\texttt{tr} \rVert
    - 2 \overline{\overline{\mu}}_{n+1}^\texttt{tr}
    \left( 1 + \dfrac{h}{3\mu} \right) \widehat{\gamma}_{n+1}
    - \sqrt{\dfrac{2}{3}} k \left( \alpha_{n+1}^\texttt{tr} + \sqrt{\dfrac{2}{3}} \widehat{\gamma}_{n+1} \right) := \mathcal{G}(\widehat{\gamma}_{n+1})
\end{equation}
based on $\eqref{Updated beta and alpha}_2$ and $\eqref{Relationship between xi and xi_trial}_1$. The expression in \eqref{Function of gamma} is an algebraic equation of $\widehat{\gamma}_{n+1}$. For the linear isotropic hardening law in \eqref{Linear Isotropic Hardening Law}, we analytically solve \eqref{Function of gamma} for $\widehat{\gamma}_{n+1}$ as
\begin{equation} \label{Linear Solution of gamma}
    \widehat{\gamma}_{n+1} =
    \dfrac{1}{2 \overline{\overline{\mu}}_{n+1}^\texttt{tr}}
    \left( 1 + \dfrac{h}{3\mu} + \dfrac{K}{3 \overline{\overline{\mu}}_{n+1}^\texttt{tr}} \right)^{-1}
    \left[ \lVert \boldsymbol{\xi}_{n+1}^\texttt{tr} \rVert - \sqrt{\dfrac{2}{3}} (\sigma_y + K \alpha_{n+1}^\texttt{tr}) \right].
\end{equation}
For the nonlinear hardening law in \eqref{Nonlinear Isotropic Hardening Law}, we adopt the Newton's method to iteratively solve $\widehat{\gamma}_{n+1}$ as
\begin{equation} \label{Nonlinear Solution of gamma}
    \widehat{\gamma}_{n+1}^{(k+1)} = \widehat{\gamma}_{n+1}^{(k)}
    - \mathcal{G}(\widehat{\gamma}_{n+1}^{(k)}) / \mathcal{G}'(\widehat{\gamma}_{n+1}^{(k)})
\end{equation}
with
\begin{equation*}
    \mathcal{G}'(\widehat{\gamma}_{n+1}^{(k)})
    = - 2 \overline{\overline{\mu}}_{n+1}^\texttt{tr} \left( 1 + \dfrac{h}{3\mu} \right)
    - \dfrac{2}{3} k' \left( \alpha_{n+1}^\texttt{tr} + \sqrt{\dfrac{2}{3}} \widehat{\gamma}_{n+1}^{(k)} \right).
\end{equation*}

After computing the unit tensor ($\mathbf{n}_{n+1}$) with $\eqref{Relationship between xi and xi_trial}_2$ and the consistency parameter ($\widehat{\gamma}_{n+1}$) with \eqref{Linear Solution of gamma} or \eqref{Nonlinear Solution of gamma}, we send them back to \eqref{Updated beta and alpha} and \eqref{Updated be_bar} for correcting $\overline{\mathbf{b}}_{n+1}^\texttt{e}$, $\overline{\boldsymbol{\beta}}_{n+1}$, and $\alpha_{n+1}$. This entire prediction--correction procedure forms a closed loop and is typically referred to as the return mapping scheme for elastoplasticity. For the reader's convenience, we present the implementation of this scheme in Algorithm \ref{Algorithm of Return Mapping Scheme} and the relationships between the primary state variables at load steps $n$ and $n+1$ in Fig. \ref{Fig: Independent Variables}.

\begin{algorithm}[!htbp]
\caption{Return mapping scheme for finite strain elastoplasticity}
\label{Algorithm of Return Mapping Scheme}

\textbf{Inputs}:

Volume-preserving part of the elastic left Cauchy--Green deformation tensor, $\overline{\mathbf{b}}_n^\texttt{e}$;

``Volume-preserving" part of the back stress tensor, $\overline{\boldsymbol{\beta}}_n$;

Equivalent plastic strain, $\alpha_n$;

Deformation gradients at load steps $n$ and $n+1$, $\mathbf{F}_n$ and $\mathbf{F}_{n+1}$, respectively;

\textbf{Prediction stage:}

Compute the relative deformation gradient as $\mathbf{f}_{n+1} = \mathbf{F}_{n+1} \mathbf{F}_n^{-1}$;

Compute the volume-preserving part of $\mathbf{f}_{n+1}$ as $\overline{\mathbf{f}}_{n+1} = {\left( J_{n+1}^f \right)}^{-1/3} \mathbf{f}_{n+1}$;

Compute the trial version of $\overline{\mathbf{b}}_{n+1}^\texttt{e}$ as 
$\overline{\mathbf{b}}_{n+1}^\texttt{e,tr} = \overline{\mathbf{f}}_{n+1} \overline{\mathbf{b}}_n^\texttt{e} \overline{\mathbf{f}}_{n+1}^\top$;

Compute the trial version of $\overline{\boldsymbol{\beta}}_{n+1}$ as $\overline{\boldsymbol{\beta}}_{n+1}^\texttt{tr} = \overline{\mathbf{f}}_{n+1} \overline{\boldsymbol{\beta}}_n \overline{\mathbf{f}}_{n+1}^\top$;

Compute the trial version of $\alpha_{n+1}$ as $\alpha_{n+1}^\texttt{tr} = \alpha_n$;

Compute the trial deviatoric Kirchhoff stress tensor as $\mathbf{s}_{n+1}^\texttt{tr} = \mu\ \text{dev} \left( \overline{\mathbf{b}}_{n+1}^\texttt{e,tr} \right)$;

Compute the trial net stress tensor as $\boldsymbol{\xi}_{n+1}^\texttt{tr} = \mathbf{s}_{n+1}^\texttt{tr} - \text{dev} \left( \overline{\boldsymbol{\beta}}_{n+1}^\texttt{tr} \right)$;

Compute the unit tensor as $\mathbf{n}_{n+1} = \mathbf{n}_{n+1}^\texttt{tr} = \boldsymbol{\xi}_{n+1}^\texttt{tr} / \lVert \boldsymbol{\xi}_{n+1}^\texttt{tr} \rVert$;

Compute the parameters $\overline{\mu}_{n+1}^\texttt{tr} = \mu\ \text{tr}\left( \overline{\mathbf{b}}_{n+1}^\texttt{e,tr} \right) / 3$ and $\overline{\overline{\mu}}_{n+1}^\texttt{tr} = \overline{\mu}_{n+1}^\texttt{tr} - \text{tr} \left( \overline{\boldsymbol{\beta}}_{n+1}^\texttt{tr} \right) / 3$;

Compute the trial yield indicator as $f_{n+1}^\texttt{tr} = \lVert \boldsymbol{\xi}_{n+1}^\texttt{tr} \rVert - \sqrt{2/3} k \left( \alpha_{n+1}^\texttt{tr} \right)$;

\textbf{Correction stage:}

\eIf{$f_{n+1}^\texttt{tr} \leq 0$}{
    Set the consistency parameter as $\widehat{\gamma}_{n+1} = 0$;
}{
    \eIf{the isotropic hardening law ($k(\alpha)$) is linear}{
        Compute $\widehat{\gamma}_{n+1} > 0$ with \eqref{Linear Solution of gamma};
    }{
        Compute $\widehat{\gamma}_{n+1} > 0$ with \eqref{Nonlinear Solution of gamma};
    }
}

Compute $\mathcal{I}_1$ based on $\eqref{Updated be_bar}_3$;

Update $\overline{\mathbf{b}}_{n+1}^\texttt{e}$ as $\overline{\mathbf{b}}_{n+1}^\texttt{e}= \text{dev} \left( \overline{\mathbf{b}}_{n+1}^\texttt{e,tr} \right) - 2 \left( \overline{\overline{\mu}}_{n+1}^\texttt{tr}/\mu \right) \widehat{\gamma}_{n+1} \mathbf{n}_{n+1} + \mathcal{I}_1 \mathbf{I}/3$;

Update $\overline{\boldsymbol{\beta}}_{n+1}$ as $\overline{\boldsymbol{\beta}}_{n+1} = \overline{\boldsymbol{\beta}}_{n+1}^\texttt{tr} + 2h/3 \left( \overline{\overline{\mu}}_{n+1}^\texttt{tr}/\mu \right) \widehat{\gamma}_{n+1} \mathbf{n}_{n+1}$;

Update $\alpha_{n+1}$ as $\alpha_{n+1} = \alpha_{n+1}^\texttt{tr} + \sqrt{2/3} \widehat{\gamma}_{n+1}$;
\end{algorithm}

\begin{figure}[!htbp]
    \centering
    \includegraphics[width=18cm]{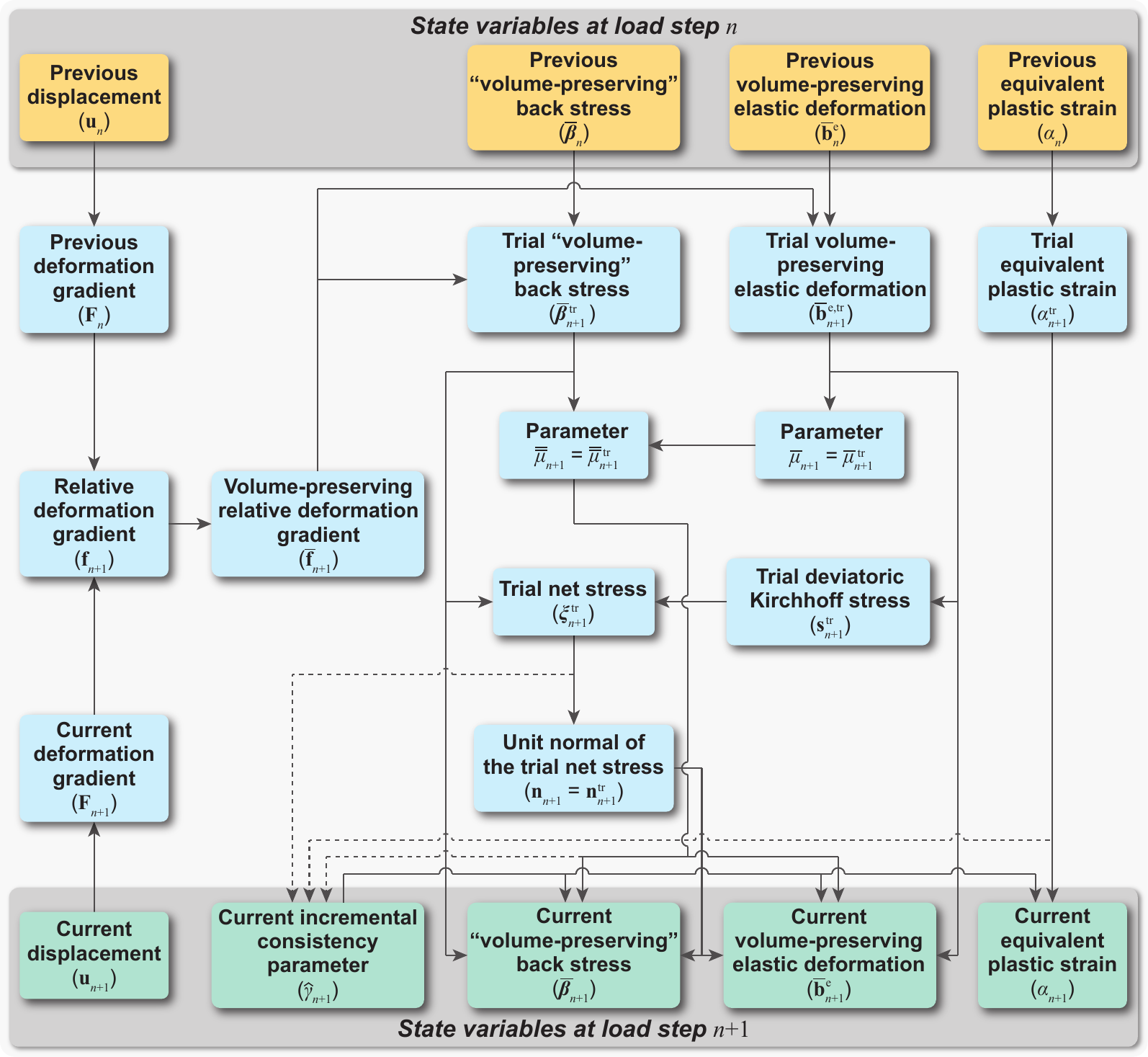}
    \caption{Relationships between the primary state variables at load steps $n$ and $n+1$. The dashed arrows signify that $\widehat{\gamma}_{n+1}$ generally does not admit to an explicit expression and needs to be numerically solved from \eqref{Function of gamma} with \eqref{Nonlinear Solution of gamma}.}
    \label{Fig: Independent Variables}
\end{figure}

\subsection{Global governing equations for finite strain elastoplasticity}

We remark again that the return mapping scheme in Section \ref{Return mapping scheme} pertains to the elastoplastic response of a single material point ($\mathbf{X} \in \Omega_0$). To predict the elastoplastic response of the entire structure defined on $\Omega_0$, we derive the global governing equations as follows.

\subsubsection{Stress measures and algorithmic tangent moduli}

Substituting \eqref{Deviatoric Kirchhoff Stress} into \eqref{Kirchhoff Stress} reformulates the Kirchhoff stress tensor ($\boldsymbol{\tau}_{n+1}$) at load step $n+1$ as
\begin{equation*}
    \boldsymbol{\tau}_{n+1} = \boldsymbol{\tau}_{n+1}^\texttt{vol}
    + \mathbf{s}_{n+1}^\texttt{tr}
    - 2 \overline{\overline{\mu}}_{n+1}^\texttt{tr} \widehat{\gamma}_{n+1} \mathbf{n}_{n+1}.
\end{equation*}
Accordingly, we compute the first Piola--Kirchhoff stress tensor ($\mathbf{P}_{n+1}$) as
\begin{equation} \label{First PK Stress}
    \mathbf{P}_{n+1} = \boldsymbol{\tau}_{n+1} \mathbf{F}_{n+1}^{-\top} = \left( \boldsymbol{\tau}_{n+1}^\texttt{vol}
    + \mathbf{s}_{n+1}^\texttt{tr}
    - 2 \overline{\overline{\mu}}_{n+1}^\texttt{tr} \widehat{\gamma}_{n+1} \mathbf{n}_{n+1} \right) \mathbf{F}_{n+1}^{-\top}
\end{equation}
and the second Piola--Kirchhoff stress tensor ($\mathbf{S}_{n+1}$) as
\begin{equation} \label{Second PK Stress}
    \mathbf{S}_{n+1} = \varphi^*(\boldsymbol{\tau}_{n+1})
    = \varphi^*(\boldsymbol{\tau}_{n+1}^\texttt{vol})
    + \varphi^*(\mathbf{s}_{n+1}^\texttt{tr})
    - 2 \overline{\overline{\mu}}_{n+1}^\texttt{tr} \widehat{\gamma}_{n+1} \varphi^*(\mathbf{n}_{n+1}).
\end{equation}

Based on the second Piola--Kirchhoff stress tensor ($\mathbf{S}_{n+1}$), we compute the second \textit{algorithmic} tangent moduli as
\begin{equation} \label{Second Algorithmic Tangent Moduli Initial}
    \mathbb{C}_{n+1} = \dfrac{\partial \mathbf{S}_{n+1}}{\partial \mathbf{E}_{n+1}} = \left\{ \begin{array}{cc}
        \mathbb{C}_{n+1}^\texttt{el}, & \text{if $f_{n+1}^\texttt{tr} \leq 0$ ($\widehat{\gamma}_{n+1} = 0$; hyperelastic)}; \\[12pt]
        \mathbb{C}_{n+1}^\texttt{ep}, & \text{if $f_{n+1}^\texttt{tr} > 0$ ($\widehat{\gamma}_{n+1} > 0$; elastoplastic)}.
    \end{array} \right.
\end{equation}
Here, the variable $\mathbb{C}_{n+1}^\texttt{el}$ is the second elastic moduli given as
\begin{equation} \label{Second Elastic Moduli}
    \begin{array}{ll}
        \mathbb{C}_{n+1}^\texttt{el} = \dfrac{\partial \varphi^*(\widehat{\boldsymbol{\tau}}_{n+1})}{\partial \mathbf{E}_{n+1}}
        = &\left( 2 \overline{\mu}_{n+1} - 2 J_{n+1} U'_{n+1} \right) \mathbb{I}_{\mathbf{C}_{n+1}^{-1}}
        + \left[ J_{n+1} (J_{n+1} U'_{n+1})'- \dfrac{2 \overline{\mu}_{n+1}}{3} \right] \mathbf{C}_{n+1}^{-1} \otimes \mathbf{C}_{n+1}^{-1} \\[12pt]
        &- \dfrac{2}{3} \left[ \varphi^* \left( \widehat{\mathbf{s}}_{n+1} \right) \otimes \mathbf{C}_{n+1}^{-1}
        + \mathbf{C}_{n+1}^{-1} \otimes \varphi^* \left( \widehat{\mathbf{s}}_{n+1} \right) \right]
    \end{array}
\end{equation}
where we define
\begin{equation*}
    \left( \mathbb{I}_{\mathbf{C}_{n+1}^{-1}} \right)_{ijkl} = \dfrac{1}{2} \left[ \left( C_{n+1}^{-1} \right)_{ik} \left( C_{n+1}^{-1} \right)_{jl} + \left( C_{n+1}^{-1} \right)_{il} \left( C_{n+1}^{-1} \right)_{jk} \right].
\end{equation*}
The variable $\mathbb{C}_{n+1}^\texttt{ep}$ is the second elastoplastic moduli (see \ref{Sec: Second Elastoplastic Moduli} for derivation) expressed as
\begin{equation} \label{Second Elastoplastic Moduli Final}
    \begin{array}{ll}
        \mathbb{C}_{n+1}^\texttt{ep} = &\left( 2 \overline{\mu}_{n+1}^\texttt{tr} - 2 c_1 \overline{\overline{\mu}}_{n+1}^\texttt{tr} - 2 J_{n+1} U'_{n+1} \right) \mathbb{I}_{\mathbf{C}_{n+1}^{-1}}
        
        + \left[ J_{n+1} (J_{n+1} U'_{n+1})' - \dfrac{2 \overline{\mu}_{n+1}^\texttt{tr}}{3} + \dfrac{2 c_1 \overline{\overline{\mu}}_{n+1}^\texttt{tr}}{3} \right] \mathbf{C}_{n+1}^{-1} \otimes \mathbf{C}_{n+1}^{-1} \\[12pt]
        
        &- \dfrac{2}{3} \left[ \varphi^*(\mathbf{s}_{n+1}^\texttt{tr}) \otimes \mathbf{C}_{n+1}^{-1} + \mathbf{C}_{n+1}^{-1} \otimes \varphi^*(\mathbf{s}_{n+1}^\texttt{tr}) \right]

        + \dfrac{2 c_1}{3} \left[ \varphi^*(\boldsymbol{\xi}_{n+1}^\texttt{tr}) \otimes \mathbf{C}_{n+1}^{-1} + \mathbf{C}_{n+1}^{-1} \otimes \varphi^*(\boldsymbol{\xi}_{n+1}^\texttt{tr}) \right] \\[12pt]

        &- c_3 \varphi^*(\mathbf{n}_{n+1}) \otimes \varphi^*(\mathbf{n}_{n+1})

        - \dfrac{c_4}{2} \left\{ \varphi^*(\mathbf{n}_{n+1}) \otimes \varphi^*\left[ \text{dev} (\mathbf{n}_{n+1}^2) \right] + \varphi^*\left[ \text{dev} (\mathbf{n}_{n+1}^2) \right] \otimes \varphi^*(\mathbf{n}_{n+1}) \right\}
    \end{array}
\end{equation}
where we define
\begin{equation} \label{Parameters in Elasticity}
    \left\{ \begin{array}{l}
        c_0 =  1 + \dfrac{h}{3\mu} + \dfrac{k'}{3 \overline{\overline{\mu}}_{n+1}^\texttt{tr}}, \quad
        c_1 = \dfrac{2 \overline{\overline{\mu}}_{n+1}^\texttt{tr} \widehat{\gamma}_{n+1}}{\lVert \overline{\boldsymbol{\xi}}_{n+1}^\texttt{tr} \rVert}, \quad
        c_2 = \dfrac{1}{c_0} - c_1, \quad \\[12pt]
        
        c_3 = 2 c_2 \overline{\overline{\mu}}_{n+1}^\texttt{tr}
        - \left[ \dfrac{1}{c_0} \left( 1 + \dfrac{h}{3 \mu} \right) - 1 \right] \dfrac{4}{3} \widehat{\gamma}_{n+1} \lVert \overline{\boldsymbol{\xi}}_{n+1}^\texttt{tr} \rVert, \quad
        c_4 = 2 c_2 \lVert \overline{\boldsymbol{\xi}}_{n+1}^\texttt{tr} \rVert.
    \end{array} \right.
\end{equation}

By examining the second elastic moduli ($\mathbb{C}_{n+1}^\texttt{el}$) in \eqref{Second Elastic Moduli} and elastoplastic moduli ($\mathbb{C}_{n+1}^\texttt{ep}$) in \eqref{Second Elastoplastic Moduli Final}, we reformulate the second algorithmic tangent moduli ($\mathbb{C}_{n+1}$) in \eqref{Second Algorithmic Tangent Moduli Initial} into a compact form as
\begin{equation} \label{Second Algorithmic Tangent Moduli Final}
    \begin{array}{ll}
        \mathbb{C}_{n+1} = &\left( 2 \overline{\mu}_{n+1}^\texttt{tr} - 2 \theta c_1 \overline{\overline{\mu}}_{n+1}^\texttt{tr} - 2 J_{n+1} U'_{n+1} \right) \mathbb{I}_{\mathbf{C}_{n+1}^{-1}}
        
        + \left[ J_{n+1} (J_{n+1} U'_{n+1})' - \dfrac{2 \overline{\mu}_{n+1}^\texttt{tr}}{3} + \dfrac{2 \theta c_1 \overline{\overline{\mu}}_{n+1}^\texttt{tr}}{3} \right] \mathbf{C}_{n+1}^{-1} \otimes \mathbf{C}_{n+1}^{-1} \\[12pt]
        
        &- \dfrac{2}{3} \left[ \varphi^*(\mathbf{s}_{n+1}^\texttt{tr}) \otimes \mathbf{C}_{n+1}^{-1} + \mathbf{C}_{n+1}^{-1} \otimes \varphi^*(\mathbf{s}_{n+1}^\texttt{tr}) \right]

        + \dfrac{2 \theta c_1}{3} \left[ \varphi^*(\boldsymbol{\xi}_{n+1}^\texttt{tr}) \otimes \mathbf{C}_{n+1}^{-1} + \mathbf{C}_{n+1}^{-1} \otimes \varphi^*(\boldsymbol{\xi}_{n+1}^\texttt{tr}) \right] \\[12pt]

        &- \theta c_3 \varphi^*(\mathbf{n}_{n+1}) \otimes \varphi^*(\mathbf{n}_{n+1})

        - \dfrac{\theta c_4}{2} \left\{ \varphi^*(\mathbf{n}_{n+1}) \otimes \varphi^*\left[ \text{dev} (\mathbf{n}_{n+1}^2) \right] + \varphi^*\left[ \text{dev} (\mathbf{n}_{n+1}^2) \right] \otimes \varphi^*(\mathbf{n}_{n+1}) \right\}
    \end{array}
\end{equation}
with
\begin{equation} \label{Theta Parameter}
    \theta = \dfrac{1}{2} \dfrac{f_{n+1}^\texttt{tr} + |f_{n+1}^\texttt{tr}|}{f_{n+1}^\texttt{tr}} = \left\{ \begin{array}{ll}
        0, & \text{if $f_{n+1}^\texttt{tr} \leq 0$ (hyperelastic)}; \\[8pt]
        1, & \text{if $f_{n+1}^\texttt{tr} > 0$ (elastoplastic)}.
    \end{array} \right.
\end{equation}

\subsubsection{Governing equations}

Based on the stress measures and tangent moduli defined above, we present the global governing equations in an incremental form as
\begin{equation} \label{Global Governing Equation}
    \left\{ \begin{array}{l}
        \displaystyle \int_{\Omega_0} \nabla \mathbf{v} : \widehat{\mathbb{C}}_{n+1}^{(k)} : \nabla \delta \mathbf{u}_{n+1}^{(k)}\ \text{d} \mathbf{X}
        = \int_{\Omega_0} \overline{\mathbf{q}}_{n+1} \cdot \mathbf{v}\ \text{d} \mathbf{X} 
        + \int_{\partial \Omega_0^\mathcal{N}} \overline{\mathbf{t}}_{n+1} \cdot \mathbf{v}\ \text{d} \mathbf{X}
        - \int_{\Omega_0} \mathbf{P}_{n+1}^{(k)} : \nabla \mathbf{v}\ \text{d} \mathbf{X}, \\[12pt]

        \delta \mathbf{u}_{n+1}^{(k)} = \overline{\mathbf{u}}_{n+1} - \mathbf{u}_{n+1}^{(k)} \quad \text{for} \quad \mathbf{X} \in \partial \Omega_0^\mathcal{D}, \\[12pt]
        
        \mathbf{u}_{n+1}^{(k+1)} = \mathbf{u}_{n+1}^{(k)} + l_{n+1}^{(k)} \delta \mathbf{u}_{n+1}^{(k)} \quad \text{for} \quad \mathbf{X} \in \overline{\Omega}_0.
    \end{array} \right.
\end{equation}
Here, the variable $\delta \mathbf{u}_{n+1}$ is the incremental displacement field, and $\mathbf{v}$ is the test displacement field. The variables $\overline{\mathbf{q}}_{n+1}$, $\overline{\mathbf{t}}_{n+1}$, and $\overline{\mathbf{u}}_{n+1}$ are the prescribed body force, traction, and displacement, respectively. The symbol $l_{n+1} \in (0, 1]$ represents a line search parameter such that the updated displacement field reduces the absolute residual (right-hand side of $\eqref{Global Governing Equation}_1$). The superscript $(k)$ of a variable represents the value evaluated at Newton iteration $k$. The symbol $\overline{\Omega}_0$ is the closed bounded domain in its undeformed configuration with Neumann ($\partial \Omega_0^\mathcal{N}$) and Dirichlet ($\partial \Omega_0^\mathcal{D}$) boundaries. The variable $\widehat{\mathbb{C}}_{n+1}$ is the first algorithmic tangent moduli, which can be computed from $\mathbf{S}_{n+1}$ in \eqref{Second PK Stress} and $\mathbb{C}_{n+1}$ in \eqref{Second Algorithmic Tangent Moduli Final} and expressed in indicial notation as
\begin{equation*}
    ( \widehat{\mathbb{C}}_{n+1} )_{rjsl} = \left( F_{n+1} \right)_{ri} \left( F_{n+1} \right)_{sk} \left( \mathbb{C}_{n+1} \right)_{ijkl} + \delta_{rs} \left( S_{n+1} \right)_{jl}.
\end{equation*}

After some tensor algebra, we write out the integrand of the left-hand side in $\eqref{Global Governing Equation}_1$ as
\begin{equation} \label{Integrand of Left-hand Side}
    \begin{array}{ll}
        \nabla \mathbf{v} : \widehat{\mathbb{C}} : \nabla \delta \mathbf{u} = &\left( 2 \overline{\mu}^{\texttt{tr}} - 2 \theta c_1 \overline{\overline{\mu}}^{\texttt{tr}} - 2 J U' \right) \left( \mathcal{U} \right)^\texttt{sym} : \mathcal{V}
        
        + \left[ J \left( J U' \right)' - \dfrac{2 \overline{\mu}^{\texttt{tr}}}{3} + \dfrac{2 \theta c_1 \overline{\overline{\mu}}^{\texttt{tr}}}{3} \right] \text{tr} \left( \mathcal{U} \right) \text{tr} \left( \mathcal{V} \right) \\[12pt]
        
        &- \dfrac{2}{3} \left[ \text{tr} \left( \mathcal{U} \right) \mathbf{s}^{\texttt{tr}} : \mathcal{V} + \text{tr} \left( \mathcal{V} \right) \mathbf{s}^{\texttt{tr}} : \mathcal{U} \right]

        + \dfrac{2 \theta c_1}{3} \left[ \text{tr} \left( \mathcal{U} \right) \boldsymbol{\xi}^{\texttt{tr}} : \mathcal{V} + \text{tr} \left( \mathcal{V} \right) \boldsymbol{\xi}^{\texttt{tr}} : \mathcal{U} \right] \\[12pt]

        &- \theta c_3 \left( \mathbf{n} : \mathcal{V} \right) \left( \mathbf{n} : \mathcal{U} \right)

        - \dfrac{\theta c_4}{2} \left\{ \left( \mathbf{n} : \mathcal{V} \right) \left[ \text{dev} \left( \mathbf{n}^2 \right) : \mathcal{U} \right] + \left[ \text{dev} \left( \mathbf{n}^2 \right) : \mathcal{V} \right] \left( \mathbf{n} : \mathcal{U} \right) \right\}

        + \left( \nabla \delta \mathbf{u}\ \mathbf{S} \right) : \nabla \mathbf{v}
    \end{array}
\end{equation}
by defining
\begin{equation*}
    \mathcal{U} = \nabla \delta \mathbf{u}\ \mathbf{F}^{-1}
    \quad \text{and} \quad
    \mathcal{V} = \nabla \mathbf{v}\ \mathbf{F}^{-1}
\end{equation*}
where we temporarily drop the subscript $n+1$ and superscript $(k)$ for conciseness. We remark that \eqref{Global Governing Equation} are explicit, linear equations of $\delta \mathbf{u}_{n+1}^{(k)}$, where all coefficients can be predetermined by $\overline{\mathbf{b}}_n^\texttt{e}$, $\overline{\boldsymbol{\beta}}_n$, $\alpha_n$, $\mathbf{u}_n$, and $\mathbf{u}_{n+1}^{(k)}$. Consequently, we can iteratively solve $\delta \mathbf{u}_{n+1}^{(k)}$ from \eqref{Global Governing Equation} and update $\mathbf{u}_{n+1}^{(k+1)}$ until convergence.

For the reader's convenience, we outline the solution scheme in Algorithm \ref{Algorithm of Overall Solution Scheme}. To enhance the convergence of the highly non-convex elastoplasticity problem, we incorporate an adaptive line search scheme (Algorithm \ref{Algorithm of Line Search}), which iteratively reduces the line search parameter ($l_{n+1}^{(k)}$) until the residual starts to decrease. Furthermore, we employ an adaptive loading scheme (Algorithm \ref{Algorithm of Adaptive Loading}) that dynamically adjusts the applied boundary values ($\overline{\mathbf{q}}$, $\overline{\mathbf{t}}$, and $\overline{\mathbf{u}}$) in cases where Newton iterations diverge occasionally. With the combination of Algorithms \ref{Algorithm of Return Mapping Scheme}–\ref{Algorithm of Adaptive Loading}, we successfully obtain converged FEA solutions for all examples introduced in Section \ref{Sec: Sample Examples}. The convergence, precision, and computational time of FEA for a sample design case are demonstrated in \ref{Sec: FEA Convergence}.

\begin{algorithm}[!htbp]
\caption{Overall solution scheme for finite strain elastoplasticity}
\label{Algorithm of Overall Solution Scheme}

\textbf{Inputs}: Initial state variables, $\overline{\mathbf{b}}_0^\texttt{e} \leftarrow \mathbf{I}$, $\overline{\boldsymbol{\beta}}_0 \leftarrow \mathbf{0}$, $\alpha_0 \leftarrow 0$, and $\mathbf{u}_0 \leftarrow \mathbf{0}$;
Prescribed body force ($\overline{\mathbf{q}}$), traction ($\overline{\mathbf{t}}$), and displacement ($\overline{\mathbf{u}}$);
Prescribed time steps, $\{t_1, t_2, \ldots, t_{M-1}, t_M\}$;
Maximum Newton iterations, $N^\texttt{iter}$;
Absolute and relative tolerances of Newton iterations, $\varepsilon^\texttt{abs}$ and $\varepsilon^\texttt{rel}$, respectively;  

Initialize the converged time step, $t^\texttt{conv} \leftarrow 0$;

Initialize the time step count, $m \leftarrow 1$, load step count, $n \leftarrow 0$, and try count, $r \leftarrow 0$;

\While{$m \leq M$}{
    Initialize the current time step, $t^\texttt{now} \leftarrow t_m$;

    \While{$t^\texttt{conv} < t_m$}{
        Update the boundary values, $\overline{\mathbf{q}}_{n+1} \leftarrow  t^\texttt{now}\ \overline{\mathbf{q}}$, $\overline{\mathbf{t}}_{n+1} \leftarrow  t^\texttt{now}\ \overline{\mathbf{t}}$, and $\overline{\mathbf{u}}_{n+1} \leftarrow t^\texttt{now}\ \overline{\mathbf{u}}$;

        Compute the previous deformation gradient, $\mathbf{F}_n \leftarrow \mathbf{I} + \nabla \mathbf{u}_n$;
    
        Initialize the displacement field at load step $n+1$ and Newton iteration 0, $\mathbf{u}_{n+1}^{(0)} \leftarrow \mathbf{u}_n$;
            
        Compute the absolute residual (right-hand side of $\eqref{Global Governing Equation}_1$), $a_{n+1}^{(0)}$;

        Initialize the convergence flag, $\texttt{convergence} \leftarrow \texttt{false}$;

        Initialize the Newton iteration count, $k \leftarrow 0$;
        
        \While{$k < N^\texttt{iter}$}{ 
            Compute the current deformation gradient, $\mathbf{F}_{n+1}^{(k)} \leftarrow \mathbf{I} + \nabla \mathbf{u}_{n+1}^{(k)}$;
    
            Compute $\alpha_{n+1}^{\texttt{tr},(k)}$, $\mathbf{s}_{n+1}^{\texttt{tr},(k)}$, $\boldsymbol{\xi}_{n+1}^{\texttt{tr},(k)}$, $\mathbf{n}_{n+1}^{\texttt{tr},(k)}$, $\overline{\mu}_{n+1}^{\texttt{tr},(k)}$, $\overline{\overline{\mu}}_{n+1}^{\texttt{tr},(k)}$, $f_{n+1}^{\texttt{tr},(k)}$, and $\widehat{\gamma}_{n+1}^{(k)}$ by providing Algorithm \ref{Algorithm of Return Mapping Scheme} with $\overline{\mathbf{b}}_n^\texttt{e}$, $\overline{\boldsymbol{\beta}}_n$,  $\alpha_n$, $\mathbf{F}_n$, and $\mathbf{F}_{n+1}^{(k)}$;
            
            Compute $\mathbf{P}_{n+1}^{(k)}$ in \eqref{First PK Stress} and $\mathbf{S}_{n+1}^{(k)}$ in \eqref{Second PK Stress} based on $\mathbf{F}_{n+1}^{(k)}$, $\mathbf{s}_{n+1}^{\texttt{tr},(k)}$, $\mathbf{n}_{n+1}^{\texttt{tr},(k)}$, $\overline{\overline{\mu}}_{n+1}^{\texttt{tr},(k)}$, and $\widehat{\gamma}_{n+1}^{(k)}$;
            
            Compute $c_1$, $c_3$, and $c_4$ in \eqref{Parameters in Elasticity} based on $\alpha_{n+1}^{\texttt{tr},(k)}$, $\boldsymbol{\xi}_{n+1}^{\texttt{tr},(k)}$, $\overline{\overline{\mu}}_{n+1}^{\texttt{tr},(k)}$, and $\widehat{\gamma}_{n+1}^{(k)}$;
    
            Compute $\theta$ in \eqref{Theta Parameter} based on $f_{n+1}^{(k)}$;
    
            Compute $\nabla \mathbf{v} : \widehat{\mathbb{C}}_{n+1}^{(k)} : \nabla \delta \mathbf{u}_{n+1}^{(k)}$ in \eqref{Integrand of Left-hand Side} (as a function of $\delta \mathbf{u}_{n+1}^{(k)}$ and $\mathbf{v}$) based on $\mathbf{F}_{n+1}^{(k)}$, $\mathbf{s}_{n+1}^{\texttt{tr},(k)}$,  
            $\boldsymbol{\xi}_{n+1}^{\texttt{tr},(k)}$,
            $\mathbf{n}_{n+1}^{\texttt{tr},(k)}$,
            $\overline{\mu}_{n+1}^{\texttt{tr},(k)}$, $\overline{\overline{\mu}}_{n+1}^{\texttt{tr},(k)}$, $c_1$, $c_3$, $c_4$, and $\theta$;
    
            Solve $\delta \mathbf{u}_{n+1}^{(k)}$ from $\eqref{Global Governing Equation}_{1,2}$ based on $\nabla \mathbf{v} : \widehat{\mathbb{C}}_{n+1}^{(k)} : \nabla \delta \mathbf{u}_{n+1}^{(k)}$, $\mathbf{P}_{n+1}^{(k)}$, $\overline{\mathbf{q}}_{n+1}$, $\overline{\mathbf{t}}_{n+1}$, and $\overline{\mathbf{u}}_{n+1}$;
            
            Compute $\mathbf{u}_{n+1}^{(k+1)}$ and $a_{n+1}^{(k+1)}$ by providing $\mathbf{u}_{n+1}^{(k)}$, $a_{n+1}^{(k)}$, and $\delta \mathbf{u}_{n+1}^{(k)}$ to the line search scheme in Algorithm \ref{Algorithm of Line Search};
    
            Compute the relative residual, $r_{n+1}^{(k+1)} \leftarrow a_{n+1}^{(k+1)} / a_{n+1}^{(1)}$;

            Update the Newton iteration count, $k \leftarrow k+1$;

            \If{$a_{n+1}^{(k)} \leq \varepsilon^\texttt{abs}$ or $r_{n+1}^{(k)} \leq \varepsilon^\texttt{rel}$}{
                Update the convergence flag, $\texttt{convergence} \leftarrow \texttt{true}$;

                Break the Newton iterations due to convergence;
            }
        }

        Update the current state variables, $\overline{\mathbf{b}}_{n+1}^\texttt{e}$, $\overline{\boldsymbol{\beta}}_{n+1}$, $\alpha_{n+1}$, $\widehat{\gamma}_{n+1}$, and $\mathbf{u}_{n+1}$, by providing $\overline{\mathbf{b}}_n^\texttt{e}$, $\overline{\boldsymbol{\beta}}_n$, $\alpha_n$, $\mathbf{F}_n$, $\mathbf{u}_{n+1}^{(k)}$, $t^\texttt{conv}$, $t^\texttt{now}$, $t_m$, $n$, and $r$ to the adaptive loading scheme in Algorithm \ref{Algorithm of Adaptive Loading};
    }

    Update the time step count, $m \leftarrow m+1$;
}

\end{algorithm}

\begin{algorithm}[!htbp]
\caption{Adaptive line search scheme}
\label{Algorithm of Line Search}

\textbf{Inputs}: Displacement and absolute residual from previous Newton iteration, $\mathbf{u}_{n+1}^{(k)}$ and $a_{n+1}^{(k)}$, respectively;
Incremental displacement, $\delta \mathbf{u}_{n+1}^{(k)}$;
Maximum iterations and tolerance for line search, $N^\texttt{search}$ and $\varepsilon^\texttt{search} \in (0, 1]$, respectively;

Initialize the line search parameter, $l_{n+1}^{(k)} \leftarrow 1$;

Initialize the iteration count for linear search, $s \leftarrow 1$;

\For{$s \leq N^\texttt{search}$}{
    Update $\mathbf{u}_{n+1}^{(k+1)}$ in $\eqref{Global Governing Equation}_3$ based on $\mathbf{u}_{n+1}^{(k)}$, $\delta \mathbf{u}_{n+1}^{(k)}$, and $l_{n+1}^{(k)}$;

    Compute the absolute residual, $a_{n+1}^{(k+1)}$;

    \eIf{$a_{n+1}^{(k+1)} < \varepsilon^\texttt{search}\ a_{n+1}^{(k)}$}{
        Break the line search due to convergence;
    }{
        Shrink the line search parameter, $l_{n+1}^{(k)} \leftarrow l_{n+1}^{(k)}/2$;

        Update the line search count, $s \leftarrow s+1$;
    }
}
\end{algorithm}

\begin{algorithm}[!htbp]
\caption{Adaptive loading scheme}
\label{Algorithm of Adaptive Loading}

\textbf{Inputs}: Convergence flag, \texttt{convergence};
Internal variables from the previous load step, $\overline{\mathbf{b}}_n^\texttt{e}$, $\overline{\boldsymbol{\beta}}_n$, and $\alpha_n$;
Deformation gradient from the previous load step, $\mathbf{F}_n$;
Displacement in the current Newton iteration, $\mathbf{u}_{n+1}^{(k)}$;
Maximum tries of the adaptive loading, $N^\texttt{try}$;
Converged, current, and target time steps, $t^\texttt{conv}$, $t^\texttt{now}$, and $t_m$, respectively;
Load step count, $n$;
Try count, $r$;

\eIf{\texttt{convergence} is \texttt{false} and $r \leq N^\texttt{try}$}{
    Shrink the current time step, $t^\texttt{now} \leftarrow (t^\texttt{conv} + t^\texttt{now}) / 2$;

    Update the try count, $r \leftarrow r + 1$; 
}{
    Reset the try count, $r \leftarrow 0$; 

    Update the converged time step, $t^\texttt{conv} \leftarrow t^\texttt{now}$;

    Update the current time step, $t^\texttt{now} \leftarrow t_m$;

    Update the current incremental consistency parameter, $\widehat{\gamma}_{n+1} \leftarrow \widehat{\gamma}_{n+1}^{(k)}$, and displacement field, $\mathbf{u}_{n+1} \leftarrow \mathbf{u}_{n+1}^{(k)}$;

    Update the current internal variables, $\overline{\mathbf{b}}_{n+1}^\texttt{e}$, $\overline{\boldsymbol{\beta}}_{n+1}$, and $\alpha_{n+1}$, by providing Algorithm \ref{Algorithm of Return Mapping Scheme} with $\overline{\mathbf{b}}_n^\texttt{e}$, $\overline{\boldsymbol{\beta}}_n$,  $\alpha_n$, $\mathbf{F}_n$, and $\mathbf{F}_{n+1} \leftarrow \mathbf{I} + \nabla \mathbf{u}_{n+1}$;

    Update the load step count, $n \leftarrow n+1$;
}
\end{algorithm}

\section{Topology optimization framework for finite strain elastoplasticity}
\label{Sec: Topology Optimization Framework}

Built upon the finite strain elastoplasticity theory in Section \ref{Sec: Finite Strain Elastoplasticity}, we now present the topology optimization framework for optimizing elastoplastic structures undergoing large deformations. This framework comprises design space parameterization, material property interpolation, topology optimization formulation, and sensitivity analysis and verification. This section details the first three components while the sensitivity analysis and verification are provided in \ref{Sec: Sensitivity Analysis and Verification}.

\subsection{Design space parameterization}

Before proceeding with topology optimization, one prerequisite is to identify suitable design variables that parameterize the large design space of structural geometries and material phases. Following the blueprint of topology optimization for infinitesimal strain elastoplasticity \citep{jia_multimaterial_2025}, we start by defining one density variable, $\rho(\mathbf{X}): \Omega_0 \mapsto [0, 1]$, and $N^\xi \geq 1$ material variables, $\xi_1(\mathbf{X}),\ \xi_2(\mathbf{X}),\ \ldots,\ \xi_{N^\xi}(\mathbf{X}): \Omega_0 \mapsto [0, 1]$. To reduce the mesh dependence and avoid the checkerboard patterns of designs, we apply a linear filter \citep{bourdin_filters_2001} on all design variables and derive their filtered versions as
\begin{equation*}
    \widetilde{\zeta}(\mathbf{X}) = \dfrac{\displaystyle \int_{\Omega_0} w_\zeta(\mathbf{X}, \mathbf{X}') \zeta(\mathbf{X}') \text{d} \mathbf{X}'}{\displaystyle \int_{\Omega_0} w_\zeta(\mathbf{X}, \mathbf{X}') \text{d} \mathbf{X}'}
    \quad \text{for} \quad
    \zeta \in \{\rho, \xi_1, \xi_2, \ldots, \xi_{N^\xi}\}.
\end{equation*}
Here, $w_\zeta(\mathbf{X}, \mathbf{X}') = \max \{ 0, R_\zeta-\lVert \mathbf{X}-\mathbf{X}' \rVert_2 \}$ is a weighting factor, and $R_\zeta \geq 0$ is the filter radius. 

To promote pure solid--void designs with discrete material interfaces, we further apply a Heaviside projection \citep{bendsoe_topology_2003} on the filtered design variables, $\widetilde{\rho}(\mathbf{X}),\ \widetilde{\xi}_1(\mathbf{X}),\ \widetilde{\xi}_2(\mathbf{X}),\ \ldots,\ \widetilde{\xi}_{N^\xi}(\mathbf{X}): \Omega_0 \mapsto [0, 1]$, and derive their projected versions as
\begin{equation} \label{Heaviside Projection}
    \widehat{\zeta}(\mathbf{X}) = \dfrac{\tanh(\beta_\zeta \theta_\zeta) + \tanh(\beta_\zeta (\widetilde{\zeta}(\mathbf{X}) - \theta_\zeta))}{\tanh(\beta_\zeta \theta_\zeta) + \tanh(\beta_\zeta (1 - \theta_\zeta))}
    \quad \text{for} \quad
     \zeta \in \{\rho, \xi_1, \xi_2, \ldots, \xi_{N^\xi}\},
\end{equation}
where $\beta_\zeta$ is a sharpness parameter, and $\theta_\zeta$ is the projection threshold. Based on the projected density variable, $\widehat{\rho}(\mathbf{X}): \Omega_0 \mapsto [0, 1]$, we define the physical density variable as $\overline{\rho}(\mathbf{X}) \equiv \widehat{\rho}(\mathbf{X}): \Omega_0 \mapsto [0, 1]$. Here, $\overline{\rho}(\mathbf{X}) = 1$ represents that material point $\mathbf{X}$ is occupied by the solid material, and $\overline{\rho}(\mathbf{X}) = 0$ signifies the void.

As for the projected material variables, $\widehat{\xi}_1(\mathbf{X}),\ \widehat{\xi}_2(\mathbf{X}),\ \ldots,\ \widehat{\xi}_{N^\xi}(\mathbf{X}): \Omega_0 \mapsto [0, 1]$, we further apply a modified version of the hypercube-to-simplex-projection (HSP) \citep{zhou_multi-component_2018} as
\begin{equation*}
    \left\{ \begin{array}{l}
        \overline{\xi}_n(\mathbf{X}) = \displaystyle{ \sum_{k=1}^{2^{N^\xi}} b_{nk}
        \left\{ (-1)^{N^\xi + \sum_{l=1}^{N^\xi}c_{kl}}
        \prod_{m=1}^{N^\xi} (\widehat{\xi}_m(\mathbf{X}) + c_{km} - 1) \right\} }
        \quad \text{for} \quad n=1,2,\ldots,N^\xi, \\[15pt]
        \overline{\xi}_{N^\texttt{mat}}(\mathbf{X}) = \displaystyle{ 1 - \sum_{n=1}^{N^\xi} } \overline{\xi}_n(\mathbf{X}),
    \end{array} \right.
\end{equation*}
to ensure that the summation of the physical material variables, $\overline{\xi}_1(\mathbf{X}),\ \overline{\xi}_2(\mathbf{X}),\ \ldots,\ \overline{\xi}_{N^\texttt{mat}}(\mathbf{X}): \Omega_0 \mapsto [0, 1]$, is equal to one ($\sum_{n=1}^{N^\texttt{mat}}\overline{\xi}_n(\mathbf{X}) = 1$). Here, the parameter $N^\texttt{mat} = N^\xi + 1$ represents the number of candidate materials. The parameter $c_{kl} \in \{0, 1\}$ is the $l$th coordinate component of the $k$th vertex of a unit hypercube in the $N^\texttt{mat}-1$ dimensional space. The parameter $b_{nk}$ is computed as
\begin{equation*}
    b_{nk} = \left\{ \begin{array}{cc}
        \dfrac{c_{nk}}{\sum_{l} c_{nl}}, & \text{if $\sum_{l} c_{nl} \geq 1$}; \\[8pt]
        0, & \text{otherwise}.
    \end{array} \right.
\end{equation*}

Note that for a material point with $\overline{\rho}(\mathbf{X}) = 1$, the expression $\overline{\xi}_n(\mathbf{X}) = 1$ for $n=1,2,\ldots,N^\texttt{mat}$ signifies that material $n$ is used. Consequently, we can use the physical density ($\overline{\rho}(\mathbf{X})$) and material ($\overline{\xi}_1(\mathbf{X}),\ \overline{\xi}_2(\mathbf{X}),\ \ldots,\ \overline{\xi}_{N^\texttt{mat}}(\mathbf{X})$) variables to describe the structural geometries and material phases in the design domain ($\Omega_0$), respectively.

\subsection{Material property interpolation}

After design space parameterization, one immediate next step is to interpolate the elasticity- and plasticity-related material constants/functions for various values of physical design variables, $\overline{\rho}(\mathbf{X}) \in [0, 1]$ and $\overline{\xi}_n(\mathbf{X}) \in [0, 1]$ for $n=1,2,\ldots,N^\texttt{mat}$. Following \citet{jia_multimaterial_2025}, we interpolate the initial bulk modulus ($\kappa$), initial shear modulus ($\mu$), isotropic hardening function ($k$), and kinematic hardening modulus ($h$) as
\begin{equation} \label{Material Interpolation Functions}
    \overline{\chi}(\mathbf{X}) = \left[\varepsilon_\rho + (1-\varepsilon_\rho) \overline{\rho}^{p_\chi}(\mathbf{X}) \right] 
    \sum_{n=1}^{N^\texttt{mat}} \overline{\xi}_n^{p_\xi}(\mathbf{X}) \chi_n
    \quad \text{for} \quad
    \chi \in \{ \kappa, \mu, k, h \}.
\end{equation}
Here, the variables $\overline{\kappa}(\mathbf{X})$, $\overline{\mu}(\mathbf{X})$, $\overline{k}(\mathbf{X})$, and $\overline{h}(\mathbf{X})$ are the interpolated material constants/functions. The material constant $\chi_n \in \{\kappa_n, \mu_n, k_n, h_n\}$ for $n=1,2,\ldots,N^\texttt{mat}$ represents the property/function of candidate material $n$. The parameters $p_\kappa = p_\mu$, $p_k$, and $p_h$ penalize the intermediate values of the physical density variable, $\overline{\rho}(\mathbf{X}) \in (0, 1)$, and $p_\xi$ penalizes the intermediate values of the physical material variables, $\overline{\xi}_n(\mathbf{X}) \in (0, 1)$ for $n=1,2,\ldots,N^\texttt{mat}$. The parameter $\varepsilon_\rho$ is a small positive number to prevent the singularity of the stiffness of void materials.

In addition to the interpolation rules in \eqref{Material Interpolation Functions} for material constants/functions, it is essential to interpolate the constitutive laws across the range of the physical density variable, $\overline{\rho} \in [0, 1]$. This consideration arises because many topology optimization problems involving finite strain analysis \citep{li_programming_2023} encounter divergence in FEA solutions. Such issues are typically caused by the excessive distortion of elements filled with void materials when subjected to large deformations. To address this challenge, a widely adopted and effective strategy for hyperelastic designs \citep{wang_interpolation_2014} assumes that void materials exhibit linear elastic behavior while solid materials retain their hyperelastic characteristics.

In the context of finite strain elastoplasticity, we propose to interpolate the first Piola--Kirchhoff stress tensor as
\begin{equation} \label{Interpolated First PK Stress}
    \widecheck{\mathbf{P}}(\mathbf{F}, \widecheck{\mathbf{F}}; \phi)
    = \phi \mathbf{P}(\widecheck{\mathbf{F}}) - \phi \boldsymbol{\sigma}^\texttt{l}(\boldsymbol{\epsilon}(\widecheck{\mathbf{F}})) + \boldsymbol{\sigma}^\texttt{l}(\boldsymbol{\epsilon}(\mathbf{F})).
\end{equation}
The variable $\widecheck{\mathbf{F}}$ is the interpolated total deformation gradient defined as
\begin{equation*}
    \widecheck{\mathbf{F}} = \mathbf{F} (\phi \mathbf{u}) = \mathbf{I} + \phi \nabla \mathbf{u} = (1-\phi) \mathbf{I} + \phi \mathbf{F}(\mathbf{u})
\end{equation*}
with
\begin{equation*}
    \phi(\mathbf{X}) = \dfrac{\tanh(\beta_\phi \theta_\phi) + \tanh(\beta_\phi (\overline{\rho}^{p_\kappa} (\mathbf{X}) - \theta_\phi))}{\tanh(\beta_\phi \theta_\phi) + \tanh(\beta_\phi (1 - \theta_\phi))}
\end{equation*}
where we assume $\phi(\mathbf{X})$ is element-wise constant and therefore $\nabla \phi(\mathbf{X}) = 0$ at the integration points. The parameters $\beta_\phi = 500$ and $\theta_\phi = 0.1$ are similar to $\beta_\zeta$ and $\theta_\zeta$ in \eqref{Heaviside Projection}, respectively. The variable $\boldsymbol{\sigma}^\texttt{l}$ is the infinitesimal stress tensor (associated with $W$ in \eqref{Strain Energy Density Function}) expressed as
\begin{equation*}
    \boldsymbol{\sigma}^\texttt{l}(\boldsymbol{\epsilon}) = \overline{\kappa}\ \text{tr}(\boldsymbol{\epsilon}) \mathbf{I} + 2 \overline{\mu}\ \text{dev}(\boldsymbol{\epsilon}),
\end{equation*}
and $\boldsymbol{\epsilon}$ is the infinitesimal strain tensor defined as
\begin{equation*}
    \boldsymbol{\epsilon}(\mathbf{F}) = \dfrac{1}{2} \left( \mathbf{F} + \mathbf{F}^\top \right) - \mathbf{I}.
\end{equation*}

\begin{remark}
The physical meaning of $\widecheck{\mathbf{P}}$ in \eqref{Interpolated First PK Stress} merits explicit clarification. For solid materials where $\overline{\rho}(\mathbf{X}) = 1$ and therefore $\phi(\mathbf{X}) = 1$, the variable $\widecheck{\mathbf{P}}$ reduces to $\mathbf{P}$ in \eqref{First PK Stress} for elastoplastic materials exactly. For void materials where $\overline{\rho}(\mathbf{X}) = 0$ and therefore $\phi(\mathbf{X}) = 0$, the variable $\widecheck{\mathbf{P}}$ reduces to $\boldsymbol{\sigma}^\texttt{l}$ for linear elastic materials exactly without any large deformations and yield behaviors (as in the air). This linear elasticity assumption for void materials brings negligible errors (as verified in \ref{Sec: Comparison of Meshes}) due to the tiny stiffness of void materials compared to solid ones.
\end{remark}

Corresponding to the interpolated first Piola--Kirchhoff stress ($\widecheck{\mathbf{P}}$) in \eqref{Interpolated First PK Stress}, we compute the interpolated first elastoplastic moduli as
\begin{equation*}
    \widecheck{\mathbb{C}} = \dfrac{\partial \widecheck{\mathbf{P}}}{\partial \mathbf{F}}
    = \phi^2 \widehat{\mathbb{C}}(\widecheck{\mathbf{F}}) + (1-\phi^2) \mathbb{C}^\texttt{l},
\end{equation*}
where $\mathbb{C}^\texttt{l}$ represents the tangent moduli of the associated linear elastic material expressed as
\begin{equation*}
    \mathbb{C}^\texttt{l} \equiv \overline{\kappa} \mathbf{I} \otimes \mathbf{I} + 2 \overline{\mu} \left( \mathbb{I} - \dfrac{1}{3} \mathbf{I} \otimes \mathbf{I} \right).
\end{equation*}
Consequently, the integrand of the left-hand side of the global incremental equilibrium equation in $\eqref{Global Governing Equation}_1$ becomes as
\begin{equation} \label{Interpolated Left-hand Side}
    \nabla \mathbf{v} : \widecheck{\mathbb{C}} : \nabla \delta \mathbf{u} = \phi^2 \nabla \mathbf{v} : \widehat{\mathbb{C}}(\widecheck{\mathbf{F}}) : \nabla \delta \mathbf{u} + (1-\phi^2) \nabla \mathbf{v} : \mathbb{C}^\texttt{l} : \nabla \delta \mathbf{u}.
\end{equation}
Note that one can effortlessly compute $\nabla \mathbf{v} : \widehat{\mathbb{C}}(\widecheck{\mathbf{F}}) : \nabla \delta \mathbf{u}$ by replacing $\mathbf{F}$ with $\widecheck{\mathbf{F}}$ in \eqref{Integrand of Left-hand Side}. As for the second term in \eqref{Interpolated Left-hand Side}, we write out
\begin{equation*}
    \nabla \mathbf{v} : \mathbb{C}^\texttt{l} : \nabla \delta \mathbf{u}
    = \left( \overline{\kappa} - \dfrac{2\overline{\mu}}{3} \right) \text{tr} (\nabla \delta \mathbf{u}) \text{tr} (\nabla \mathbf{v}) 
    + 2\overline{\mu} (\nabla \delta \mathbf{u})^\texttt{sym} : \nabla \mathbf{v}.
\end{equation*}
Finally, substituting the interpolation rules in \eqref{Material Interpolation Functions}, \eqref{Interpolated First PK Stress}, and \eqref{Interpolated Left-hand Side} into \eqref{Global Governing Equation} yields the parameterized global equilibrium equations, from which one can solve for the structural elastoplastic responses with various structural geometries and material phases.

\subsection{Topology optimization formulation}

Once deriving the parameterized global equilibrium equations, we can present the topology optimization formulation as
\begin{equation} \label{Topology Optimization Formulation}
    \left\{ \begin{array}{ll}
        \underset{\rho, \xi_1, \ldots, \xi_{N^\xi}}{\text{maximize:}} & J(\overline{\rho}, \overline{\xi}_1, \ldots, \overline{\xi}_{N^\texttt{mat}}; \mathbf{u}_1, \ldots, \mathbf{u}_N; \overline{\mathbf{b}}_1^\texttt{e}, \ldots, \overline{\mathbf{b}}_N^\texttt{e}; \overline{\boldsymbol{\beta}}_1, \ldots, \overline{\boldsymbol{\beta}}_N; \alpha_1, \ldots, \alpha_N; \widehat{\gamma}_1, \ldots, \widehat{\gamma}_N); \\[12pt]
        
        \text{subject to:} & \left\{ \begin{array}{l}
            g_{V0} = \displaystyle \dfrac{1}{|\Omega_0|} \int_{\Omega_0} \overline{\rho}(\mathbf{X}) \text{d} \mathbf{X} - \overline{V} \leq 0; \\[15pt]
            g_{Vn} = \displaystyle \dfrac{1}{|\Omega_0|} \int_{\Omega_0} \overline{\rho}(\mathbf{X}) \overline{\xi}_n(\mathbf{X}) \text{d} \mathbf{X} - \overline{V}_n \leq 0 \quad \text{for} \quad n = 1,2,\ldots,N^\texttt{mat}; \\[15pt]
            g_\Psi = \displaystyle \dfrac{1}{|\Omega_0|} \int_{\Omega_0} \overline{\rho}(\mathbf{X}) \sum_{n=1}^{N^\texttt{mat}} \left[ \overline{\xi}_n(\mathbf{X}) \Psi_n \right] \text{d} \mathbf{X} - \overline{\Psi} \leq 0 \quad \text{for} \quad \Psi \in \{P, M, C\}; \\[15pt]
            \rho(\mathbf{X}), \xi_1(\mathbf{X}), \ldots, \xi_{N^\xi}(\mathbf{X}) \in [0, 1];
        \end{array} \right. \\
        \\[2pt]

        \text{with:} & \left\{ \begin{array}{l}
            \text{Updating formulae for $\overline{\boldsymbol{\beta}}_1, \ldots, \overline{\boldsymbol{\beta}}_N$ and $\alpha_1, \ldots, \alpha_N$ in \eqref{Updated beta and alpha}}; \\[8pt]
            \text{Updating formulae for $\overline{\mathbf{b}}_1^\texttt{e}, \ldots, \overline{\mathbf{b}}_N^\texttt{e}$ in \eqref{Updated be_bar}}; \\[8pt]
            \text{Kuhn--Tucker condition for $\widehat{\gamma}_1, \ldots, \widehat{\gamma}_N$ in \eqref{Function of gamma}}; \\[8pt]
            \text{Global equilibrium equations for $\mathbf{u}_1, \ldots, \mathbf{u}_N$ in \eqref{Global Governing Equation}}.
        \end{array} \right.
    \end{array} \right.
\end{equation}
In these expressions, the variable $J$ represents a general objective function to be maximized, with its detailed formulation to be introduced later. The parameter $N$ represents the maximum load steps in FEA. The constraints $g_{V0}$ and $g_{Vn}$, for $n = 1, 2, \ldots, N^\texttt{mat}$, correspond to the total material volume fraction (with an upper bound $\overline{V} \in (0, 1]$) and the volume fraction of material $n$ (with an upper bound $\overline{V}_n \in [0, \overline{V}]$), respectively. Additionally, the constraints $g_P$, $g_M$, and $g_C$ impose limits on price, mass density, and CO$_2$ footprint, respectively, with $\overline{P}$, $\overline{M}$, and $\overline{C}$ representing their respective upper bounds. The parameters $P_n$, $M_n$, and $C_n$ denote the price, mass density, and CO$_2$ footprint of material $n$, respectively.

In this study, we aim to optimize the structural performance metrics of stiffness, strength, and effective structural toughness (through total energy). Following the footprint of \citet{jia_controlling_2023, jia_multimaterial_2025}, we design a comprehensive multi-objective function for finite strain elastoplasticity as
\begin{equation} \label{Multi-objective Function}
    J = w_\texttt{stiff} J_\texttt{stiff} + w_\texttt{force} J_\texttt{force} + w_\texttt{energy} J_\texttt{energy}.
\end{equation}
Here, $w_\texttt{stiff}$, $w_\texttt{force}$, and $w_\texttt{energy} \in [0, 1]$ are weighting factors subject to $w_\texttt{stiff} + w_\texttt{force} + w_\texttt{energy} = 1$. The energy-type variables $J_\texttt{stiff}$, $J_\texttt{force}$, and $J_\texttt{energy}$ represent stiffness, strength (end force), and structural toughness (total energy), respectively, and are defined as
\begin{equation} \label{Objective Function Terms}
    \left\{ \begin{array}{l}
        \displaystyle J_\texttt{stiff} = \dfrac{1}{2} \int_{\Omega_0} \widecheck{\mathbf{P}}_1 : \nabla \mathbf{u}_1 \text{d} \mathbf{X}, \quad 
        J_\texttt{force} = \int_{\Omega_0} \widecheck{\mathbf{P}}_N : \nabla \mathbf{u}_N \text{d} \mathbf{X}, \\[10pt]
        \displaystyle J_\texttt{energy} = \dfrac{1}{2} \sum_{n=1}^N \int_{\Omega_0} \left( \widecheck{\mathbf{P}}_n + \widecheck{\mathbf{P}}_{n-1} \right) : \left( \nabla \mathbf{u}_n - \nabla \mathbf{u}_{n-1} \right) \text{d} \mathbf{X}.
    \end{array} \right.
\end{equation}
Recall that $\mathbf{u}_n$ and $\widecheck{\mathbf{P}}_n$ for $n = 1,2,\ldots,N$ are the displacement and interpolated first Piola--Kirchhoff stress at load step $n$, respectively. The terms $J_\texttt{stiff}$, $J_\texttt{force}$, and $J_\texttt{energy}$ in \eqref{Objective Function Terms} are interpreted as follows. 
\begin{itemize}
    \item The variable $J_\texttt{stiff}$ represents the strain energy at the first load step, where plasticity has not yet evolved. Under displacement loading, $J_\texttt{stiff}$ is positively correlated with the initial stiffness of the structure. Therefore, maximizing $J_\texttt{stiff}$ enhances structural stiffness.
    \item The variable $J_\texttt{force}$ corresponds to the product of force and displacement at the final load step, also known as the ``end compliance" for hyperelastic structures. Under displacement loading, maximizing $J_\texttt{force}$ increases the final reaction force.
    \item Finally, the variable $J_\texttt{energy}$ captures the total energy absorption and dissipation of the structure, equivalent to the area enclosed by its force--displacement curve.
\end{itemize}

\begin{remark}
In the multi-objective function in \eqref{Multi-objective Function}, we opt to use the (interpolated) first Piola--Kirchhoff stress ($\widecheck{\mathbf{P}}$) as the the primary measure in contrast to the Kirchhoff stress ($\widehat{\boldsymbol{\tau}}$) in Section \ref{Sec: Finite Strain Elastoplasticity}. It is because $\widecheck{\mathbf{P}}$ is the easiest stress measure \citep{kumar_phase-field_2020} to gauge practically. However, we note that the multi-objective function in \eqref{Multi-objective Function} is general, and one can effortlessly replace $\widecheck{\mathbf{P}}$ with any other stress measures if desired.
\end{remark}

\begin{remark}
The proposed topology optimization formulation in \eqref{Topology Optimization Formulation} equipped with the general multi-objective function in \eqref{Multi-objective Function} can be simplified to recover several special cases that have demonstrated success. For example, setting $w_\texttt{stiff} = w_\texttt{energy} = 0$ and $w_\texttt{force}=1$ yields $J = J_\texttt{force}$, which mainly maximizes the ``end compliance" as in \citet{ivarsson_plastic_2021}. Similarly, setting $w_\texttt{stiff} = w_\texttt{force} = 0$ and $w_\texttt{energy} = 1$ derives $J = J_\texttt{energy}$, which primarily maximizes the energy of structures as in \citet{wallin_topology_2016}. Additionally, employing $w_\texttt{stiff} = 0$ and $w_\texttt{force} = w_\texttt{energy} = 0.5$ generates $J = 0.5 J_\texttt{force} + 0.5 J_\texttt{energy}$ as in \citet{abueidda_topology_2021}. As illustrated in Section \ref{Sec: Sample Examples}, utilizing the comprehensive multi-objective function in \eqref{Multi-objective Function} enables the generation of optimized structures exhibiting diverse and tailored elastoplastic responses under large deformations.
\end{remark}

\section{Optimized elastoplastic designs with real-world applications}
\label{Sec: Sample Examples}

In this section, we present optimized elastoplastic designs with real-world applications through four representative examples. These examples demonstrate the effectiveness of the proposed multimaterial topology optimization framework (Section \ref{Sec: Topology Optimization Framework}) in optimizing structural elastoplastic responses under large deformations. Additionally, we provide mechanical insights into achieving these optimized behaviors.

In the first example, we optimize the energy dissipation of multi-alloy dampers subjected to various cyclic loadings. This investigation reveals that employing multiple materials enhances energy dissipation compared to single-material ones under large deformations. Additionally, we demonstrate the transition from the kinematic hardening dominance to the isotropic hardening dominance, which improves energy dissipation as applied displacements increase. We also show that optimized dampers exhibit superior energy dissipation across multiple loading cycles compared to an intuitive design.

In the second example, we optimize the initial stiffness and end force of double-clamped beams. This example demonstrates the versatility of the proposed topology optimization framework in handling multiple design objectives for structures made of purely hyperelastic, purely elastoplastic, and mixed materials. We also explore the synergistic use of hyperelastic and elastoplastic materials to achieve various stiffness--strength interplays in composite structures.

In the third example, we maximize the crashworthiness of impact-resisting bumpers. Extending the optimization task to 3D geometries, we consider scenarios with more than two candidate materials. Through this example, we show that the proposed optimization framework applies to various spatial dimensions and an arbitrary number of candidate materials.

Finally, we optimize the end force of cold working profiled sheets accounting for both the processing (metal-forming) and service (load-carrying) stages, which involve ultra-large elastoplastic deformations. This example introduces multi-stage topology optimization, which incorporates both intra-stage and inter-stage history dependencies. Additionally, we integrate practical constraints --- cost, lightweight, and sustainability --- into the design and optimization of elastoplastic structures, bridging the gap between mechanical designs and non-mechanical considerations.

We remark that these examples span a spectrum of material properties, as summarized in Table \ref{Table: Material Properties}. For most examples, we consistently use titanium, bronze, nickel--chromium, and stainless steel, which exhibit perfect plasticity, linear isotropic hardening, nonlinear isotropic hardening, and kinematic hardening, respectively. Importantly, these four materials can be joined together using 3D printing techniques, as reported in \citet{wei_overview_2020, wei_cu10sn_2022}. In the double-clamped beam example, we utilize elastoplastic lithium, characterized by combined isotropic and kinematic hardening, alongside hyperelastic PCL, modeled by assigning a sufficiently large yield strength of $\sigma_y = 25$ MPa.

\begin{table}[!htbp]
    \caption{Material properties used in numerical examples}
    \label{Table: Material Properties}
    \centering
    \scriptsize
    \begin{tabular}{lccccccccc}
        \hline
        \textbf{Materials}
        & \begin{tabular}[x]{@{}c@{}}
            \textbf{Titanium} \\
            (Ti--6Al--4V)
        \end{tabular}
        & \begin{tabular}[x]{@{}c@{}}
            \textbf{Bronze} \\
            (CuSn10)
        \end{tabular}
        & \begin{tabular}[x]{@{}c@{}}
            \textbf{Nickel} \\
            \textbf{--chromium} \\
            (INCONEL 718)
        \end{tabular}
        & \begin{tabular}[x]{@{}c@{}}
            \textbf{Stainless steel} \\
            (AISI 316L)
        \end{tabular}
        & \begin{tabular}[x]{@{}c@{}}
            \textbf{Lithium} \\
            (commercial \\
            purity)
        \end{tabular}
        & \begin{tabular}[x]{@{}c@{}}
            \textbf{PCL} \\
            (polycapro \\
            -lactone)
        \end{tabular} \\
        \hline
        Features & \begin{tabular}[x]{@{}c@{}}
            Perfect \\
            plasticity
        \end{tabular} & \begin{tabular}[x]{@{}c@{}}
            Linear \\
            isotropic \\
            hardening
        \end{tabular} & \begin{tabular}[x]{@{}c@{}}
            Nonlinear \\
            isotropic \\
            hardening
        \end{tabular} & \begin{tabular}[x]{@{}c@{}}
            Kinematic \\
            hardening
        \end{tabular} & \begin{tabular}[x]{@{}c@{}}
            Combined \\
            hardening
        \end{tabular} & \begin{tabular}[x]{@{}c@{}}
            Near \\
            hyperelasticity \end{tabular} \\[10pt]

        \begin{tabular}[x]{@{}l@{}}
            Bulk moduli, \\
            $\kappa$ (GPa)
        \end{tabular} & 115.6 & 88.9 & 165.0 & 141.3 & 5.8 & 0.3880 \\[10pt]
        \begin{tabular}[x]{@{}l@{}}
            Shear moduli, \\
            $\mu$ (GPa)
        \end{tabular} & 41.4 & 29.6 & 76.2 & 76.8 & 1.8 & 0.0157 \\[10pt]
        \begin{tabular}[x]{@{}l@{}}
            Young's moduli, \\
            $E$ (GPa)
        \end{tabular} & 111.0 & 80.0 & 198.0 & 195.0 & 4.9 & 0.0466 \\[10pt]
        Poisson's ratios, $\nu$ & 0.34 & 0.35 & 0.30 & 0.27 & 0.36 & 0.48 \\[10pt]
        \begin{tabular}[x]{@{}l@{}}
            Kinematic hardening \\
            moduli, $h$ (MPa)
        \end{tabular} & 0.0 & 0.0 & 0.0 & 1339.1 & 2.5 & 0.0 \\[10pt]
        \begin{tabular}[x]{@{}l@{}}
            Isotropic hardening \\
            moduli, $K$ (MPa)
        \end{tabular} & 0.0 & 952.0 & 129.0 & 0.0 & 2.5 & 0.0 \\[10pt]
        \begin{tabular}[x]{@{}l@{}}
            Initial yield \\
            strengths, $\sigma_y$ (MPa)
        \end{tabular} & 853.0 & 145.0 & 450.0 & 226.0 & 1.0 & 25.0 \\[10pt]
        \begin{tabular}[x]{@{}l@{}}
            Residual yield \\
            strengths, $\sigma_\infty$ (MPa)
        \end{tabular} & 853.0 & 145.0 & 715.0 & 226.0 & 1.0 & 25.0 \\[10pt]
        \begin{tabular}[x]{@{}l@{}}
            Saturation \\
            exponents, $\delta$ 
        \end{tabular} & 0.0 & 0.0 & 16.9 & 0.0 & 0.0 & 0.0 \\[10pt]
        \begin{tabular}[x]{@{}l@{}}
            Isotropic hardening \\
            functions, $k(\alpha)$ (MPa)
        \end{tabular} & 853.0 & $145.0 + 952.0 \alpha$ & \begin{tabular}[x]{@{}c@{}}
            $450.0 + 129.0 \alpha$ \\
            $+ 265.0 (1-e^{-16.9 \alpha})$
        \end{tabular}  & 226.0 & $1.0 + 2.5 \alpha$ & 25.0 \\[10pt]
        Price (USD/kg) & 24.4 & 13.3 & 25.2 & 6.6 & 127.0 & 6.8 \\[10pt]
        \begin{tabular}[x]{@{}l@{}}
            Mass densities \\
             ($10^3$ kg/m$^3$)
        \end{tabular} & 4.4 & 8.8 & 8.2 & 8.0 & 0.5 & 1.1 \\[10pt]
        CO$_2$ footprint (kg/kg) & 40.4 & 6.0 & 16.6 & 7.4 & 79.6 & 2.3 \\
        \hline
    \end{tabular}
\end{table}

To implement the proposed topology optimization framework, we adapt the open-source software\footnote{Available in GitHub at \url{https://github.com/missionlab/fenitop}.} \citep{jia_fenitop_2024} to incorporate finite strain elastoplasticity. The software efficiently computes partial derivatives using automatic differentiation and accelerates computations through parallel processing. To ensure high solution precision, we utilize a direct solver with lower--upper factorization to solve the global equilibrium equations in \eqref{Global Governing Equation} and the adjoint equations in \eqref{Reduced Adjoint Equation at Final Step}--\eqref{Reduced Adjoint Equation at Remaining Steps}. Additionally, to determine the optimal design variables, we employ the well-established gradient-based optimizer, the method of moving asymptotes \citep{svanberg_method_1987}. This optimizer utilizes the sensitivity analysis presented in \ref{Sec: Sensitivity Analysis and Verification} and efficiently updates the design variables until convergence. The optimized elastoplastic structures and the mechanisms enabling these design objectives are detailed below.

\subsection{Metallic yielding dampers with maximized hysteretic energy}
\label{Sec: Dampers}

\subsubsection{Single-alloy versus multi-alloy dampers}

In this subsection, we optimize the energy dissipation of metallic dampers \citep{zhang_simplified_2018, jia_residual_2022}, which are typically used in structural engineering to suppress vibrations induced by earthquakes or winds. Through these optimized dampers, we aim to demonstrate the effectiveness of the proposed topology optimization framework in optimizing structural elastoplastic responses under large deformations. Additionally, we highlight the necessity of employing multiple alloys to achieve greater energy dissipation.

To achieve these goals, we revisit the damper optimization problem in \citet{jia_multimaterial_2025}. As shown in Fig. \ref{Fig: Damper-Part 1}(a), the design domain is a rectangle subjected to simple shear loadings under the plane strain condition. We aim to optimally distribute the bronze and steel materials, whose Kirchhoff stress--Lagrangian strain ($\tau_{11}$--$E_{11}$) curves are shown in Fig. \ref{Fig: Damper-Part 1}(b), to maximize the total hysteretic energy (Fig. \ref{Fig: Damper-Part 1}(c)). Notably, the current setup incorporates non-monotonic, cyclic loadings and accounts for large deformations, representing a more complex scenario compared to the original problem in \citet{jia_multimaterial_2025}.

To solve this problem, we employ the following FEA and optimization strategies. The design domain is discretized using a structured mesh consisting of 15,000 first-order quadrilateral elements. The filter radius is set to $R_\zeta = 10$ mm for $\zeta \in \{\rho, \xi_1, \ldots, \xi_{N^\xi} \}$. The Heaviside sharpness parameter, $\beta_\zeta$, is initially set to 1 and is doubled every 40 optimization iterations starting from iteration 41, until it reaches a maximum value of $\beta_\zeta = 512$. The penalty parameters for the density variable are fixed at $p_\kappa = p_\mu = p_h = 3$ and $p_k = 2.5$, while the penalty parameter for the material variables, $p_\xi$, starts at 1 and increases by 0.25 every 40 optimization iterations from iteration 41, up to a maximum of $p_\xi = 3$. The weighting factors in \eqref{Multi-objective Function} are set as $w_\texttt{stiff} = w_\texttt{force} = 0$ and $w_\texttt{energy} = 1$. The constraint is defined by the total material volume, $g_{V0} \leq 0$, with an upper bound of $\overline{V} = 0.5$. The maximum number of optimization iterations is set to 500. During the topology optimization process, we use 32 non-uniform load steps in the FEA. Following optimization, we refine the analysis by using 44 load steps to evaluate the elastoplastic responses of all designs.

\begin{figure}[!htbp]
    \centering
    \includegraphics[height=20cm]{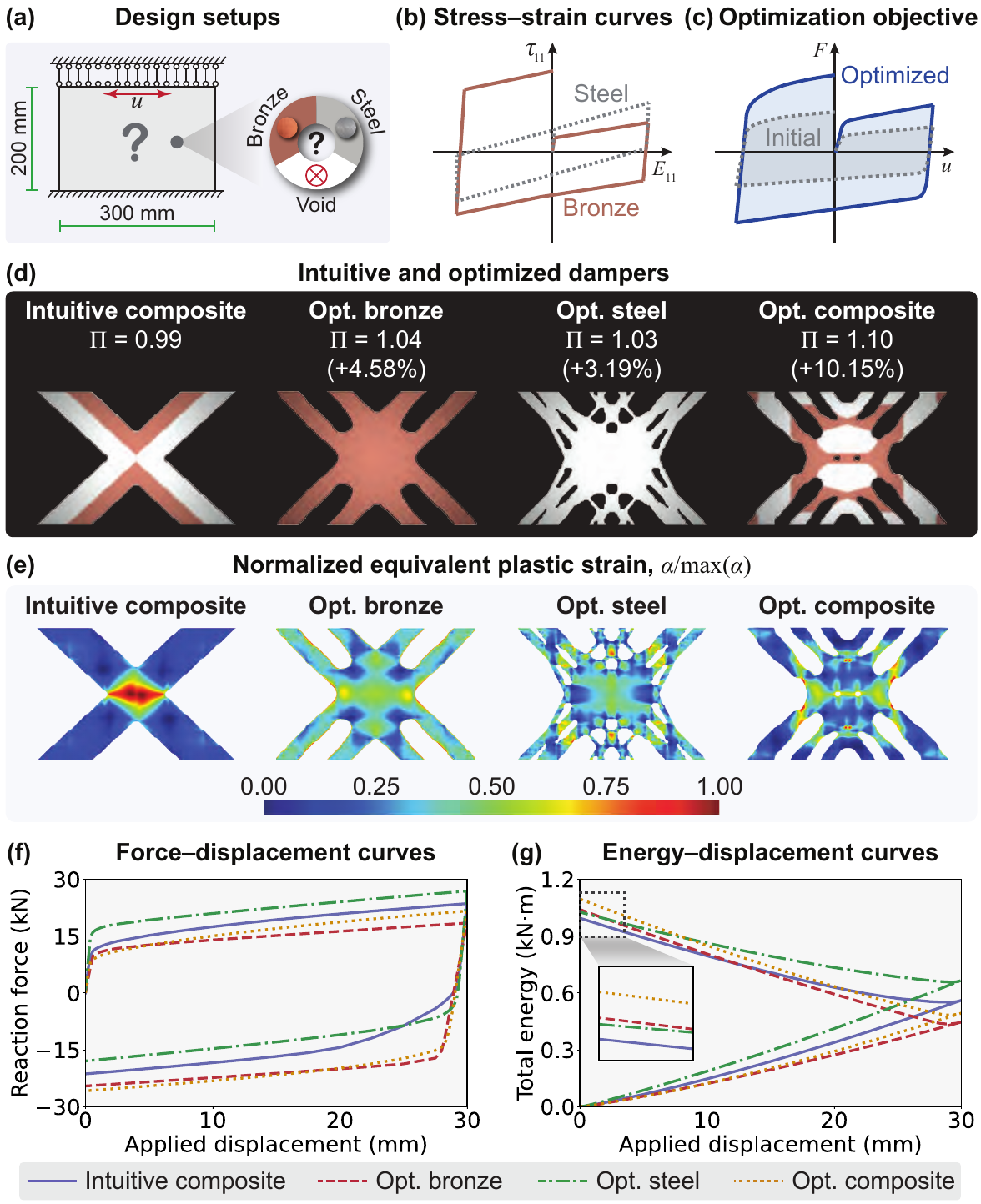}
    \caption{Design and optimization of metallic yielding dampers. (a) Design setups: the design domain, boundary conditions, and candidate materials (bronze, steel, and void). The variable $u$ is the applied displacement. (b) Uniaxial Kirchhoff stress--Lagrangian strain ($\tau_{11}$--$E_{11}$) curves of candidate materials. (c) Optimization objective: maximizing the total energy of dampers. The variable $F$ is the reaction force. (d) Intuitive and optimized dampers. The total energy ($\Pi$) is in kN$\cdot$m. The percentages are the total energy increments compared to the intuitive design. (e) Normalized equivalent plastic strain ($\alpha/\max (\alpha)$). (f) Force--displacement ($F$--$u$) curves. (g) Energy--displacement ($\Pi$--$u$) curves.}
    \label{Fig: Damper-Part 1}
\end{figure}

The damper designs under half-cycle loadings, along with the hysteretic energy ($\Pi$ in kN$\cdot$m), are illustrated in Fig. \ref{Fig: Damper-Part 1}(d). As a baseline, we include one intuitive composite design that uses the same material volume fraction ($\overline{V} = 0.5$) as the three optimized designs. These optimized designs are tailored for bronze, steel, and multiple materials, respectively. Compared to the intuitive composite design, the three optimized dampers achieve greater hysteretic energy, with increases of 4.58\%, 3.19\%, and 10.15\%, respectively. These improvements highlight the effectiveness of the proposed topology optimization framework in enhancing elastoplastic responses under large deformations. Moreover, the optimized composite design naturally favors the use of two materials and achieves higher energy dissipation than the two single-material optimized designs. This performance gain aligns with the infinitesimal strain scenario in \citet{jia_multimaterial_2025} and reinforces the benefit of using multiple materials to improve design performance under large deformations.

To understand the performance gains, we present the normalized equivalent plastic strain ($\alpha/\max (\alpha)$) for all dampers in Fig. \ref{Fig: Damper-Part 1}(e). In the intuitive composite design, plastic strain is confined to the central region, whereas the three optimized designs distribute plastic deformation more evenly, allowing a greater portion of the material to effectively contribute to total energy dissipation. Further, by examining the Kirchhoff stress--Lagrangian strain ($\tau_{11}$--$E_{11}$) curves in Fig. \ref{Fig: Damper-Part 1}(b), we observe that the steel material dissipates more energy during the first quadrant (loading stage), while the bronze material dissipates more energy during the fourth quadrant (unloading stage). This difference in energy dissipation between materials is manifested in the dissimilar hysteretic loops of the force--displacement ($F$--$u$) curves for the optimized bronze and steel designs (Fig. \ref{Fig: Damper-Part 1}(f)). In contrast, the optimized composite damper effectively utilizes both materials to balance energy dissipation during the loading and unloading stages, forming the largest hysteresis loop, which ultimately results in the highest total energy dissipation shown in Fig. \ref{Fig: Damper-Part 1}(g).

\subsubsection{Dampers under complete cyclic loadings with increasing displacement amplitudes}

Having demonstrated the effectiveness of the topology optimization framework and the necessity of employing multiple materials, we now evaluate the performance of optimized composite dampers under complete cyclic loadings with increasing displacement amplitudes. As shown in Fig. \ref{Fig: Damper-Part 2}(a), we present three optimized dampers with applied displacement amplitudes of $u_\texttt{max} = 10$, 20, and 30 mm, respectively. For each optimized damper, we report the total energy ($\Pi$) and its increment compared to the intuitive design in Fig. \ref{Fig: Damper-Part 1}(d). All three optimized dampers demonstrate performance gains --- 20.25\%, 17.28\%, and 18.60\%, respectively --- relative to the intuitive design. These gains are attributed to the optimized dampers' ability to engage more materials in energy dissipation, as evidenced by the normalized total energy density ($W/\max (W)$) in Fig. \ref{Fig: Damper-Part 2}(c). Notably, these performance improvements exceed the 10.15\% gain achieved by the optimized composite damper under half-cycle loadings (Fig. \ref{Fig: Damper-Part 1}(d)). This enhanced performance is due to the full-cycle optimized designs achieving significant energy increases in the second quadrant (reloading stage), despite slight energy decreases in the first quadrant (loading stage) for $u_\texttt{max} = 20$ and 30 mm, as shown in Fig. \ref{Fig: Damper-Part 2}(b) and (c).

A comparison among the three optimized designs in Fig. \ref{Fig: Damper-Part 2}(a) further reveals the influence of displacement amplitude ($u_\texttt{max}$) on material distribution. As displacement amplitudes increase, the optimized dampers increasingly favor the use of bronze over steel. This trend can be explained by the material behaviors illustrated in the Kirchhoff stress--Lagrangian strain ($\tau_{11}$--$E_{11}$) curves in Fig. \ref{Fig: Damper-Part 1}(b). The steel material, which exhibits kinematic hardening, produces a hysteretic loop approximating a parallelogram. Its total energy density is expressed as
\begin{equation} \label{Total Energy Density of Steel}
    W^\texttt{s} \approx 4 \sigma_y^\texttt{s} E_\texttt{max}
\end{equation}
where $\sigma_y^\texttt{s}$ is the yield strength of steel and $E_\texttt{max}$ is the Lagrangian strain amplitude. In contrast, the hysteretic loop of the bronze material expands with larger strains due to isotropic hardening, with its total energy density given by
\begin{equation} \label{Total Energy Density of Bronze}
    W^\texttt{b} \approx 2 E_\texttt{max} \left[ \sigma_y^\texttt{b} + (\sigma_y^\texttt{b} + 4 E_\texttt{t}^\texttt{b} E_\texttt{max}) \right]
    = 8 E_\texttt{t}^\texttt{b} E_\texttt{max}^2 + 4 \sigma_y^\texttt{b} E_\texttt{max}
\end{equation}
where $\sigma_y^\texttt{b}$ and $E_\texttt{t}^\texttt{b} > 0$ are the yield strength and average plastic tangent modulus of bronze, respectively. At small strains, $W^\texttt{b}$ is dominated by the linear term, with $W^\texttt{b} \rightarrow 4 \sigma_y^\texttt{b} E_\texttt{max} < W^\texttt{s}$ as $E_\texttt{max} \rightarrow 0$. At larger strains, $W^\texttt{b}$ is dominated by the quadratic term, with $W^\texttt{b} \rightarrow 8 E_\texttt{t}^\texttt{b} E_\texttt{max}^2 > W^\texttt{s}$ as $E_\texttt{max} \rightarrow \infty$. Consequently, as $E_\texttt{max}$ (and $u_\texttt{max}$) increases, the optimized dampers gradually favor bronze over steel.

We also note that these findings are specific to the context of large deformations, where both the materials (Fig. \ref{Fig: Damper-Part 1}(b)) and structures (Fig. \ref{Fig: Damper-Part 2}(b)) yield immediately upon loading. Therefore, the transition from kinematic to isotropic hardening --- rather than the stiffness-to-strength transition observed in the infinitesimal strain scenario \citep{jia_multimaterial_2025} --- dominates the material distribution in the optimized dampers.

\begin{figure}[!htbp]
    \centering
    \includegraphics[height=12.5cm]{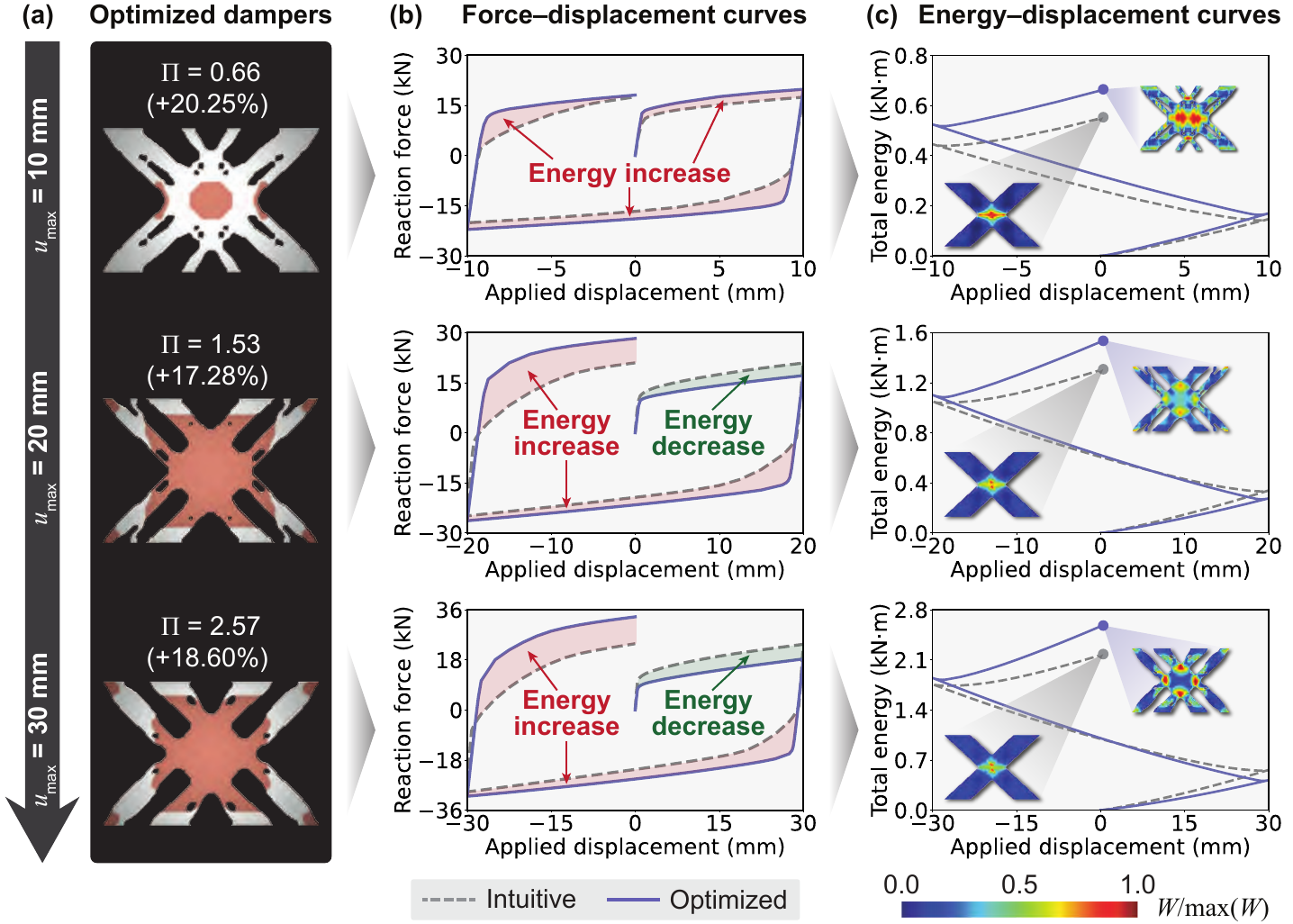}
    \caption{Optimized dampers under increasing applied displacements. (a) Optimized dampers. The total energy ($\Pi$) is in kN$\cdot$m, and the percentages are the total energy increments compared to the intuitive design in Fig. \ref{Fig: Damper-Part 1}(d). (b) Force--displacement ($F$--$u$) curves. (c) Energy--displacement ($\Pi$--$u$) curves. The insets are the normalized total energy density ($W/\max(W)$) at the final load step.}
    \label{Fig: Damper-Part 2}
\end{figure}

\subsubsection{Dampers under multiple-cycle loadings}

Practically, energy-dissipating dampers are subjected to multiple-cycle loadings with non-constant displacement amplitudes \citep{zhang_simplified_2018, jia_novel_2019, jia_double_2021, jia_residual_2022}. In this subsection, we apply such complex loading conditions and evaluate the total energy of the optimized damper. As shown in Fig. \ref{Fig: Damper-Part 3}(a), we compare the intuitive design from Fig. \ref{Fig: Damper-Part 1}(d) with a damper optimized under multiple-cycle loadings. The optimized design exhibits an organic material distribution and achieves a total energy increase of 32.84\%, rising from 4.45 kN$\cdot$m in the intuitive design to 5.91 kN$\cdot$m.

The force--displacement ($F$--$u$) and energy--displacement ($\Pi$--$u$) curves of the two designs are shown in Fig. \ref{Fig: Damper-Part 3}(b) and (c), respectively. In Fig. \ref{Fig: Damper-Part 3}(b), each force--displacement curve consists of 166 non-uniform load steps and spans three cycles with increasing displacement amplitudes of 10, 20, and 30 mm. Compared to the intuitive design, the optimized damper demonstrates a slightly smaller hysteresis loop during the first cycle but more expanded loops in the second and third cycles. Here is why.

Based on \eqref{Total Energy Density of Steel}, the total energy densities of steel during the three cycles are expressed as
\begin{equation*}
    W_1^\texttt{s} \approx 4 \sigma_y^\texttt{s} E_{1,\texttt{max}},
    \quad W_2^\texttt{s} \approx 4 \sigma_y^\texttt{s} E_{2,\texttt{max}},
    \quad \text{and} \quad
    W_3^\texttt{s} \approx 4 \sigma_y^\texttt{s} E_{3,\texttt{max}}
\end{equation*}
where $W_i^\texttt{s}$ for $i=1,2,3$ is the total energy density of steel during cycle $i$, and $E_{i,\texttt{max}}$ is the strain amplitude of that cycle. Similarly, based on \eqref{Total Energy Density of Bronze}, the total energy densities of bronze during the three cycles are
\begin{equation*}
    \left\{ \begin{array}{ll}
        W_1^\texttt{b} 
        & \approx 2 E_{1,\texttt{max}} \left[ \sigma_y^\texttt{b} + (\sigma_y^\texttt{b} + 4 E_\texttt{t}^\texttt{b} E_{1,\texttt{max}}) \right]
        = 8 E_\texttt{t}^\texttt{b} E_{1,\texttt{max}}^2 + 4 \sigma_y^\texttt{b} E_{1,\texttt{max}} \\[5pt]
        
        W_2^\texttt{b} 
        & \approx 2 E_{2,\texttt{max}} \left\{ (\sigma_y^\texttt{b} + 4 E_\texttt{t}^\texttt{b} E_{1,\texttt{max}}) + \left[ \sigma_y^\texttt{b} + 4 E_\texttt{t}^\texttt{b} (E_{1,\texttt{max}} + E_{2,\texttt{max}}) \right] \right\} \\[5pt]
        & = 8 E_\texttt{t}^\texttt{b} E_{2,\texttt{max}}^2 + 4 \sigma_y^\texttt{b} E_{2,\texttt{max}} + 16 E_\texttt{t}^\texttt{b} E_{1,\texttt{max}} E_{2,\texttt{max}} \\[5pt]

        W_3^\texttt{b} 
        & \approx 2 E_{3,\texttt{max}} \left\{ \left[ \sigma_y^\texttt{b} + 4 E_\texttt{t}^\texttt{b} (E_{1,\texttt{max}} + E_{2,\texttt{max}}) \right] + \left[ \sigma_y^\texttt{b} + 4 E_\texttt{t}^\texttt{b} (E_{1,\texttt{max}} + E_{2,\texttt{max}} + E_{3,\texttt{max}}) \right] \right\} \\[5pt]
        & = 8 E_\texttt{t}^\texttt{b} E_{3,\texttt{max}}^2 + 4 \sigma_y^\texttt{b} E_{3,\texttt{max}} + 16 E_\texttt{t}^\texttt{b} E_{1,\texttt{max}} E_{3,\texttt{max}} + 16 E_\texttt{t}^\texttt{b} E_{2,\texttt{max}} E_{3,\texttt{max}} \\[5pt]
    \end{array} \right.
\end{equation*}
where $W_i^\texttt{b}$ for $i=1,2,3$ is the total energy density of bronze during cycle $i$. Given that $E_{1,\texttt{max}} < E_{2,\texttt{max}} < E_{3,\texttt{max}}$, we observe $W_1^\texttt{s} < W_2^\texttt{s} < W_3^\texttt{s}$ and $W_1^\texttt{b} < W_2^\texttt{b} < W_3^\texttt{b}$. This progression indicates that the total energy contributions increase across successive cycles. Consequently, the optimized damper contracts the force--displacement curve and sacrifices the total energy during the first cycle. However, it compensates by expanding the force--displacement curves and increasing total energy during the second and third cycles. Eventually, the optimized damper shows greater end energy than the intuitive design (Fig. \ref{Fig: Damper-Part 3}(c)), demonstrating its superior performance under multiple-cycle loading.

\begin{figure}[!htbp]
    \centering
    \includegraphics[height=7.5cm]{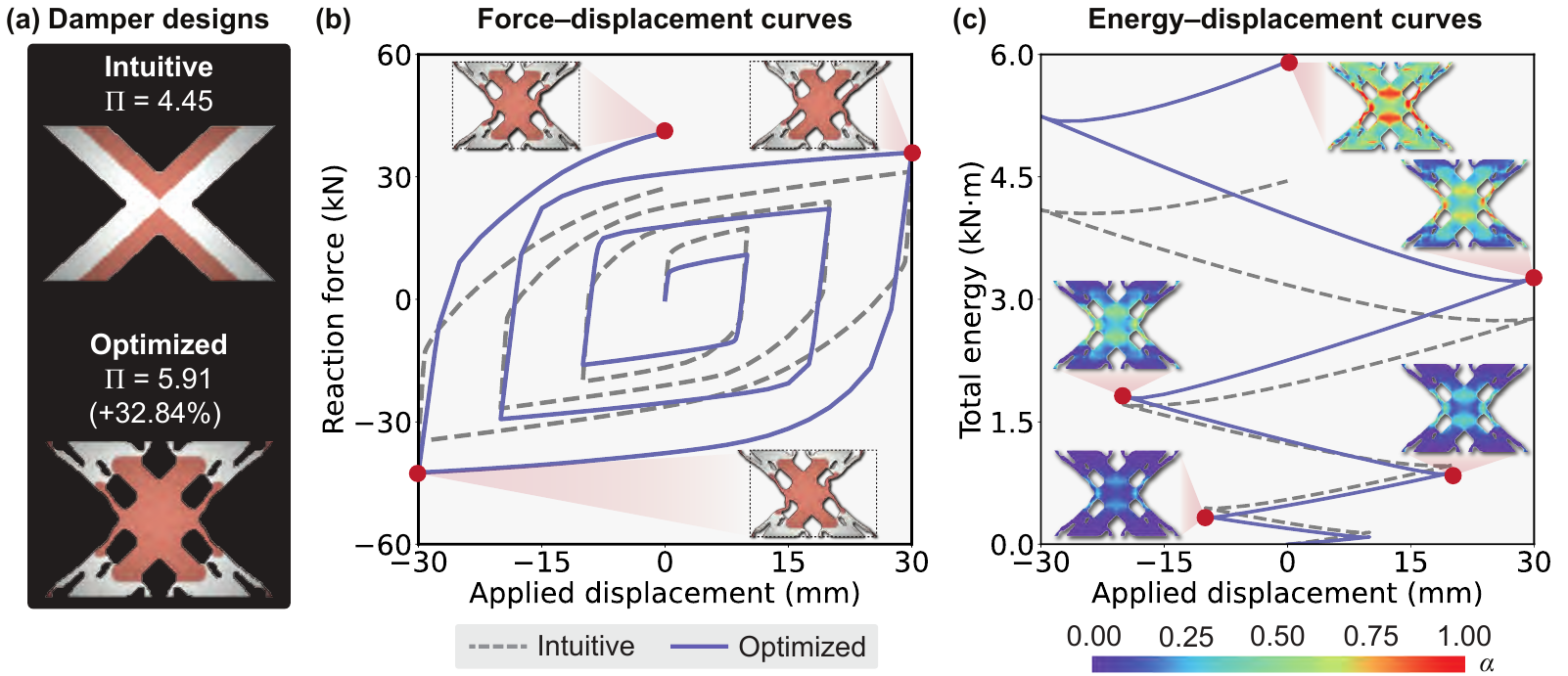}
    \caption{Dampers under multiple cycles of loadings. (a) Damper designs. The total energy ($\Pi$) is in kN$\cdot$m, and the percentage is the total energy increment compared to the intuitive design. (b) Force--displacement ($F$--$u$) curves. The insets are the deformed configurations of the optimized damper. (c) Energy--displacement ($\Pi$--$u$) curves. The insets are the equivalent plastic strains ($\alpha$), and the values above 1 are plotted as 1 for better visualization.}
    \label{Fig: Damper-Part 3}
\end{figure}

Through this example of optimizing metallic yielding dampers, we numerically prove the effectiveness of the proposed multimaterial topology optimization framework in enhancing elastoplastic responses of structures under large deformations. Exploiting this framework, we present a series of optimized dampers with superior energy dissipation compared to an intuitive design, irrespective of the applied loading conditions, including half-cycle, complete-cycle, and multiple-cycle loadings. Furthermore, we highlight the necessity of employing multiple materials to improve the energy dissipation of the optimized dampers. Our analysis also reveals key phenomena, such as the transition from kinematic to isotropic hardening under complete-cycle loadings with increasing displacement amplitudes and the dominance of later cycles over initial cycles in multiple-cycle loadings. Together with the framework for infinitesimal strain elastoplasticity in \citet{jia_multimaterial_2025}, this study completes a comprehensive narrative for optimizing energy-dissipating devices. 

\subsection{Hyperelastic--elastoplastic composite structures with tailored stiffness--strength balance}

In this example, we design and optimize composite structures composed of both hyperelastic and elastoplastic materials. These optimized designs showcase various stiffness--strength interplays by harnessing the complementary properties of elasticity and plasticity. Additionally, we demonstrate the versatility of the proposed topology optimization framework in optimizing material distributions across a range of configurations, including purely hyperelastic, purely elastoplastic, and mixed-material systems.

The design setups are shown in Fig. \ref{Fig: Beam}(a). We optimize a double-clamped beam, where the top middle edge is subjected to a downward displacement, $u$. The design objective is to distribute hyperelastic PCL and elastoplastic lithium, whose Kirchhoff stress--Lagrangian strain ($\tau_{11}$--$E_{11}$) curves are shown in Fig. \ref{Fig: Beam}(b), to simultaneously maximize the initial stiffness and end force of the structures (Fig. \ref{Fig: Beam}(c)).

\begin{figure}[!htbp]
    \centering
    \includegraphics[height=15cm]{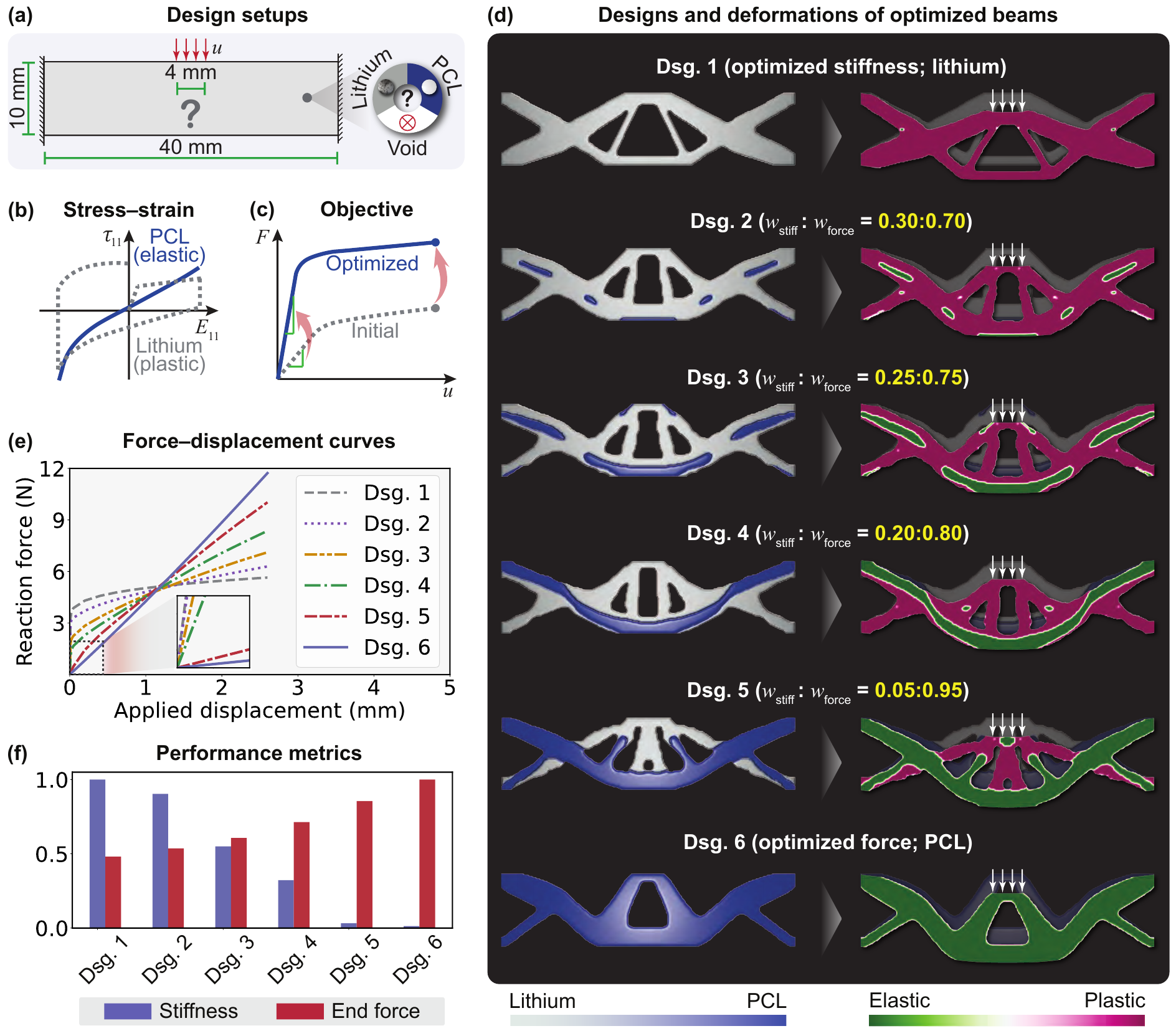}
    \caption{Design and optimization of double-clamped beams. (a) Design setups: the design domain, boundary conditions, and candidate materials (lithium, PCL, and void). The variable $u$ is the applied displacement. (b) Uniaxial Kirchhoff stress--Lagrangian strain ($\tau_{11}$--$E_{11}$) curves of candidate materials. (c) Optimization objective: maximizing the initial stiffness and end force of beams. The variable $F$ is the reaction force. (d) Designs and deformations of optimized beams. The deformed configurations also show the elastic/plastic regions. (e) Force--displacement ($F$--$u$) curves. (f) Performance metrics of the initial stiffness and end force. The two metrics are normalized by the maxima of all designs, respectively.}
    \label{Fig: Beam}
\end{figure}

To achieve the design objective, we employ the following FEA and optimization treatments. The design domain is discretized using a structured mesh of 14,400 first-order quadrilateral elements. The filter radius is set to $R_\rho=1$ mm for the density variable ($\rho$) and $R_\zeta = 3$ mm for $\zeta \in \{ \xi_1, \ldots, \xi_{N^\xi} \}$. The Heaviside sharpness parameter ($\beta_\zeta$) is initialized at 1 and doubled every 20 optimization iterations starting from iteration 21 until it reaches a maximum value of $\beta_\zeta = 512$. The penalty parameters for the density variable are fixed at $p_\kappa = p_\mu = p_h = p_k = 3$, while the penalty parameter for the material variables ($p_\xi$) starts at 1 and increases by 0.25 every 20 optimization iterations from iteration 21, up to a maximum of $p_\xi = 5$. The constraint is defined by the total material volume ($g_{V0} \leq 0$) with an upper bound of $\overline{V} = 0.5$. The maximum number of optimization iterations is set to 500. For performance evaluation using FEA, we consistently apply 18 non-uniform load steps across all designs.

The optimized beams and their deformations are shown in Fig. \ref{Fig: Beam}(d). We present four optimized bi-material designs (Dsgs. 2--5) with decreasing weighting ratios of initial stiffness ($w_\texttt{stiff}$ = 0.30, 0.25, 0.20, and 0.05, respectively) and increasing weighting ratios of end force ($w_\texttt{force}$ = 0.70, 0.75, 0.80, and 0.95, respectively). Additionally, two single-material designs are included as references: a lithium design (Dsg. 1) optimized for initial stiffness and a PCL design (Dsg. 6) optimized for the end force.

A comparison of Dsgs. 2--5 reveals a gradual shift in material preference as $w_\texttt{stiff}$ decreases and $w_\texttt{force}$ increases. The optimized designs progressively favor the hyperelastic PCL over the elastoplastic lithium. This trend is driven by the inherent properties of the materials: while metals like lithium are stiffer than polymers like PCL, lithium yields immediately under large deformations, with its Kirchhoff stress confined by the yield surface (Figs. \ref{Fig: Beam}(b) and (d)). In contrast, PCL’s hyperelastic behavior allows its Kirchhoff stress to increase rapidly under deformation. Consequently, lithium contributes more to stiffness and is preferred when $w_\texttt{stiff}$ is larger, whereas PCL contributes more to the end force and is favored when $w_\texttt{force}$ is larger.

These material contributions are further verified by the force--displacement curves in Fig. \ref{Fig: Beam}(e). As $w_\texttt{stiff}$ decreases and $w_\texttt{force}$ increases from Dsg. 2 to Dsg. 5, the initial stiffness diminishes while the peak force increases due to the greater incorporation of PCL in the optimized designs. Notably, the stiffness and peak force of the bi-material designs (Dsgs. 2--5) are bounded by the stiffness of the optimized lithium design (Dsg. 1) and the peak force of the optimized PCL design (Dsg. 6), respectively. Through these optimized designs, we demonstrate the generality of the proposed framework in optimizing structures composed of hyperelastic and/or elastoplastic materials under finite deformations. This generality enables tailoring stiffness--strength (end force) interplay for composite structures (Fig. \ref{Fig: Beam}(f)).

\subsection{Front bumpers with maximized crashworthiness}

In this subsection, we extend the optimization framework to 3D by maximizing the crashworthiness of impact-resisting front bumpers \citep{patel_crashworthiness_2009, sun_crashworthiness_2018, wang_structure_2018, ren_effective_2020, wang_multi-objective_2020} and demonstrate its capability to handle more than two candidate materials.

\subsubsection{Bi-material bumpers in 3D}

As shown in Fig. \ref{Fig: Bumper-Part 1}(a), we consider an arch-shaped design domain with four fixed corners, subjected to quasi-static displacement loading ($u$). The design objective is to optimally distribute titanium, bronze, nickel--chromium, and steel, whose Kirchhoff stress--Lagrangian strain ($\tau_{11}$--$E_{11}$) curves are shown in Fig. \ref{Fig: Bumper-Part 1}(b), to maximize the total energy --- encompassing both elastic energy absorption and plastic energy dissipation --- as shown in Fig. \ref{Fig: Bumper-Part 1}(c).

We begin by optimizing bi-material (titanium and bronze) bumpers in 3D, and FEA and optimization setups are as follows. The design domain is discretized using an unstructured mesh consisting of 138,400 first-order hexahedral elements. The filter radius is set to $R_\zeta = 40$ mm for $\zeta \in \{\rho, \xi_1, \ldots, \xi_{N^\xi} \}$. The Heaviside sharpness parameter, $\beta_\zeta$, is initially set to 1 and is doubled every 40 optimization iterations starting from iteration 41 until it reaches a maximum value of $\beta_\zeta = 512$. The penalty parameters for the density variable are fixed at $p_\kappa = p_\mu = p_h = 3$ and $p_k = 2.5$, while the penalty parameter for the material variables, $p_\xi$, starts at 1 and increases by 0.25 every 40 optimization iterations from iteration 41, up to a maximum of $p_\xi = 3$. The weighting factors in \eqref{Multi-objective Function} are set as $w_\texttt{stiff} = w_\texttt{force} = 0$ and $w_\texttt{energy} = 1$. The constraints include the total material volume ($g_{V0} \leq 0$) with an upper bound of $\overline{V} = 0.2$, and individual material volumes ($g_{V1} \leq 0$ and $g_{V2} \leq 0$) with upper bounds $\overline{V}_1 = \overline{V}_2 = 0.1$ to account for material availability. The maximum number of optimization iterations is set to 600. During the topology optimization process, we use 10 uniform load steps in the FEA. Following optimization, we refine the analysis by using 50 load steps to evaluate the elastoplastic responses of all designs.

\begin{figure}[!htbp]
    \centering
    \includegraphics[height=13.5cm]{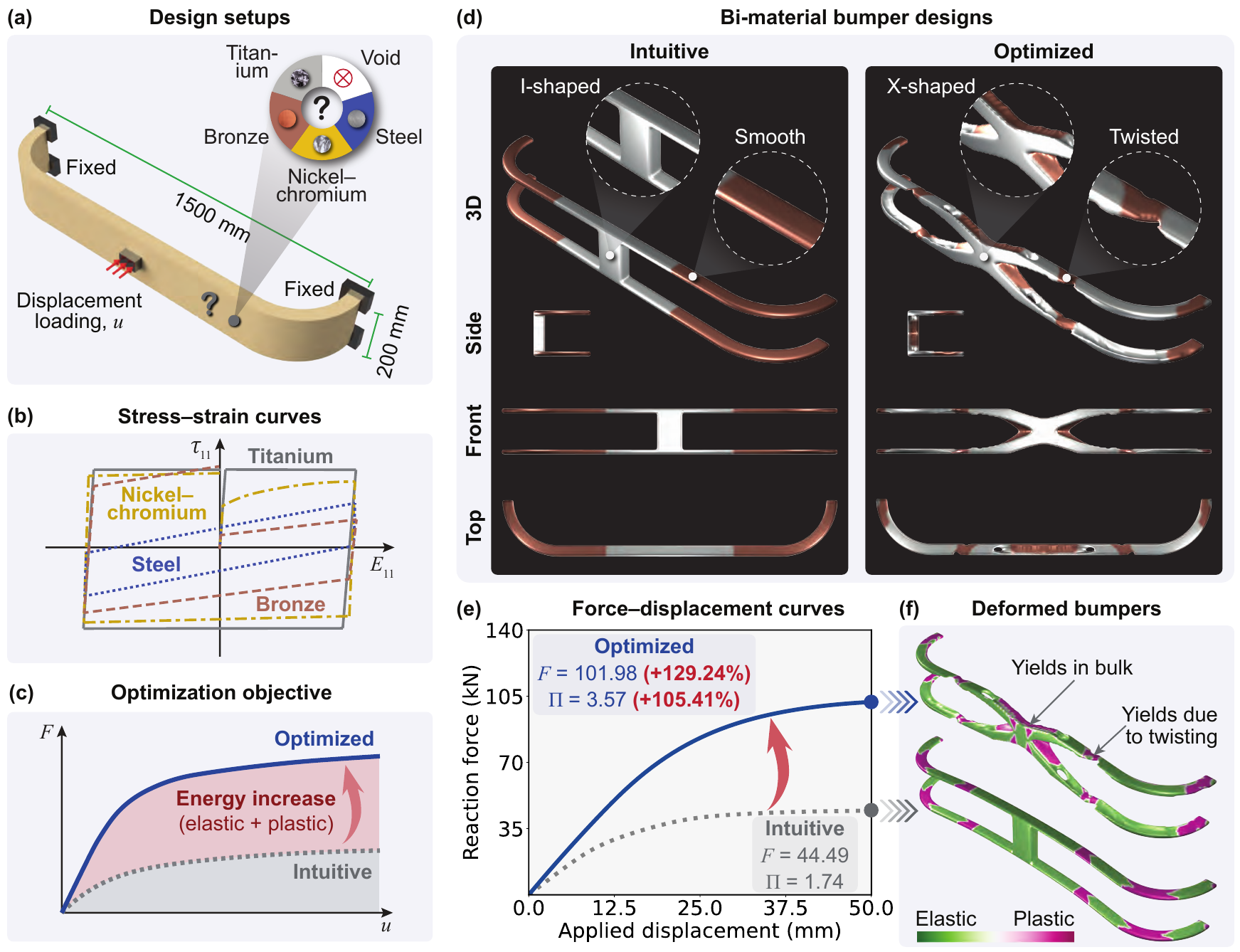}
    \caption{Design and optimization of 3D bumpers. (a) Design setups: the design domain, boundary conditions, and candidate materials (titanium, bronze, nickel--chromium, steel, and void). The variable $u$ is the applied displacement. (b) Uniaxial Kirchhoff stress--Lagrangian strain ($\tau_{11}$--$E_{11}$) curves of candidate materials. (c) Optimization objective: maximizing the total energy. The variable $F$ is the reaction force. (d) Various views of the intuitive and optimized bi-material bumper designs. (e) Force--displacement ($F$--$u$) curves. The force ($F$) is in kN, and the total energy ($\Pi$) is in kN$\cdot$m. The percentages are the improved values compared to the intuitive design. (f) Elastic/plastic regions in the deformed configuration at the final load step.}
    \label{Fig: Bumper-Part 1}
\end{figure}

The bi-material bumper designs are illustrated in Fig. \ref{Fig: Bumper-Part 1}(d). For comparison, we include an intuitive design as a reference, which shares the same total and individual material volumes as the optimized bumper. The comparison reveals that the optimized design favors an X-shaped part in the middle, in contrast to the I-shaped structure in the intuitive design. This X-shaped topology shortens load paths, thereby increasing structural stiffness and enhancing elastic energy absorption. Additionally, compared to the smooth members of the intuitive design, the optimized bumper features non-smooth, twisted regions that effectively concentrate stress, promoting localized material yielding and greater plastic energy dissipation.

Beyond the advantageous structural geometries, the optimized bumper also strategically distributes the material phases. By analyzing the material distribution in Fig. \ref{Fig: Bumper-Part 1}(d) alongside the elastic/plastic regions shown in Fig. \ref{Fig: Bumper-Part 1}(f), it is evident that the optimized design positions titanium primarily in elastic regions and bronze in plastic regions. This distribution aligns with the material properties illustrated in Fig. \ref{Fig: Bumper-Part 1}(b). Titanium, with its higher yield strength, remains in the elastic deformation regime, providing substantial elastic energy absorption. In contrast, bronze, with its lower yield strength, undergoes plastic deformation. Despite its lower initial strength, bronze exhibits isotropic hardening, allowing its Kirchhoff stress to increase progressively under loading. This behavior contrasts with the perfect plasticity (no hardening) of titanium, enabling bronze to dissipate certain plastic energy.

The combined effects of tailored structural geometries and material phases contribute to the better performance of the optimized bumper compared to the intuitive design. As shown in Fig. \ref{Fig: Bumper-Part 1}(e), the optimized bumper achieves a 105.41\% increase in total energy ($\Pi$), rising from 1.74 to 3.57 kN$\cdot$m, and a 129.24\% improvement in end force, increasing from 44.49 to 101.98 kN. These performance gains are a direct result of the simultaneous optimization of density and material variables enabled by the proposed framework.

\subsubsection{Tri-material and four-material bumpers}

After demonstrating the effectiveness of the proposed framework in optimizing bi-material bumpers in 3D, we now extend the approach to optimize bumpers with more than two candidate materials. Specifically, we optimize tri-material bumpers composed of titanium, bronze, and nickel--chromium. The constraints include the total material volume ($g_{V0} \leq 0$) with an upper bound of $\overline{V} = 0.2$, and individual material volumes ($g_{V1} \leq 0$, $g_{V2} \leq 0$, and $g_{V3} \leq 0$) with equal upper bounds of $\overline{V}_1 = \overline{V}_2 = \overline{V}_3 = \overline{V}/3$.

The intuitive and optimized tri-material bumper designs are illustrated in Fig. \ref{Fig: Bumper-Part 2}(a). Similar to the optimized bi-material design shown in Fig. \ref{Fig: Bumper-Part 1}(d), the optimized tri-material bumper retains the X-shaped structure in the middle, which shortens load paths to enhance structural stiffness. It also exhibits twisted regions that concentrate on plastic deformation. These features enhance both elastic energy absorption and plastic energy dissipation, contributing to the lifted force--displacement ($F$--$u$) curve in Fig. \ref{Fig: Bumper-Part 2}(b) and energy--displacement ($\Pi$--$u$) curve in Fig. \ref{Fig: Bumper-Part 2}(c). Ultimately, the optimized design achieves a 159.72\% increase in end force and a 160.49\% increase in total energy compared to the intuitive design.

\begin{figure}[!htbp]
    \centering
    \includegraphics[height=9.0cm]{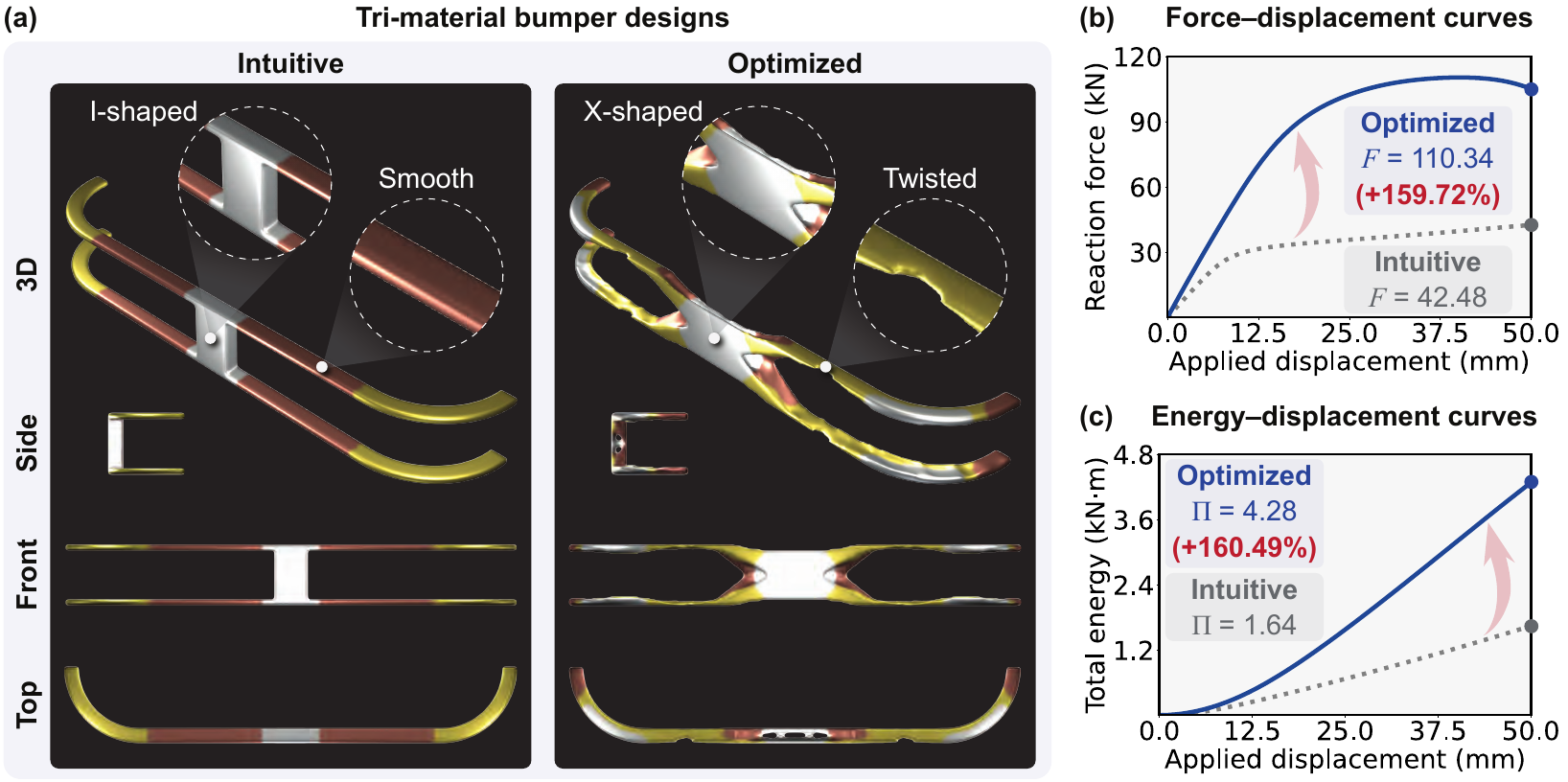}
    \caption{Design and optimization of tri-material bumpers. (a) Various views of the intuitive and optimized designs. (b)--(c) Force--displacement ($F$--$u$) and energy--displacement ($\Pi$--$u$) curves, respectively. The force ($F$) is in kN, and the total energy ($\Pi$) is in kN$\cdot$m. The percentages are the improved values compared to the intuitive design.}
    \label{Fig: Bumper-Part 2}
\end{figure}

We now optimize a four-material bumper by incorporating steel in addition to titanium, bronze, and nickel--chromium. The constraints include the total material volume ($g_{V0} \leq 0$) with an upper bound of $\overline{V} = 0.2$, and individual material volumes ($g_{V1} \leq 0$, $g_{V2} \leq 0$, $g_{V3} \leq 0$, and $g_{V4} \leq 0$) with equal upper bounds of $\overline{V}_1 = \overline{V}_2 = \overline{V}_3 = \overline{V}_4 = 0.05$. The intuitive and optimized designs are shown in Fig. \ref{Fig: Bumper-Part 3}(a), which share the same usage for each material. Compared to the intuitive design, the optimized four-material bumper features a bulky structure in the middle, connected to four twisted regions. This highly non-intuitive configuration strategically exploits material properties: bulk steel material yields around the displacement loading area, while bronze yields near the fixed ends (similar to the elastic/plastic distribution shown in Fig. \ref{Fig: Bumper-Part 1}(f)), which provide most plastic energy dissipation. Meanwhile, titanium and nickel--chromium primarily deform elastically (except in the twisted regions) due to their higher yield strengths, which mainly contribute to elastic energy absorption. Eventually, as shown in Figs. \ref{Fig: Bumper-Part 3}(b) and (c), the end force of the optimized bumper increases by 300.27\% compared to the intuitive design, while the total energy improves by 259.31\%.

Taking a broader view of the optimized bi-material, tri-material, and four-material bumpers, we draw several observations as follows. First, the proposed topology optimization framework is highly versatile, accommodating a wide range of design problems involving finite strain elastoplasticity, regardless of dimensionality, material hardening types, or the number of candidate materials. Second, as the number of materials increases and the design space expands, it is increasingly difficult to create effective elastoplastic structures based on intuition or experience alone. This is evident from the progressively larger total energy increments achieved by the optimized designs compared to the intuitive ones --- 105.41\% in the bi-material case, 160.49\% in the tri-material case, and 259.31\% in the four-material case. In contrast, the proposed topology optimization framework consistently delivers high-performance designs by leveraging rigorous mechanics-based elastoplastic analysis and gradient-based optimization, effectively navigating the complexities of large design spaces.

\begin{figure}[!htbp]
    \centering
    \includegraphics[height=9.0cm]{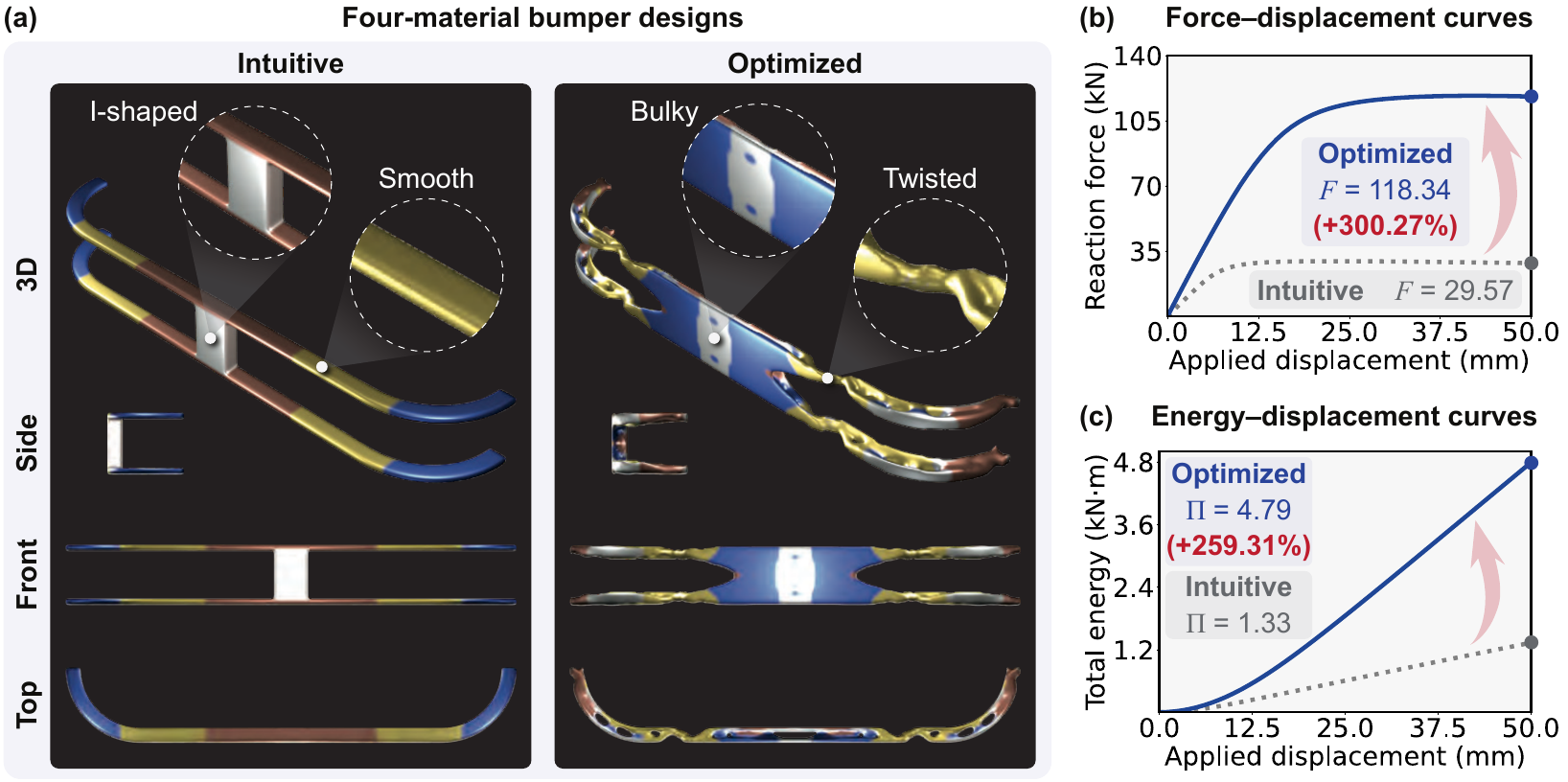}
    \caption{Design and optimization of four-material bumpers. (a) Various views of the intuitive and optimized designs. (b)--(c) Force--displacement ($F$--$u$) and energy--displacement ($\Pi$--$u$) curves, respectively. The force ($F$) is in kN, and the total energy ($\Pi$) is in kN$\cdot$m. The percentages are the improved values compared to the intuitive design.}
    \label{Fig: Bumper-Part 3}
\end{figure}

\subsection{Cold working profiled sheets with maximized load-bearing capacity and multiple engineering constraints}

In this final example, we highlight the full potential of the proposed framework in optimizing elastoplastic responses involving ultra-large deformations and complex load histories. Furthermore, we demonstrate the versatility and generality of the framework in bridging the gap between mechanical design and practical considerations such as cost, lightweight, and sustainability.

\subsubsection{Multi-stage topology optimization of profiled sheets}

To illustrate the optimization of large elastoplastic deformations and complex load histories, we design profiled sheets as depicted in Fig. \ref{Fig: Sheet-Part 1}. These sheets are manufactured by bending a flat metallic plate at room temperature, a process commonly known as cold working or metal-forming \citep{gearing_plasticity_2001, cvitanic_finite_2008}. The resulting corrugated shapes of the profiled sheets provide both mechanical (e.g., increased area moment of inertia and enhanced material strength due to strain hardening) and non-mechanical (e.g., improved rainwater drainage and better architectural aesthetics) advantages. These attributes make profiled sheets widely applicable in various domains \citep{wright_use_1987}, such as rooftops and walls in structural engineering.

To design such profiled sheets, we begin with a raw flat sheet, as shown in Fig. \ref{Fig: Sheet-Part 1}(a). Our goal is to optimally distribute titanium, bronze, nickel--chromium, and steel (material properties are provided in Table \ref{Table: Material Properties} and Fig. \ref{Fig: Bumper-Part 1}(b)) in the processing (metal-forming) stage while maximizing the end force during the service (load-carrying) stage (Fig. \ref{Fig: Sheet-Part 1}(b)).

The processing and service stages of the profiled sheets are detailed in Fig. \ref{Fig: Sheet-Part 1}(c). During the processing stage, the flat metallic sheet is bent along the short edge to achieve the desired corrugated profile. In the service stage, the sheet undergoes bending along its long edge under new boundary conditions to simulate practical usage. These stages together form the whole-life analysis of the profiled sheet. Notably, unlike traditional topology optimization problems where design variables and objectives belong to the same stage, in this case, the design variables (material distribution) are defined in the processing stage, while the design objective (maximizing load-carrying capacity) belongs to the service stage. Additionally, the irreversible nature of elastoplasticity introduces history dependence both within each stage (where only boundary values, $\overline{\mathbf{q}}$, $\overline{\mathbf{t}}$, and $\overline{\mathbf{u}}$, are updated) and across the stages (where Dirichlet and Neumann boundaries, $\partial \Omega_0^\mathcal{D}$ and $\partial \Omega_0^\mathcal{N}$, respectively, are updated).

To address this complex design problem, we propose a multi-stage topology optimization approach. Within each optimization iteration, boundary conditions are updated to account for different stages of the profiled sheet --- process and service stages; or more precisely, raw, cold-worked, undeformed, and deformed stages in Fig. \ref{Fig: Sheet-Part 1}(c). Note that the initial state variables ($\mathbf{u}$, $\overline{\mathbf{b}}^\texttt{e}$, $\overline{\boldsymbol{\beta}}$, and $\alpha$) of each stage inherit the converged values from the previous stage to capture history dependence. This multi-stage topology optimization approach is naturally accessible through the proposed framework thanks to the comprehensive history-dependent sensitivity analysis in \ref{Sec: Sensitivity Analysis and Verification}. Consequently, no modifications to the framework are required, and we proceed directly to solving the design problem using the following FEA and optimization parameters.

The design domain is discretized using a structured mesh consisting of 28,800 first-order hexahedral elements. The filter radius is set to $R_\zeta = 20$ mm for $\zeta \in \{\rho, \xi_1, \ldots, \xi_{N^\xi} \}$. The Heaviside sharpness parameter, $\beta_\zeta$, is initially set to 1 and doubles every 40 optimization iterations starting from iteration 41 until it reaches a maximum value of $\beta_\zeta = 256$. The penalty parameters for the density variable are fixed at $p_\kappa = p_\mu = p_h = p_k = 3$, while the penalty parameter for the material variables, $p_\xi$, starts at 1 and increases by 0.5 every 60 optimization iterations from iteration 41, up to a maximum of $p_\xi = 4$. The weighting factors in \eqref{Multi-objective Function} are set as $w_\texttt{stiff} = w_\texttt{energy} = 0$ and $w_\texttt{force} = 1$. The constraints include the individual material volumes ($g_{V1} \leq 0$, $g_{V2} \leq 0$, $g_{V3} \leq 0$, and $g_{V4} \leq 0$) with equal upper bounds $\overline{V}_1 = \overline{V}_2 = \overline{V}_3 = \overline{V}_4 = 0.25$ to account for material availability. The maximum number of optimization iterations is set to 400.

During the topology optimization process, 12 non-uniform load steps are used for the processing stage, 16 uniform load steps for the service stage, and 13 steps for transitioning from the processing to service stages to simulate the removal of support and load blocks and restore the sheet to a stress-free state. Following optimization, refined analyses are conducted with 40 load steps for the processing stage, 61 load steps for the service stage, and 13 load steps for the transition stage.

\begin{figure}[!htbp]
    \centering
    \includegraphics[width=17.5cm]{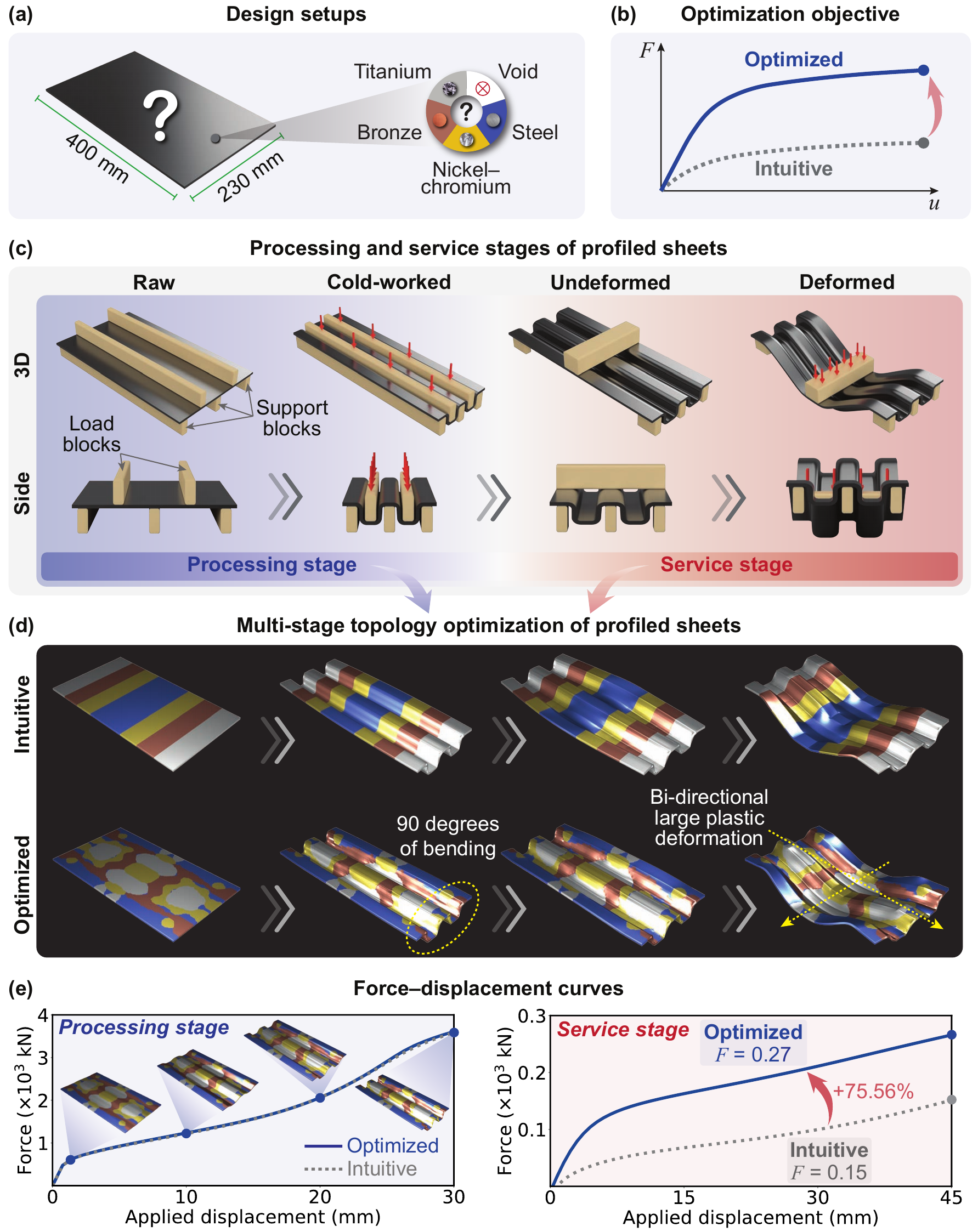}
    \caption{Design and optimization of profiled sheets. (a) Design setups: the design domain and candidate materials (titanium, bronze, nickel--chromium, steel, and void). (b) Optimization objective: maximizing the end force. The variables $F$ and $u$ are the reaction force and applied displacement, respectively. (c) Processing and service stages of profiled sheets. (d) Multi-stage topology optimization of profiled sheets. (e) Force--displacement ($F$--$u$) curves of the two stages.}
    \label{Fig: Sheet-Part 1}
\end{figure}

After multi-stage topology optimization, we present the optimized profiled sheet in Fig. \ref{Fig: Sheet-Part 1}(d), alongside an intuitive design with identical material usage for reference. For each design, four snapshots are shown, corresponding to the key stages depicted in Fig. \ref{Fig: Sheet-Part 1}(c). Unlike the serial arrangement of the four candidate materials in the intuitive design, the optimized profiled sheet features an intertwined material distribution that leverages the complementary properties of the materials shown in Fig. \ref{Fig: Bumper-Part 1}(b). Additionally, both designs exhibit large deformations in two directions, which necessitates the proposed framework to account for finite strain elastoplasticity --- an advancement beyond the reach of the infinitesimal strain version in \citet{jia_multimaterial_2025}.

The force--displacement ($F$--$u$) curves for both processing and service stages are compared in Fig. \ref{Fig: Sheet-Part 1}(e). In the processing stage, the force--displacement curves of the two designs are similar; however, in the service stage, the optimized structure exhibits a force--displacement curve that surpasses the intuitive design, achieving a 75.56\% increase in the end force. This performance improvement remarks the effectiveness of the proposed framework in tackling optimization challenges involving large elastoplastic deformations and multiple stages of loading.

\subsubsection{Optimized profiled sheets with practical constraints}

Despite the superior mechanical performance of the optimized profiled sheet endowed by the proposed framework, practical constraints --- such as cost, lightweight, and sustainability --- typically need to be integrated to create a useful mechanical product \citep{kundu2025sustainability}. In this case study, we incorporate these practical constraints into the framework, thereby bridging the gap between elastoplastic design optimization and broader real-world considerations.

To achieve this goal, we replace the individual material volume constraints ($g_{V1} \leq 0$, $g_{V2} \leq 0$, $g_{V3} \leq 0$, and $g_{V4} \leq 0$) imposed on the optimized design in Fig. \ref{Fig: Sheet-Part 2}(d) with practical constraints on price ($g_P \leq 0$), mass density ($g_M \leq 0$), and CO$_2$ footprint ($g_C \leq 0$). These constraints are defined with upper bounds set as the averages of the corresponding properties of the four candidate materials listed in Table \ref{Table: Material Properties}: $\overline{P} = 17.31$ USD/kg, $\overline{M} = 7,340$ kg/m$^3$, and $\overline{C} = 17.64$ kg/kg.

Under these updated constraints, we present the optimized profiled sheets in Fig. \ref{Fig: Sheet-Part 2}(a). Among the four optimized designs, Dsgs. 1--3 are constrained individually by price, mass density, and CO$_2$ footprint, respectively, while Dsg. 4 considers all three constraints simultaneously. These designs feature distinct material distributions (Fig. \ref{Fig: Sheet-Part 2}(a)) and usage patterns (Fig. \ref{Fig: Sheet-Part 2}(b)) as follows.
\begin{itemize}
    \item The price-constrained design (Dsg. 1) primarily incorporates titanium and steel due to their superior initial-strength-to-price ratios of 34.96 and 34.24 MPa$\cdot$kg/USD, respectively, compared to bronze (10.90 MPa$\cdot$kg/USD) and nickel--chromium (17.86 MPa$\cdot$kg/USD). A small amount of bronze appears sparingly, likely due to the strength enhancement achieved through linear isotropic hardening. On the other hand, nickel--chromium usage is minimized due to its highest absolute price.
    \item The weight-constrained design (Dsg. 2) is overwhelmingly dominated by titanium, a result of its optimal combination of the highest initial yield strength and the lowest mass density among the four materials.
    \item In contrast, the CO$_2$-constrained design (Dsg. 3) is largely composed of nickel--chromium. This material is favored because of its medium initial-strength-to-CO$_2$ ratio of 27.11 MPa$\cdot$kg/kg, which lies between the ratios of titanium (21.11 MPa$\cdot$kg/kg), bronze (24.17 MPa$\cdot$kg/kg), and steel (30.54 MPa$\cdot$kg/kg). Furthermore, with nonlinear isotropic hardening, the strength-to-CO$_2$ ratio of nickel--chromium improves to 43.07 MPa$\cdot$kg/kg, making it a suitable option.
    \item Finally, the balanced design (Dsg. 4) leverages similar amounts of all four candidate materials. This distribution reflects the fact that the upper bounds for the price, mass density, and CO$_2$ footprint constraints are averages of the properties of all four materials. As a result, the design naturally balances the contributions of titanium, bronze, nickel--chromium, and steel.
\end{itemize}

\begin{remark}
In addition to the mechanical performance improvements demonstrated in Fig. \ref{Fig: Damper-Part 1}, the incorporation of practical constraints also leads to the automatic inclusion of multiple materials in the optimized designs. This outcome further shows the necessity of a multimaterial topology optimization framework.
\end{remark}

\begin{figure}[!htbp]
    \centering
    \includegraphics[width=18.0cm]{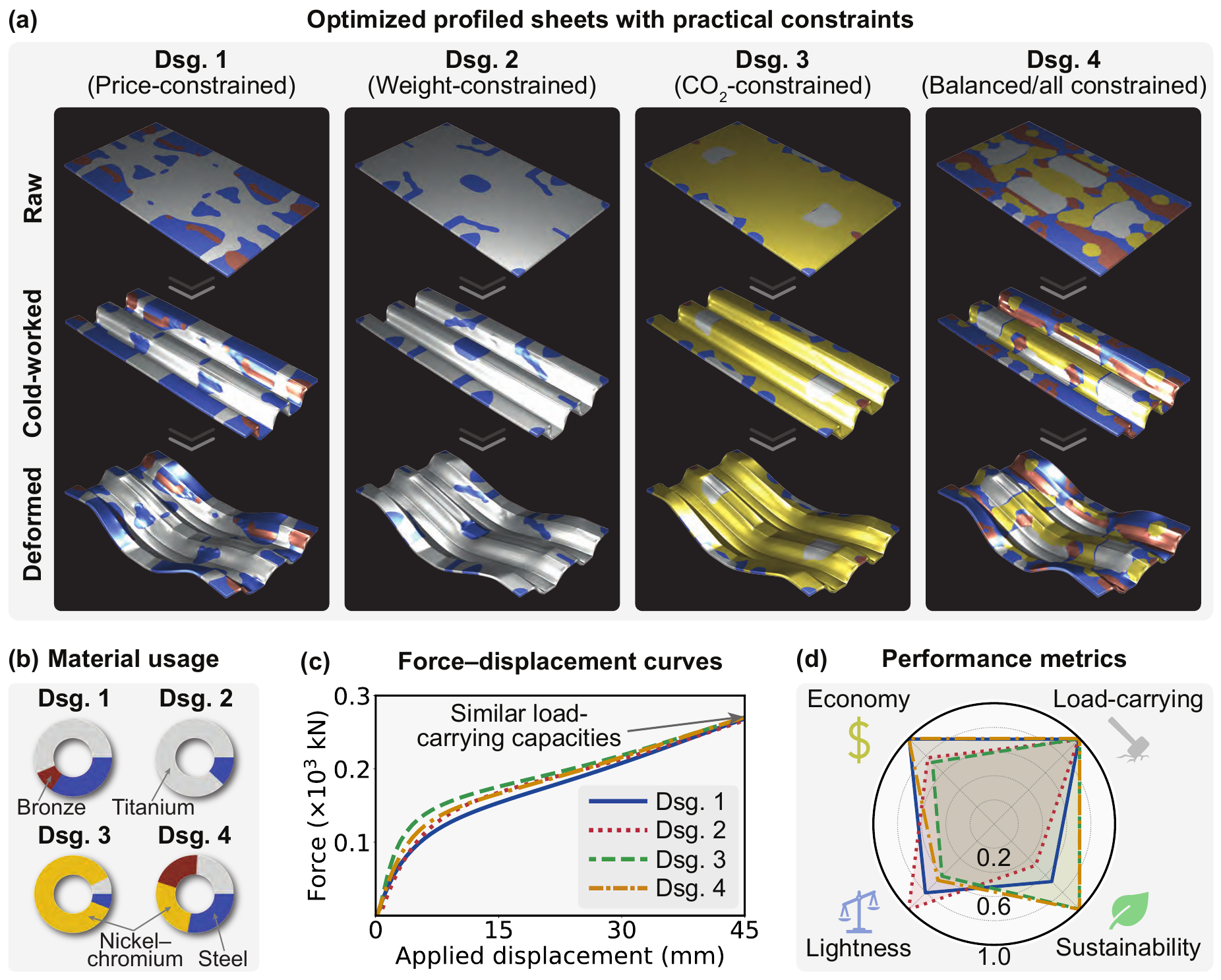}
    \caption{Multi-stage topology optimization of profiled sheets with practical constraints. (a) Optimized profiled sheets with various practical constraints. (b) Material usage of the four optimized profiled sheets. (c) Force--displacement ($F$--$u$) curves of the service stage. (d) Performance metrics of the optimized profiled sheets.}
    \label{Fig: Sheet-Part 2}
\end{figure}
 
We assess the performance of these optimized designs in Fig. \ref{Fig: Sheet-Part 2}(c) and evaluate their performance metrics in Fig. \ref{Fig: Sheet-Part 2}(d). The load-carrying capacity is measured by the end force of the structures, while economic consideration, lightweight, and sustainability are quantified as the reciprocals of the price, mass density, and CO$_2$ footprint, respectively. All metrics are normalized on a 0--1 scale for comparative purposes.

Interestingly, all four designs achieve similar load-carrying capacities (Fig. \ref{Fig: Sheet-Part 2}(c)). This result highlights the non-convex nature of the design problem, where multiple local minima exist. Therefore, practical constraints can be incorporated to tailor the non-mechanical performance of designs without compromising mechanical performance, as illustrated in Fig. \ref{Fig: Sheet-Part 2}(d). With this final generalization of the proposed framework to practical constraints, we conclude all example demonstrations.

\section{Conclusions}
\label{Sec: Conclusions}

In this study, we introduced the theory, method, and application of a multimaterial topology optimization approach for programming elastoplastic responses of structures under large deformations. The framework simultaneously determines the optimal structural geometries and material phases, leveraging a mechanics-based finite strain elastoplasticity theory that rigorously ensures isochoric plastic flow. Furthermore, the framework integrates a comprehensive path-dependent sensitivity analysis using the reversed adjoint method and automatic differentiation, enabling gradient-based optimization of design variables.

To demonstrate the effectiveness of the proposed framework, we presented four real-world application examples and uncovered the mechanisms to achieve target behaviors. First, we optimized energy-dissipating dampers, demonstrating the superior energy dissipation performance of the optimized designs compared to intuitive configurations under various loading conditions, including half-cycle, full-cycle, and multiple-cycle scenarios. This example also highlighted the transition from kinematic to isotropic hardening with increasing displacement amplitudes to maximize energy dissipation. Next, we explored the synergistic use of hyperelastic and elastoplastic materials for achieving diverse stiffness--strength interplays of double-clamped beams, showcasing the framework's versatility in handling different material types --- purely hyperelastic, purely elastoplastic, or mixed. In a further example, we extended the optimization to 3D by designing impact-resistant bumpers, while illustrating the capability to handle more than two candidate materials. Finally, we demonstrated multi-stage topology optimization for profiled sheets, focusing on maximizing load-carrying capacity under ultra-large deformations. This example also incorporated practical constraints, including cost, lightweight, and sustainability, bridging the gap between elastoplastic design and real-world considerations.

Across these examples, the proposed framework demonstrated its ability to optimize stiffness, strength, and effective structural toughness for elastoplastic structures in 2D and 3D across various spatial geometries, material types, hardening behaviors, and candidate material combinations. By fully exploiting the potential of elastoplasticity, this framework represents a step forward in designing the next generation of engineering structures. Looking ahead, we aspire to further generalize this framework to account for rate dependence and pressure dependence, which remain an ongoing focus of our research.

\section*{CRediT authorship contribution statement}
\textbf{Yingqi Jia}: Conceptualization, Methodology, Software, Validation, Formal analysis, Investigation, Data curation, Writing--original draft, Writing--review \& editing, Visualization. \textbf{Xiaojia Shelly Zhang}: Conceptualization, Methodology, Investigation, Resources, Writing--original draft, Writing--review \& editing, Supervision, Project administration, Funding acquisition.

\section*{Declaration of competing interest}
The authors declare that they have no known competing financial interests or personal relationships that could have appeared to influence the work reported in this paper.

\section*{Acknowledgments}
Authors X.S.Z. and Y.J. are grateful for the support from the U.S. Defense Advanced Research Projects Agency (DARPA) Award HR0011-24-2-0333. The information provided in this paper is the sole opinion of the authors and does not necessarily reflect the view of the sponsoring agencies.

\section*{Distribution statement}
Approved for public release; distribution is unlimited.

\section*{Data availability}
Data will be made available on request.

\appendix

\section[]{Updating formulae of $\overline{\mathbf{b}}_{n+1}^\texttt{e}$ that enforce the isochoric plastic flow}
\label{Sec: Derivation of Updating Formula of be_bar}

In this section, we introduce the updating formulae of $\mathbf{b}_{n+1}^\texttt{e}$ in \eqref{Updated be_bar} that enforce the isochoric plastic flow ($J_{n+1}^\texttt{p} = 1$). We recall the definitions of $\mathbf{b}_{n+1}^\texttt{e}$ in \eqref{Elastic Part of b} and $\mathbf{C}_{n+1}^\texttt{p}$ in \eqref{Plastic Part of C} and derive
\begin{equation*}
    \mathbf{b}_{n+1}^\texttt{e} = \mathbf{F}_{n+1} (\mathbf{C}_{n+1}^\texttt{p})^{-1} \mathbf{F}_{n+1}^\top.
\end{equation*}
The isochoric plastic flow then requires
\begin{equation} \label{Isochoric Plastic Flow}
    J_{n+1}^\texttt{p} = 1
    \quad \Longleftrightarrow \quad
    \det(\mathbf{b}_{n+1}^\texttt{e}) = J_{n+1}^2
    \quad \Longleftrightarrow \quad
    \det \left( J_{n+1}^{-2/3} \mathbf{b}_{n+1}^\texttt{e} \right) = 1
    \quad \Longleftrightarrow \quad
    \det \left( \overline{\mathbf{b}}_{n+1}^\texttt{e} \right) = 1.
\end{equation}
Additionally, we rewrite $\eqref{Local Governing Equations}_1$ as
\begin{equation*}
    \overline{\mathbf{b}}_{n+1}^\texttt{e}
    = \overline{\mathbf{b}}_{n+1}^\texttt{e,tr}
    - \dfrac{2 \overline{\overline{\mu}}_{n+1}^\texttt{tr}}{\mu} \widehat{\gamma}_{n+1} \mathbf{n}_{n+1}
\end{equation*}
based on Proposition \ref{Compute mu_bar and mu_bar_bar}. This relationship determines the deviatoric part of $\overline{\mathbf{b}}_{n+1}^\texttt{e}$ as
\begin{equation} \label{Deviatoric be_bar}
    \text{dev} \left( \overline{\mathbf{b}}_{n+1}^\texttt{e} \right) = \text{dev} \left( \overline{\mathbf{b}}_{n+1}^\texttt{e,tr} \right) - \dfrac{2 \overline{\overline{\mu}}_{n+1}^\texttt{tr}}{\mu} \widehat{\gamma}_{n+1} \mathbf{n}_{n+1}
\end{equation}
where we remark $\mathbf{n}_{n+1}$ is deviatoric based on its definition in $\eqref{Local Governing Equations}_6$. Based on \eqref{Isochoric Plastic Flow} and \eqref{Deviatoric be_bar}, the enforcement of the isochoric plastic flow amounts to identifying the volumetric part of $\overline{\mathbf{b}}_{n+1}^\texttt{e}$ --- or more specifically, its first invariant, $\mathcal{I}_1 = \text{tr} (\overline{\mathbf{b}}_{n+1}^\texttt{e})$ --- such that $\mathcal{I}_3 = \det(\overline{\mathbf{b}}_{n+1}^\texttt{e}) = 1$ for given $\text{dev} (\overline{\mathbf{b}}_{n+1}^\texttt{e})$. We now determine $\mathcal{I}_1$ by extending the analysis in \citet{simo_associative_1992}.

Note that the invariants $\mathcal{I}_1$, $\mathcal{I}_3$, $\mathcal{J}_2$, and $\mathcal{J}_3$ of a symmetric second-order tensor satisfy
\begin{equation*}
    \dfrac{\mathcal{I}_1^3}{27} - \dfrac{\mathcal{I}_1 \mathcal{J}_2}{3} + \mathcal{J}_3 - \mathcal{I}_3 = 0
\end{equation*}
where $\mathcal{J}_2$ and $\mathcal{J}_3$ are defined as
\begin{equation*}
    \mathcal{J}_2 = \dfrac{1}{2} \left\lVert \text{dev} \left( \overline{\mathbf{b}}_{n+1}^\texttt{e} \right) \right\rVert^2
    \quad \text{and} \quad
    \mathcal{J}_3 = \det \left[ \text{dev} \left( \overline{\mathbf{b}}_{n+1}^\texttt{e} \right) \right].
\end{equation*}
Consequently, a necessary condition of $\mathcal{I}_3 = \det(\overline{\mathbf{b}}_{n+1}^\texttt{e}) = 1$ reads as
\begin{equation} \label{Cubic Equation}
    t^3 + \mathcal{P} t + \mathcal{Q} = 0
\end{equation}
where
\begin{equation*}
    t = \dfrac{\mathcal{I}_1}{3} > 0, \quad
    \mathcal{P} = - \mathcal{J}_2 \leq 0, \quad \text{and} \quad
    \mathcal{Q} = \mathcal{J}_3 - 1.
\end{equation*}
Our objective then boils down to compute the positive real root from the depressed cubic equation in \eqref{Cubic Equation}.

Following a standard procedure, we compute the discriminant of \eqref{Cubic Equation} as
\begin{equation*}
    \Delta = - \left( \dfrac{\mathcal{P}^3}{27} + \dfrac{\mathcal{Q}^2}{4} \right)
\end{equation*}
and discuss the solutions of \eqref{Cubic Equation} based on the sign of $\Delta$ as follows.
\begin{itemize}
    \item The case of $\Delta < 0$. The equation in \eqref{Cubic Equation} has one positive real root and two non-real complex conjugate roots. We should take the real root as
    \begin{equation*}
        t_1 = \left( -\dfrac{\mathcal{Q}}{2} + \sqrt{-\Delta} \right)^{1/3} + \left( -\dfrac{\mathcal{Q}}{2} - \sqrt{-\Delta} \right)^{1/3}.
    \end{equation*}

    \item The case of $\Delta = 0$. The equation in \eqref{Cubic Equation} has a triple root of zero (when $\mathcal{P} = 0$) or the combination of a simple root of $3\mathcal{Q} / \mathcal{P}$ and a double root of $- 3\mathcal{Q} / 2\mathcal{P}$ (when $\mathcal{P} < 0$). We should take the only positive real root as
    \begin{equation*}
        t_2 = \max \left\{ \dfrac{3 \mathcal{Q}}{\mathcal{P}}, - \dfrac{3 \mathcal{Q}}{2 \mathcal{P}} \right\}.
    \end{equation*}

    \item The case of $\Delta > 0$. The equation in \eqref{Cubic Equation} has three distinct real roots expressed in trigonometric forms as shown in \eqref{Trigonometric Solutions}. Based on Descartes' rule of signs, only one solution among $r_1$, $r_2$, and $r_3$ is positive when $\mathcal{Q} \leq 0$, and we should set the solution to \eqref{Cubic Equation} as
    \begin{equation*}
        t_3 = \max\{ r_1, r_2, r_3\}.
    \end{equation*}
    Otherwise ($\mathcal{Q} > 0$), two solutions among $r_1$, $r_2$, and $r_3$ are positive, and we should select the solution such that $\overline{\mathbf{b}}_{n+1}^\texttt{e}$ is positive definite by ensuring
    \begin{equation*}
        \left( \overline{\mathbf{b}}_{n+1}^\texttt{e} \right)_{11} > 0
        \quad \text{and} \quad
        \left( \overline{\mathbf{b}}_{n+1}^\texttt{e} \right)_{11} \left( \overline{\mathbf{b}}_{n+1}^\texttt{e} \right)_{22} - \left( \overline{\mathbf{b}}_{n+1}^\texttt{e} \right)_{12} \left( \overline{\mathbf{b}}_{n+1}^\texttt{e} \right)_{21} > 0
    \end{equation*}
    based on the Sylvester’s criterion.
\end{itemize}

Finally, combining the solutions from all scenarios ($\Delta < 0$, $\Delta = 0$, and $\Delta > 0$), we derive the updating formulae of $\overline{\mathbf{b}}_{n+1}^\texttt{e}$ in \eqref{Updated be_bar}. We remark that \citet{simo_associative_1992} only consider the scenario of $\Delta < 0$ by assuming $\mathcal{J}_2 \ll |1-\mathcal{J}_3|$. Here we remove this assumption and incorporate the scenarios of $\Delta = 0$ and $\Delta > 0$ to account for a broader type of material responses.

\section{Derivation of the second elastoplastic moduli}
\label{Sec: Second Elastoplastic Moduli}

This section derives the second elastoplastic moduli in \eqref{Second Elastoplastic Moduli Final}. Taking the derivative of the second Piola--Kirchhoff stress ($\mathbf{S}_{n+1}$) in \eqref{Second PK Stress} with respect to the Lagrangian strain tensor ($\mathbf{E}_{n+1}$) in $\eqref{Finite Strain Tensors}$ yields
\begin{equation} \label{Second Elastoplastic Moduli Initial}
    \begin{array}{ll}
        \mathbb{C}_{n+1}^\texttt{ep}
        = \dfrac{\partial \mathbf{S}_{n+1}}{\partial \mathbf{E}_{n+1}}
        = 2 \dfrac{\partial \mathbf{S}_{n+1}}{\partial \mathbf{C}_{n+1}}
        = & 2 \dfrac{\partial \varphi^*(\boldsymbol{\tau}_{n+1}^\texttt{vol})}{\partial \mathbf{C}_{n+1}}
        + 2 \dfrac{\partial \varphi^*(\mathbf{s}_{n+1}^\texttt{tr})}{\partial \mathbf{C}_{n+1}}
        - 4 \widehat{\gamma}_{n+1} \varphi^*(\mathbf{n}_{n+1}) \otimes \dfrac{\partial \overline{\overline{\mu}}_{n+1}^\texttt{tr}}{\partial \mathbf{C}_{n+1}} \\[12pt]

        &- 4 \overline{\overline{\mu}}_{n+1}^\texttt{tr} \varphi^* (\mathbf{n}_{n+1}) \otimes \dfrac{\partial \widehat{\gamma}_{n+1}}{\partial \mathbf{C}_{n+1}}
        - 4 \overline{\overline{\mu}}_{n+1}^\texttt{tr} \widehat{\gamma}_{n+1} \dfrac{\partial \varphi^*(\mathbf{n}_{n+1})}{\partial \mathbf{C}_{n+1}}.
    \end{array}
\end{equation}
One immediate next step is to compute the partial derivatives involved in $\mathbb{C}_{n+1}^\texttt{ep}$, which are laid out below.

\subsection[]{The computation of $\partial \varphi^*(\boldsymbol{\tau}_{n+1}^\texttt{vol}) / \partial \mathbf{C}_{n+1}$}

Note that
\begin{equation} \label{Partial tau Partial C Initial}
    \dfrac{\partial \varphi^*(\boldsymbol{\tau}_{n+1}^\texttt{vol})}{\partial \mathbf{C}_{n+1}}
    = \dfrac{\partial \left( J_{n+1} U'_{n+1} \mathbf{C}_{n+1}^{-1} \right)}{\partial \mathbf{C}_{n+1}}
    = \mathbf{C}_{n+1}^{-1} \otimes \dfrac{\partial (J_{n+1} U'_{n+1})}{\partial \mathbf{C}_{n+1}} 
    + (J_{n+1} U'_{n+1}) \dfrac{\partial \mathbf{C}_{n+1}^{-1}}{\partial \mathbf{C}_{n+1}}.
\end{equation}
Substituting
\begin{equation*}
    \dfrac{\partial (J_{n+1} U'_{n+1})}{\partial \mathbf{C}_{n+1}}
    = \dfrac{\partial (J_{n+1} U'_{n+1})}{\partial J_{n+1}}
    \dfrac{\partial J_{n+1}}{\partial \mathbf{C}_{n+1}}
    = (J_{n+1} U'_{n+1})' \left( \dfrac{1}{2} J_{n+1} \mathbf{C}_{n+1}^{-1} \right)
\end{equation*}
and
\begin{equation*}
    \dfrac{\partial \mathbf{C}_{n+1}^{-1}}{\partial \mathbf{C}_{n+1}}
    = - \mathbb{I}_{\mathbf{C}_{n+1}^{-1}}
    \quad \text{with} \quad
    \left( \mathbb{I}_{\mathbf{C}_{n+1}^{-1}} \right)_{ijkl} = \dfrac{1}{2} \left[ \left( C_{n+1}^{-1} \right)_{ik} \left( C_{n+1}^{-1} \right)_{jl} + \left( C_{n+1}^{-1} \right)_{il} \left( C_{n+1}^{-1} \right)_{jk} \right]
\end{equation*}
into \eqref{Partial tau Partial C Initial} renders
\begin{equation} \label{Partial tau Partial C Final}
    \dfrac{\partial \varphi^*(\boldsymbol{\tau}_{n+1}^\texttt{vol})}{\partial \mathbf{C}_{n+1}}
    = \dfrac{1}{2} J_{n+1} (J_{n+1} U'_{n+1})' \mathbf{C}_{n+1}^{-1} \otimes \mathbf{C}_{n+1}^{-1}
    - J_{n+1} U'_{n+1} \mathbb{I}_{\mathbf{C}_{n+1}^{-1}}.
\end{equation}

\subsection[]{The computation of $\partial \varphi^*(\mathbf{s}_{n+1}^\texttt{tr}) / \partial \mathbf{C}_{n+1}$} 

Note that
\begin{equation} \label{be-bar Trial}
    \begin{array}{ll}
        \overline{\mathbf{b}}_{n+1}^\texttt{e,tr}
        &= \overline{\mathbf{f}}_{n+1} \overline{\mathbf{b}}_n^\texttt{e} \overline{\mathbf{f}}_{n+1}^\top
        = \left[ \left( J_{n+1}^f \right)^{-1/3} \mathbf{f}_{n+1} \right]
        \left[ \left( J_{n}^\texttt{e} \right)^{-2/3} \mathbf{b}_n^\texttt{e} \right]
        \left[ \left( J_{n+1}^f \right)^{-1/3} \mathbf{f}_{n+1}^\top \right] \\[10pt]

        &= \left( \dfrac{J_{n+1} J_n^\texttt{e}}{J_n} \right)^{-2/3} \left( \mathbf{F}_{n+1} \mathbf{F}_n^{-1} \right) 
        \left( \mathbf{F}_n^\texttt{e} (\mathbf{F}_n^\texttt{e})^\top \right)
        \left( \mathbf{F}_n^{-\top}  \mathbf{F}_{n+1}^\top \right) \\[10pt]

        & = J_{n+1}^{-2/3} \mathbf{F}_{n+1} \left( \mathbf{F}_n^\texttt{p} \right)^{-1} \left( \mathbf{F}_n^\texttt{p} \right)^{-\top} \mathbf{F}_{n+1}^\top
        = J_{n+1}^{-2/3} \mathbf{F}_{n+1} \left( \mathbf{C}_n^\texttt{p} \right)^{-1} \mathbf{F}_{n+1}^\top
    \end{array}
\end{equation}
due to $J_n = J_n^\texttt{e}$ (assuming isochoric plastic flow) and $J_{n+1} = J_{n+1}^f J_n$. We recall
\begin{equation*}
    \begin{array}{ll}
        \mathbf{s}_{n+1}^\texttt{tr} &= \mu\ \text{dev} \left( \overline{\mathbf{b}}_{n+1}^\texttt{e,tr} \right)
        = \mu J_{n+1}^{-2/3} \text{dev} \left[ \mathbf{F}_{n+1} \left( \mathbf{C}_n^\texttt{p} \right)^{-1} \mathbf{F}_{n+1}^\top \right] \\[10pt]
        &= \mu J_{n+1}^{-2/3} \left\{ \mathbf{F}_{n+1} \left( \mathbf{C}_n^\texttt{p} \right)^{-1} \mathbf{F}_{n+1}^\top
        - \dfrac{1}{3} \text{tr} \left[ \mathbf{F}_{n+1} \left( \mathbf{C}_n^\texttt{p} \right)^{-1} \mathbf{F}_{n+1}^\top \right] \mathbf{I} \right\} \\[10pt]
        &= \mu J_{n+1}^{-2/3} \left\{ \mathbf{F}_{n+1} \left( \mathbf{C}_n^\texttt{p} \right)^{-1} \mathbf{F}_{n+1}^\top
        - \dfrac{1}{3} \left[ \left( \mathbf{C}_n^\texttt{p} \right)^{-1} : \mathbf{C}_{n+1} \right] \mathbf{I} \right\}
    \end{array}
\end{equation*}
and derive
\begin{equation*}
    \varphi^*(\mathbf{s}_{n+1}^\texttt{tr})
    = \mathbf{F}_{n+1}^{-1} \mathbf{s}_{n+1}^\texttt{tr} \mathbf{F}_{n+1}^{-\top}
    = \mu J_{n+1}^{-2/3} \left\{ \left( \mathbf{C}_n^\texttt{p} \right)^{-1}
    - \dfrac{1}{3} \left[ \left( \mathbf{C}_n^\texttt{p} \right)^{-1} : \mathbf{C}_{n+1} \right] \mathbf{C}_{n+1}^{-1} \right\}.
\end{equation*}

Next, we compute
\begin{equation} \label{Partial s-trial-forward Partial C Initial}
    \dfrac{\partial \varphi^*(\mathbf{s}_{n+1}^\texttt{tr})}{\partial \mathbf{C}_{n+1}}
    = \mu \left\{ \left( \mathbf{C}_n^\texttt{p} \right)^{-1}
    - \dfrac{1}{3} \left[ \left( \mathbf{C}_n^\texttt{p} \right)^{-1} : \mathbf{C}_{n+1} \right] \mathbf{C}_{n+1}^{-1} \right\} \dfrac{\partial J_{n+1}^{-2/3}}{\partial \mathbf{C}_{n+1}}
    - \dfrac{\mu}{3} J_{n+1}^{-2/3} \dfrac{\partial}{\partial \mathbf{C}_{n+1}} \left\{ \left[ \left( \mathbf{C}_n^\texttt{p} \right)^{-1} : \mathbf{C}_{n+1} \right] \mathbf{C}_{n+1}^{-1} \right\}.
\end{equation}
Substituting
\begin{equation} \label{Partial J^(-2/3) Partial C}
    \dfrac{\partial J_{n+1}^{-2/3}}{\partial \mathbf{C}_{n+1}}
    = - \dfrac{2}{3} J_{n+1}^{-5/3} \dfrac{\partial J_{n+1}}{\partial \mathbf{C}_{n+1}}
    = - \dfrac{2}{3} J_{n+1}^{-5/3} \left( \dfrac{1}{2} J_{n+1} \mathbf{C}_{n+1}^{-1} \right)
    = - \dfrac{1}{3} J_{n+1}^{-2/3} \mathbf{C}_{n+1}^{-1}
\end{equation}
and
\begin{equation*}
    \begin{array}{ll}
        \dfrac{\partial}{\partial \mathbf{C}_{n+1}} \left\{ \left[ \left( \mathbf{C}_n^\texttt{p} \right)^{-1} : \mathbf{C}_{n+1} \right] \mathbf{C}_{n+1}^{-1} \right\}
        &= \mathbf{C}_{n+1}^{-1} \otimes \dfrac{\partial \left[ \left( \mathbf{C}_n^\texttt{p} \right)^{-1} : \mathbf{C}_{n+1} \right]}{\partial \mathbf{C}_{n+1}}
        + \left[ \left( \mathbf{C}_n^\texttt{p} \right)^{-1} : \mathbf{C}_{n+1} \right] \dfrac{\partial \mathbf{C}_{n+1}^{-1}}{\partial \mathbf{C}_{n+1}} \\[12pt]
        & = \mathbf{C}_{n+1}^{-1} \otimes \left( \mathbf{C}_n^\texttt{p} \right)^{-1}
        - \left[ \left( \mathbf{C}_n^\texttt{p} \right)^{-1} : \mathbf{C}_{n+1} \right] \mathbb{I}_{\mathbf{C}_{n+1}^{-1}}
    \end{array}
\end{equation*}
into \eqref{Partial s-trial-forward Partial C Initial} renders
\begin{equation*}
    \begin{array}{ll}
        \dfrac{\partial \varphi^*(\mathbf{s}_{n+1}^\texttt{tr})}{\partial \mathbf{C}_{n+1}}
        = &- \dfrac{\mu}{3} J_{n+1}^{-2/3} \left\{ \left( \mathbf{C}_n^\texttt{p} \right)^{-1}
        - \dfrac{1}{3} \left[ \left( \mathbf{C}_n^\texttt{p} \right)^{-1} : \mathbf{C}_{n+1} \right] \mathbf{C}_{n+1}^{-1} \right\} \otimes \mathbf{C}_{n+1}^{-1} \\[12pt]

         & - \dfrac{\mu}{3} J_{n+1}^{-2/3}
        \left\{ \mathbf{C}_{n+1}^{-1} \otimes \left( \mathbf{C}_n^\texttt{p} \right)^{-1}
        - \left[ \left( \mathbf{C}_n^\texttt{p} \right)^{-1} : \mathbf{C}_{n+1} \right] \mathbb{I}_{\mathbf{C}_{n+1}^{-1}} \right\}.
    \end{array}
\end{equation*}
Regrouping terms yields
\begin{equation} \label{Partial s-trial-forward Partial C Middle}
    \dfrac{\partial \varphi^*(\mathbf{s}_{n+1}^\texttt{tr})}{\partial \mathbf{C}_{n+1}}
    = - \dfrac{\mu}{3} J_{n+1}^{-2/3} \left\{ \left( \mathbf{C}_n^\texttt{p} \right)^{-1} \otimes \mathbf{C}_{n+1}^{-1} + \mathbf{C}_{n+1}^{-1} \otimes \left( \mathbf{C}_n^\texttt{p} \right)^{-1}
    - \left[ \left( \mathbf{C}_n^\texttt{p} \right)^{-1} : \mathbf{C}_{n+1} \right]
    \left( \mathbb{I}_{\mathbf{C}_{n+1}^{-1}} + \dfrac{1}{3} \mathbf{C}_{n+1}^{-1} \otimes \mathbf{C}_{n+1}^{-1} \right) \right\}.
\end{equation}

Based on \eqref{be-bar Trial}, we derive $J_{n+1}^{-2/3} \left( \mathbf{C}_n^\texttt{p} \right)^{-1} = \varphi^* \left( \overline{\mathbf{b}}_{n+1}^\texttt{tr} \right)$ and further compute
\begin{equation*}
    \left\{ \begin{array}{l}
        \mu \varphi^* \left( \overline{\mathbf{b}}_{n+1}^\texttt{tr} \right) \otimes \mathbf{C}_{n+1}^{-1}
        = \varphi^* \left( \mathbf{s}_{n+1}^\texttt{tr} \right) \otimes \mathbf{C}_{n+1}^{-1}
        + \overline{\mu}_{n+1}^\texttt{tr} \mathbf{C}_{n+1}^{-1} \otimes \mathbf{C}_{n+1}^{-1}, \\[12pt]
        
        \mu \mathbf{C}_{n+1}^{-1} \otimes \varphi^* \left( \overline{\mathbf{b}}_{n+1}^\texttt{tr} \right)
        = \mathbf{C}_{n+1}^{-1} \otimes \varphi^* \left( \mathbf{s}_{n+1}^\texttt{tr} \right)
        + \overline{\mu}_{n+1}^\texttt{tr} \mathbf{C}_{n+1}^{-1} \otimes \mathbf{C}_{n+1}^{-1}, \\[12pt]
        
        \dfrac{\mu}{3} \varphi^* \left( \overline{\mathbf{b}}_{n+1}^\texttt{tr} \right) : \mathbf{C}_{n+1} = \dfrac{\mu}{3} \text{tr} \left( \overline{\mathbf{b}}_{n+1}^\texttt{tr} \right) = \overline{\mu}_{n+1}^\texttt{tr}.
    \end{array} \right.
\end{equation*}
We then reduce \eqref{Partial s-trial-forward Partial C Middle} to
\begin{equation} \label{Partial s-trial-forward Partial C Final}
    \dfrac{\partial \varphi^*(\mathbf{s}_{n+1}^\texttt{tr})}{\partial \mathbf{C}_{n+1}}
    = \overline{\mu}_{n+1}^\texttt{tr} \left( \mathbb{I}_{\mathbf{C}_{n+1}^{-1}} - \dfrac{1}{3} \mathbf{C}_{n+1}^{-1} \otimes \mathbf{C}_{n+1}^{-1} \right)
    - \dfrac{1}{3} \left[ \varphi^* \left( \mathbf{s}_{n+1}^\texttt{tr} \right) \otimes \mathbf{C}_{n+1}^{-1}
    + \mathbf{C}_{n+1}^{-1} \otimes \varphi^* \left( \mathbf{s}_{n+1}^\texttt{tr} \right) \right].
\end{equation}

\subsection[]{The computation of $\partial \overline{\overline{\mu}}_{n+1}^\texttt{tr} / \partial \mathbf{C}_{n+1}$} 

Note that
\begin{equation} \label{Partial mu-bar-bar Partial C}
    \dfrac{\partial \overline{\overline{\mu}}_{n+1}^\texttt{tr}}{\partial \mathbf{C}_{n+1}} = \dfrac{\partial \overline{\mu}_{n+1}^\texttt{tr}}{\partial \mathbf{C}_{n+1}}
    - \dfrac{1}{3} \dfrac{\partial}{\partial \mathbf{C}_{n+1}}
    \text{tr} \left( \overline{\mathbf{f}}_{n+1} \overline{\boldsymbol{\beta}}_n \overline{\mathbf{f}}_{n+1}^\top \right).
\end{equation}
According to \eqref{be-bar Trial}, we derive
\begin{equation*}
    \overline{\mu}_{n+1}^\texttt{tr}
    = \dfrac{\mu}{3} \text{tr} \left( \overline{\mathbf{b}}_{n+1}^\texttt{e,tr} \right)
    = \dfrac{\mu}{3} J_{n+1}^{-2/3} \text{tr} \left( \mathbf{F}_{n+1} \left( \mathbf{C}_n^\texttt{p} \right)^{-1} \mathbf{F}_{n+1}^\top \right)
    = \dfrac{\mu}{3} J_{n+1}^{-2/3} \left( \mathbf{C}_n^\texttt{p} \right)^{-1} : \mathbf{C}_{n+1}.
\end{equation*}
Taking the derivative of $\overline{\mu}_{n+1}^\texttt{tr}$ yields
\begin{equation} \label{Partial mu-bar Partial C}
    \begin{array}{ll}
        \dfrac{\partial \overline{\mu}_{n+1}^\texttt{tr}}{\partial \mathbf{C}_{n+1}}
        & = \dfrac{\mu}{3} \left[ \left( \mathbf{C}_n^\texttt{p} \right)^{-1} : \mathbf{C}_{n+1} \right] \dfrac{\partial J_{n+1}^{-2/3}}{\partial \mathbf{C}_{n+1}}
        + \dfrac{\mu}{3} J_{n+1}^{-2/3} \dfrac{\partial}{\partial \mathbf{C}_{n+1}} \left[ \left( \mathbf{C}_n^\texttt{p} \right)^{-1} : \mathbf{C}_{n+1} \right] \\[12pt]
        & = \dfrac{\mu}{3} \left[ \left( \mathbf{C}_n^\texttt{p} \right)^{-1} : \mathbf{C}_{n+1} \right]
        \left( - \dfrac{1}{3} J_{n+1}^{-2/3} \mathbf{C}_{n+1}^{-1} \right)
        + \dfrac{\mu}{3} J_{n+1}^{-2/3} \left( \mathbf{C}_n^\texttt{p} \right)^{-1} \\[12pt]
        &= \dfrac{\mu}{3} J_{n+1}^{-2/3}
        \left\{ \left( \mathbf{C}_n^\texttt{p} \right)^{-1}
        - \dfrac{1}{3} \left[ \left( \mathbf{C}_n^\texttt{p} \right)^{-1} : \mathbf{C}_{n+1} \right] \mathbf{C}_{n+1}^{-1} \right\}
    \end{array}
\end{equation}
based on \eqref{Partial J^(-2/3) Partial C}.

Additionally, we rewrite
\begin{equation} \label{J^(-2/3) F Gamma F}
    \begin{array}{ll}
        \overline{\mathbf{f}}_{n+1} \overline{\boldsymbol{\beta}}_n \overline{\mathbf{f}}_{n+1}^\top &= \left[ \left( J_{n+1}^f \right)^{-1/3} \mathbf{f}_{n+1} \right] \left( J_n^{-2/3} \boldsymbol{\beta}_n \right) \left[ \left( J_{n+1}^f \right)^{-1/3} \mathbf{f}_{n+1}^\top \right] \\[12pt]

        & = J_{n+1}^{-2/3} \left( \mathbf{F}_{n+1} \mathbf{F}_n^{-1} \right)
        \boldsymbol{\beta}_n \left( \mathbf{F}_n^{-\top} \mathbf{F}_{n+1}^\top \right) = J_{n+1}^{-2/3} \mathbf{F}_{n+1} \boldsymbol{\Gamma}_n \mathbf{F}_{n+1}^\top
    \end{array}
\end{equation}
where we define $\boldsymbol{\Gamma}_n = \mathbf{F}_n^{-1} \boldsymbol{\beta}_n \mathbf{F}_n^{-\top}$. Next, we derive
\begin{equation*}
    \text{tr}\left( \overline{\mathbf{f}}_{n+1} \overline{\boldsymbol{\beta}}_n \overline{\mathbf{f}}_{n+1}^\top \right)
    = J_{n+1}^{-2/3} \boldsymbol{\Gamma}_n : \mathbf{C}_{n+1}
\end{equation*}
and
\begin{equation} \label{Partial f-beta-f Partial C}
    \dfrac{\partial}{\partial \mathbf{C}_{n+1}}
    \text{tr} \left( \overline{\mathbf{f}}_{n+1} \overline{\boldsymbol{\beta}}_n \overline{\mathbf{f}}_{n+1}^\top \right) = 
    J_{n+1}^{-2/3}
    \left[ \boldsymbol{\Gamma}_n
    - \dfrac{1}{3} \left( \boldsymbol{\Gamma}_n : \mathbf{C}_{n+1} \right) \mathbf{C}_{n+1}^{-1} \right]
\end{equation}
in a similar vein as \eqref{Partial mu-bar Partial C}. Substituting \eqref{Partial mu-bar Partial C} and \eqref{Partial f-beta-f Partial C} into \eqref{Partial mu-bar-bar Partial C} renders
\begin{equation*}
    \dfrac{\partial \overline{\overline{\mu}}_{n+1}^\texttt{tr}}{\partial \mathbf{C}_{n+1}} = 
    \dfrac{\mu}{3} J_{n+1}^{-2/3}
    \left\{ \left( \mathbf{C}_n^\texttt{p} \right)^{-1}
    - \dfrac{1}{3} \left[ \left( \mathbf{C}_n^\texttt{p} \right)^{-1} : \mathbf{C}_{n+1} \right] \mathbf{C}_{n+1}^{-1} \right\}
    - \dfrac{1}{3} J_{n+1}^{-2/3}
    \left[ \boldsymbol{\Gamma}_n
    - \dfrac{1}{3} \left( \boldsymbol{\Gamma}_n : \mathbf{C}_{n+1} \right) \mathbf{C}_{n+1}^{-1} \right].
\end{equation*}

For simplification, we further apply the push-forward operator ($\varphi_*$) and compute
\begin{equation*}
    \begin{array}{ll}
        \varphi_* \left( \dfrac{\partial \overline{\overline{\mu}}_{n+1}^\texttt{tr}}{\partial \mathbf{C}_{n+1}} \right)
        &= \dfrac{\mu}{3} \left[ \overline{\mathbf{f}}_{n+1} \overline{\mathbf{b}}_n^\texttt{e} \overline{\mathbf{f}}_{n+1}^\top
        - \dfrac{1}{3} \text{tr} \left( \overline{\mathbf{f}}_{n+1} \overline{\mathbf{b}}_n^\texttt{e} \overline{\mathbf{f}}_{n+1}^\top \right) \mathbf{I} \right]
        - \dfrac{1}{3} \left[ \overline{\mathbf{f}}_{n+1} \overline{\boldsymbol{\beta}}_n \overline{\mathbf{f}}_{n+1}^\top - \dfrac{1}{3} \text{tr} \left( \overline{\mathbf{f}}_{n+1} \overline{\boldsymbol{\beta}}_n \overline{\mathbf{f}}_{n+1}^\top \right) \mathbf{I} \right] \\[12pt]
         & = \dfrac{\mu}{3} \text{dev} \left( \overline{\mathbf{f}}_{n+1} \overline{\mathbf{b}}_n^\texttt{e} \overline{\mathbf{f}}_{n+1}^\top \right)
         - \dfrac{1}{3} \text{dev} \left( \overline{\mathbf{f}}_{n+1} \overline{\boldsymbol{\beta}}_n \overline{\mathbf{f}}_{n+1}^\top \right)
         = \dfrac{\mu}{3} \text{dev} \left( \overline{\mathbf{b}}_{n+1}^\texttt{e,tr} \right)
         - \dfrac{1}{3} \text{dev} \left( \overline{\boldsymbol{\beta}}_{n+1}^\texttt{tr} \right) \\[12pt]
         & = \dfrac{1}{3} \left[ \mathbf{s}_{n+1}^\texttt{tr}
         - \text{dev} \left( \overline{\boldsymbol{\beta}}_{n+1}^\texttt{tr} \right) \right]
         = \dfrac{1}{3} \boldsymbol{\xi}_{n+1}^\texttt{tr}
    \end{array}
\end{equation*}
by using \eqref{be-bar Trial} and \eqref{J^(-2/3) F Gamma F}. Finally, we derive
\begin{equation} \label{Partial mu-bar-bar Partial C Final}
    \dfrac{\partial \overline{\overline{\mu}}_{n+1}^\texttt{tr}}{\partial \mathbf{C}_{n+1}} = \dfrac{1}{3} \varphi^* \left( \boldsymbol{\xi}_{n+1}^\texttt{tr} \right). 
\end{equation}

\subsection[]{The computation of $\partial \widehat{\gamma}_{n+1} / \partial \mathbf{C}_{n+1}$} 

Recall the consistency parameter, $\widehat{\gamma}_{n+1}$, is governed by the algebraic equation, $\mathcal{G}(\widehat{\gamma}_{n+1}) = 0$ in \eqref{Function of gamma}. Taking the derivatives on \eqref{Function of gamma} yields
\begin{equation} \label{Partial gamma Partial C Initial}
    \dfrac{\partial \lVert \boldsymbol{\xi}_{n+1}^\texttt{tr} \rVert}{\partial \mathbf{C}_{n+1}}
    - 2 \left( 1 + \dfrac{h}{3\mu} \right) \widehat{\gamma}_{n+1}
    \dfrac{\partial \overline{\overline{\mu}}_{n+1}^\texttt{tr}}{\partial \mathbf{C}_{n+1}}
    - 2 \overline{\overline{\mu}}_{n+1}^\texttt{tr}
    \left( 1 + \dfrac{h}{3\mu} \right) \dfrac{\partial \widehat{\gamma}_{n+1}}{\partial \mathbf{C}_{n+1}}
    - \dfrac{2 k'}{3} \dfrac{\partial \widehat{\gamma}_{n+1}}{\partial \mathbf{C}_{n+1}} = 0,
\end{equation}
and one apparent next step is to compute $\partial \lVert \boldsymbol{\xi}_{n+1}^\texttt{tr} \rVert / \partial \mathbf{C}_{n+1}$. To this end, we consider
\begin{equation*}
    \lVert \boldsymbol{\xi}_{n+1}^\texttt{tr} \rVert^2 = \boldsymbol{\xi}_{n+1}^\texttt{tr}:\boldsymbol{\xi}_{n+1}^\texttt{tr}
    = \varphi_* \left[ \varphi^* (\boldsymbol{\xi}_{n+1}^\texttt{tr}) \right] : \varphi_* \left[ \varphi^* (\boldsymbol{\xi}_{n+1}^\texttt{tr}) \right]
    = \left[ \mathbf{C}_{n+1} \varphi^* (\boldsymbol{\xi}_{n+1}^\texttt{tr}) \right] : \left[ \mathbf{C}_{n+1} \varphi^* (\boldsymbol{\xi}_{n+1}^\texttt{tr}) \right]
\end{equation*}
and derive
\begin{equation} \label{Partial xi-trial-norm-square Partial C}
    \dfrac{\partial \lVert \boldsymbol{\xi}_{n+1}^\texttt{tr} \rVert^2}{\partial \mathbf{C}_{n+1}}
    = 2 \left[ \mathbf{C}_{n+1} \varphi^* (\boldsymbol{\xi}_{n+1}^\texttt{tr}) \mathbf{C}_{n+1} \right] : \dfrac{\partial \varphi^* (\boldsymbol{\xi}_{n+1}^\texttt{tr})}{\partial \mathbf{C}_{n+1}}
    + 2 \varphi^* (\boldsymbol{\xi}_{n+1}^\texttt{tr}) \mathbf{C}_{n+1} \varphi^* (\boldsymbol{\xi}_{n+1}^\texttt{tr}).
\end{equation}

Following the derivation of $\partial \varphi^* (\mathbf{s}_{n+1}^\texttt{tr}) / \partial \mathbf{C}_{n+1}$ from \eqref{Partial s-trial-forward Partial C Final}, we compute
\begin{equation*}
    \dfrac{\partial \varphi^*(\overline{\boldsymbol{\beta}}_{n+1}^\texttt{tr})}{\partial \mathbf{C}_{n+1}}
    = \dfrac{1}{3}\text{tr} \left( \overline{\mathbf{f}}_{n+1} \overline{\boldsymbol{\beta}}_n \overline{\mathbf{f}}_{n+1}^\top \right)
    \left( \mathbb{I}_{\mathbf{C}_{n+1}^{-1}} - \dfrac{1}{3} \mathbf{C}_{n+1}^{-1} \otimes \mathbf{C}_{n+1}^{-1} \right)
    - \dfrac{1}{3} \left[ \varphi^* \left( \overline{\boldsymbol{\beta}}_{n+1}^\texttt{tr} \right) \otimes \mathbf{C}_{n+1}^{-1}
    + \mathbf{C}_{n+1}^{-1} \otimes \varphi^* \left( \overline{\boldsymbol{\beta}}_{n+1}^\texttt{tr} \right) \right]
\end{equation*}
and further derive
\begin{equation} \label{Partial xi-trial-forward Partial C}
    \begin{array}{ll}
        \dfrac{\partial \varphi^*(\boldsymbol{\xi}_{n+1}^\texttt{tr})}{\partial \mathbf{C}_{n+1}} &= \dfrac{\partial \varphi^*(\mathbf{s}_{n+1}^\texttt{tr})}{\partial \mathbf{C}_{n+1}} -\dfrac{\partial \varphi^*(\overline{\boldsymbol{\beta}}_{n+1}^\texttt{tr})}{\partial \mathbf{C}_{n+1}} \\[12pt]
        &= \overline{\overline{\mu}}_{n+1}^\texttt{tr} \left( \mathbb{I}_{\mathbf{C}_{n+1}^{-1}} - \dfrac{1}{3} \mathbf{C}_{n+1}^{-1} \otimes \mathbf{C}_{n+1}^{-1} \right) - \dfrac{1}{3} \left[ \varphi^* \left( \boldsymbol{\xi}_{n+1}^\texttt{tr} \right) \otimes \mathbf{C}_{n+1}^{-1}
        + \mathbf{C}_{n+1}^{-1} \otimes \varphi^* \left( \boldsymbol{\xi}_{n+1}^\texttt{tr} \right) \right].
    \end{array}
\end{equation}
Substituting \eqref{Partial xi-trial-forward Partial C} into \eqref{Partial xi-trial-norm-square Partial C} and applying the push-forward operator ($\varphi_*$) yield
\begin{equation*}
    \varphi_* \left( \dfrac{\partial \lVert \boldsymbol{\xi}_{n+1}^\texttt{tr} \rVert^2}{\partial \mathbf{C}_{n+1}} \right)
    = 2 \overline{\overline{\mu}}_{n+1}^\texttt{tr} \boldsymbol{\xi}_{n+1}^\texttt{tr}
    - \dfrac{2}{3} \lVert \boldsymbol{\xi}_{n+1}^\texttt{tr} \rVert^2 \mathbf{I}
    + 2 \left( \boldsymbol{\xi}_{n+1}^\texttt{tr} \right)^2.
\end{equation*}
Considering
\begin{equation*}
    \dfrac{\partial \lVert \boldsymbol{\xi}_{n+1}^\texttt{tr} \rVert^2}{\partial \mathbf{C}_{n+1}}
    = \dfrac{\partial \lVert \boldsymbol{\xi}_{n+1}^\texttt{tr} \rVert^2}{\partial \lVert \boldsymbol{\xi}_{n+1}^\texttt{tr} \rVert}
    \dfrac{\partial \lVert \boldsymbol{\xi}_{n+1}^\texttt{tr} \rVert}{\partial \mathbf{C}_{n+1}}
    = 2 \lVert \boldsymbol{\xi}_{n+1}^\texttt{tr} \rVert \dfrac{\partial \lVert \boldsymbol{\xi}_{n+1}^\texttt{tr} \rVert}{\partial \mathbf{C}_{n+1}}
    \quad \Rightarrow \quad
    \dfrac{\partial \lVert \boldsymbol{\xi}_{n+1}^\texttt{tr} \rVert}{\partial \mathbf{C}_{n+1}} 
    = \dfrac{1}{2 \lVert \boldsymbol{\xi}_{n+1}^\texttt{tr} \rVert}
    \dfrac{\partial \lVert \boldsymbol{\xi}_{n+1}^\texttt{tr} \rVert^2}{\partial \mathbf{C}_{n+1}},
\end{equation*}
we derive
\begin{equation*}
    \varphi_* \left( \dfrac{\partial \lVert \boldsymbol{\xi}_{n+1}^\texttt{tr} \rVert}{\partial \mathbf{C}_{n+1}} \right)
    = \overline{\overline{\mu}}_{n+1}^\texttt{tr} \mathbf{n}_{n+1}
    - \dfrac{1}{3} \lVert \boldsymbol{\xi}_{n+1}^\texttt{tr} \rVert \left( \mathbf{n}_{n+1} : \mathbf{n}_{n+1} \right) \mathbf{I}
    + \lVert \boldsymbol{\xi}_{n+1}^\texttt{tr} \rVert \mathbf{n}_{n+1}^2.
\end{equation*}

We remark that the term $\varphi_* \left( \partial \lVert \boldsymbol{\xi}_{n+1}^\texttt{tr} \rVert / \partial \mathbf{C}_{n+1} \right)$ is deviatoric by noticing
\begin{equation*}
    \text{tr} \left[ \varphi_* \left( \dfrac{\partial \lVert \boldsymbol{\xi}_{n+1}^\texttt{tr} \rVert}{\partial \mathbf{C}_{n+1}} \right) \right]
    = \overline{\overline{\mu}}_{n+1}^\texttt{tr} \text{tr}(\mathbf{n}_{n+1})
    - \dfrac{1}{3} \lVert \boldsymbol{\xi}_{n+1}^\texttt{tr} \rVert \mathbf{n}_{n+1} : \mathbf{n}_{n+1} \text{tr}(\mathbf{I})
    + \lVert \boldsymbol{\xi}_{n+1}^\texttt{tr} \rVert \text{tr} (\mathbf{n}_{n+1}^2) = 0.
\end{equation*}
Consequently, we can rewrite
\begin{equation*}
    \varphi_* \left( \dfrac{\partial \lVert \boldsymbol{\xi}_{n+1}^\texttt{tr} \rVert}{\partial \mathbf{C}_{n+1}} \right)
    = \text{dev} \left[ \varphi_* \left( \dfrac{\partial \lVert \boldsymbol{\xi}_{n+1}^\texttt{tr} \rVert}{\partial \mathbf{C}_{n+1}} \right) \right]
    = \overline{\overline{\mu}}_{n+1}^\texttt{tr} \mathbf{n}_{n+1} + \lVert \boldsymbol{\xi}_{n+1}^\texttt{tr} \rVert \text{dev} (\mathbf{n}_{n+1}^2),
\end{equation*}
which renders
\begin{equation} \label{Partial xi-trial-norm Partial C}
    \dfrac{\partial \lVert \boldsymbol{\xi}_{n+1}^\texttt{tr} \rVert}{\partial \mathbf{C}_{n+1}}
    = \varphi^* \left[ \overline{\overline{\mu}}_{n+1}^\texttt{tr} \mathbf{n}_{n+1} + \lVert \boldsymbol{\xi}_{n+1}^\texttt{tr} \rVert \text{dev} (\mathbf{n}_{n+1}^2) \right]
    = \overline{\overline{\mu}}_{n+1}^\texttt{tr} \varphi^*(\mathbf{n}_{n+1})
    + \lVert \boldsymbol{\xi}_{n+1}^\texttt{tr} \rVert \varphi^*\left[ \text{dev} (\mathbf{n}_{n+1}^2) \right].
\end{equation}
Finally, substituting \eqref{Partial mu-bar-bar Partial C Final} and \eqref{Partial xi-trial-norm Partial C} into \eqref{Partial gamma Partial C Initial} yields
\begin{equation} \label{Partial gamma Partial C Final}
    \begin{array}{ll}
        \dfrac{\partial \widehat{\gamma}_{n+1}}{\partial \mathbf{C}_{n+1}} &= \dfrac{1}{2 c_0 \overline{\overline{\mu}}_{n+1}^\texttt{tr}} 
        \left[ \dfrac{\partial \lVert \boldsymbol{\xi}_{n+1}^\texttt{tr} \rVert}{\partial \mathbf{C}_{n+1}}
        - 2 \left( 1 + \dfrac{h}{3\mu} \right) \widehat{\gamma}_{n+1}
        \dfrac{\partial \overline{\overline{\mu}}_{n+1}^\texttt{tr}}{\partial \mathbf{C}_{n+1}} \right] \\[12pt]
        
        &= \dfrac{1}{2 c_0 \overline{\overline{\mu}}_{n+1}^\texttt{tr}}  
        \left\{ \overline{\overline{\mu}}_{n+1}^\texttt{tr} \varphi^*(\mathbf{n}_{n+1})
        + \lVert \boldsymbol{\xi}_{n+1}^\texttt{tr} \rVert \varphi^*\left[ \text{dev} (\mathbf{n}_{n+1}^2) \right]
        - 2 \left( 1 + \dfrac{h}{3\mu} \right) \widehat{\gamma}_{n+1}
        \left[ \dfrac{1}{3} \varphi^* \left( \boldsymbol{\xi}_{n+1}^\texttt{tr} \right) \right] \right\} \\[12pt]

        &= \dfrac{1}{2 c_0 \overline{\overline{\mu}}_{n+1}^\texttt{tr}} 
        \left\{ \left[ \overline{\overline{\mu}}_{n+1}^\texttt{tr} - \dfrac{2}{3} \left( 1 + \dfrac{h}{3\mu} \right) \widehat{\gamma}_{n+1} \lVert \boldsymbol{\xi}_{n+1}^\texttt{tr} \rVert \right] \varphi^* (\mathbf{n}_{n+1})
        + \lVert \boldsymbol{\xi}_{n+1}^\texttt{tr} \rVert \varphi^*\left[ \text{dev} (\mathbf{n}_{n+1}^2) \right] \right\}
    \end{array}
\end{equation}
where we define
\begin{equation*}
    c_0 =  1 + \dfrac{h}{3\mu} + \dfrac{k'}{3 \overline{\overline{\mu}}_{n+1}^\texttt{tr}}.
\end{equation*} 

\subsection[]{The computation of $\partial \varphi^*(\mathbf{n}_{n+1}) / \partial \mathbf{C}_{n+1}$} 

Based on the chain rules, we write out
\begin{equation} \label{Partial n-forward Partial C Initial}
    \dfrac{\partial \varphi^*(\mathbf{n}_{n+1})}{\partial \mathbf{C}_{n+1}} = \dfrac{\partial}{\partial \mathbf{C}_{n+1}} \left[ \dfrac{\varphi^* (\boldsymbol{\xi}_{n+1}^\texttt{tr})}{\lVert \boldsymbol{\xi}_{n+1}^\texttt{tr} \rVert} \right]
    = \dfrac{1}{\lVert \boldsymbol{\xi}_{n+1}^\texttt{tr} \rVert^2}
    \left[ \lVert \boldsymbol{\xi}_{n+1}^\texttt{tr} \rVert \dfrac{\partial \varphi^* (\boldsymbol{\xi}_{n+1}^\texttt{tr})}{\partial \mathbf{C}_{n+1}}
    - \varphi^*(\boldsymbol{\xi}_{n+1}^\texttt{tr}) \otimes \dfrac{\partial \lVert \boldsymbol{\xi}_{n+1}^\texttt{tr} \rVert}{\partial \mathbf{C}_{n+1}} \right].
\end{equation}
Substituting \eqref{Partial xi-trial-norm Partial C} and \eqref{Partial xi-trial-forward Partial C} into \eqref{Partial n-forward Partial C Initial} renders
\begin{equation} \label{Partial n-forward Partial C Final}
    \begin{array}{ll}
        \dfrac{\partial \varphi^*(\mathbf{n}_{n+1})}{\partial \mathbf{C}_{n+1}}
        = &\dfrac{\overline{\overline{\mu}}_{n+1}^\texttt{tr}}{\lVert \boldsymbol{\xi}_{n+1}^\texttt{tr} \rVert}
        \left( \mathbb{I}_{\mathbf{C}_{n+1}^{-1}} - \dfrac{1}{3} \mathbf{C}_{n+1}^{-1} \otimes \mathbf{C}_{n+1}^{-1} \right) - \dfrac{1}{3} \left[ \varphi^* \left( \mathbf{n}_{n+1} \right) \otimes \mathbf{C}_{n+1}^{-1}
        + \mathbf{C}_{n+1}^{-1} \otimes \varphi^* \left( \mathbf{n}_{n+1} \right) \right] \\[12pt]
        
        &- \dfrac{\overline{\overline{\mu}}_{n+1}^\texttt{tr}}{\lVert \boldsymbol{\xi}_{n+1}^\texttt{tr} \rVert}
        \varphi^*(\mathbf{n}_{n+1}) \otimes
        \varphi^*(\mathbf{n}_{n+1})
        - \varphi^*(\mathbf{n}_{n+1}) \otimes \varphi^*\left[ \text{dev} (\mathbf{n}_{n+1}^2) \right].
    \end{array}
\end{equation}

\subsection[]{Complete expression of the second elastoplastic moduli}

Substituting $\partial \varphi^*(\boldsymbol{\tau}_{n+1}^\texttt{vol}) / \partial \mathbf{C}_{n+1}$ from \eqref{Partial tau Partial C Final}, $\partial \varphi^*(\mathbf{s}_{n+1}^\texttt{tr}) / \partial \mathbf{C}_{n+1}$ from \eqref{Partial s-trial-forward Partial C Final}, $\partial \overline{\overline{\mu}}_{n+1}^\texttt{tr} / \partial \mathbf{C}_{n+1}$ from \eqref{Partial mu-bar-bar Partial C Final}, $\partial \widehat{\gamma}_{n+1} / \partial \mathbf{C}_{n+1}$ from \eqref{Partial gamma Partial C Final}, and $\partial \varphi^*(\mathbf{n}_{n+1}) / \partial \mathbf{C}_{n+1}$ from \eqref{Partial n-forward Partial C Final} into \eqref{Second Elastoplastic Moduli Initial}, we finally derive
\begin{equation*}
    \begin{array}{ll}
        \mathbb{C}_{n+1}^\texttt{ep} = &J_{n+1} (J_{n+1} U'_{n+1})' \mathbf{C}_{n+1}^{-1} \otimes \mathbf{C}_{n+1}^{-1}
        - 2 J_{n+1} U'_{n+1} \mathbb{I}_{\mathbf{C}_{n+1}^{-1}} \\[12pt]
        
        &+ 2 \overline{\mu}_{n+1}^\texttt{tr} \left( \mathbb{I}_{\mathbf{C}_{n+1}^{-1}} - \dfrac{1}{3} \mathbf{C}_{n+1}^{-1} \otimes \mathbf{C}_{n+1}^{-1} \right)
        - \dfrac{2}{3} \left[ \varphi^* \left( \mathbf{s}_{n+1}^\texttt{tr} \right) \otimes \mathbf{C}_{n+1}^{-1}
        + \mathbf{C}_{n+1}^{-1} \otimes \varphi^* \left( \mathbf{s}_{n+1}^\texttt{tr} \right) \right] \\[12pt]

        &- \dfrac{4}{3} \widehat{\gamma}_{n+1} \lVert \boldsymbol{\xi}_{n+1}^\texttt{tr} \rVert \varphi^*(\mathbf{n}_{n+1}) \otimes \varphi^* (\mathbf{n}_{n+1}) \\[12pt]

        &- \dfrac{2}{c_0} \varphi^* (\mathbf{n}_{n+1}) \otimes
        \left\{ \left[ \overline{\overline{\mu}}_{n+1}^\texttt{tr} - \dfrac{2}{3} \left( 1 + \dfrac{h}{3\mu} \right) \widehat{\gamma}_{n+1} \lVert \boldsymbol{\xi}_{n+1}^\texttt{tr} \rVert \right] \varphi^* (\mathbf{n}_{n+1})
        + \lVert \boldsymbol{\xi}_{n+1}^\texttt{tr} \rVert \varphi^*\left[ \text{dev} (\mathbf{n}_{n+1}^2) \right] \right\} \\[12pt]

        &- \dfrac{4 \left( \overline{\overline{\mu}}_{n+1}^\texttt{tr} \right)^2 \widehat{\gamma}_{n+1}}{\lVert \boldsymbol{\xi}_{n+1}^\texttt{tr} \rVert}
        \left( \mathbb{I}_{\mathbf{C}_{n+1}^{-1}} - \dfrac{1}{3} \mathbf{C}_{n+1}^{-1} \otimes \mathbf{C}_{n+1}^{-1} \right) \\[12pt] 
        
        &+ \dfrac{4 \overline{\overline{\mu}}_{n+1}^\texttt{tr} \widehat{\gamma}_{n+1}}{3} \left[ \varphi^* \left( \mathbf{n}_{n+1} \right) \otimes \mathbf{C}_{n+1}^{-1}
        + \mathbf{C}_{n+1}^{-1} \otimes \varphi^* \left( \mathbf{n}_{n+1} \right) \right] \\[12pt] 

        &+ \dfrac{4 \left( \overline{\overline{\mu}}_{n+1}^\texttt{tr} \right)^2 \widehat{\gamma}_{n+1}}{\lVert \boldsymbol{\xi}_{n+1}^\texttt{tr} \rVert}
        \varphi^*(\mathbf{n}_{n+1}) \otimes
        \varphi^*(\mathbf{n}_{n+1})
        
        + 4 \overline{\overline{\mu}}_{n+1}^\texttt{tr} \widehat{\gamma}_{n+1} \varphi^*(\mathbf{n}_{n+1}) \otimes \varphi^*\left[ \text{dev} (\mathbf{n}_{n+1}^2) \right].
    \end{array}
\end{equation*}
Regrouping terms renders
\begin{equation*}
    \begin{array}{ll}
        \mathbb{C}_{n+1}^\texttt{ep} = &\left( 2 \overline{\mu}_{n+1}^\texttt{tr} - 2 c_1 \overline{\overline{\mu}}_{n+1}^\texttt{tr} - 2 J_{n+1} U'_{n+1} \right) \mathbb{I}_{\mathbf{C}_{n+1}^{-1}}
        
        + \left[ J_{n+1} (J_{n+1} U'_{n+1})' - \dfrac{2 \overline{\mu}_{n+1}^\texttt{tr}}{3} + \dfrac{2 c_1 \overline{\overline{\mu}}_{n+1}^\texttt{tr}}{3} \right] \mathbf{C}_{n+1}^{-1} \otimes \mathbf{C}_{n+1}^{-1} \\[12pt]
        
        &- \dfrac{2}{3} \left[ \varphi^*(\mathbf{s}_{n+1}^\texttt{tr}) \otimes \mathbf{C}_{n+1}^{-1} + \mathbf{C}_{n+1}^{-1} \otimes \varphi^*(\mathbf{s}_{n+1}^\texttt{tr}) \right]

        + \dfrac{2 c_1}{3} \left[ \varphi^*(\boldsymbol{\xi}_{n+1}^\texttt{tr}) \otimes \mathbf{C}_{n+1}^{-1} + \mathbf{C}_{n+1}^{-1} \otimes \varphi^*(\boldsymbol{\xi}_{n+1}^\texttt{tr}) \right] \\[12pt]

        &- c_3 \varphi^*(\mathbf{n}_{n+1}) \otimes \varphi^*(\mathbf{n}_{n+1})

        - c_4 \varphi^*(\mathbf{n}_{n+1}) \otimes \varphi^*\left[ \text{dev} (\mathbf{n}_{n+1}^2) \right].
    \end{array}
\end{equation*}
with $c_1$, $c_3$, and $c_4$ defined in \eqref{Parameters in Elasticity}. To ensure the major symmetry of $\mathbb{C}_{n+1}^\texttt{ep}$ for computational benefits, we manually symmetrize its last term by following \citet{simo_framework_1988} and derive the expression in \eqref{Second Elastoplastic Moduli Final}.

\section{Verification of FEA implementation for finite strain elastoplasticity}
\label{Sec: FEA Verification}

In this section, we compare the FEA and semi-analytical solutions for a 3D column under uniaxial cyclic loadings. Through this comparison, we aim to demonstrate that the proposed formulae for updating $\overline{\mathbf{b}}^\texttt{e}$ in $\eqref{Updated be_bar}$ ensure the isochoric plastic flow ($J^\texttt{p}=1$), and consequently, the FEA solution matches the semi-analytical solution exactly.

\subsection{Construction of the semi-analytical solution}

For the simplicity of the semi-analytical solution, we assume no hardening occurs ($k(\alpha) \equiv \sigma_y$ and $h=0$) when the material yields. At any time, the deformation gradients can be expressed as
\begin{equation*}
    \mathbf{F} = \left[ \begin{array}{ccc}
        \lambda & & \\
        & \lambda^\texttt{l} & \\
        & & \lambda^\texttt{l}
    \end{array} \right], \quad
    \mathbf{F}^\texttt{p} = \left[ \begin{array}{ccc}
        \lambda^\texttt{p} & & \\
        & \dfrac{1}{\sqrt{\lambda^\texttt{p}}} & \\
        & & \dfrac{1}{\sqrt{\lambda^\texttt{p}}}
    \end{array} \right], \quad \text{and} \quad
    \mathbf{F}^\texttt{e} = \left[ \begin{array}{ccc}
        \dfrac{\lambda}{\lambda^\texttt{p}} & & \\
        & \lambda^\texttt{l} \sqrt{\lambda^\texttt{p}} & \\
        & & \lambda^\texttt{l} \sqrt{\lambda^\texttt{p}}
    \end{array} \right].
\end{equation*}
Here, the variables $\lambda$, $\lambda^\texttt{l}$, and $\lambda^\texttt{p}$ represent the total applied stretch, total lateral stretch, and plastic stretch (Fig. \ref{Fig: Isochoric Plastic Flow}(a)), respectively. Note that the above deformation gradients satisfy the requirements of $\mathbf{F} = \mathbf{F}^\texttt{e} \mathbf{F}^\texttt{p}$ and $\det(\mathbf{F}^\texttt{p}) = 1$. Based on the elastic deformation gradient ($\mathbf{F}^\texttt{e}$), we compute
\begin{equation*}
    \text{dev} (\overline{\mathbf{b}}^\texttt{e}) = \dfrac{1}{3} \left[ \lambda^\texttt{p} \left( \dfrac{\lambda^\texttt{l}}{\lambda} \right)^{2/3}
    - \dfrac{1}{(\lambda^\texttt{p})^2} \left( \dfrac{\lambda}{\lambda^\texttt{l}} \right)^{4/3} \right]
    \left[ \begin{array}{ccc}
        -2 & & \\
        & 1 & \\
        & & 1
    \end{array} \right]
\end{equation*}
and derive the components of the Kirchhoff stress tensor ($\boldsymbol{\tau}$) as
\begin{equation*}
    \left\{ \begin{array}{l}
        \tau_{11} = J U'(J) - \dfrac{2\mu}{3} \left[ \lambda^\texttt{p} \left( \dfrac{\lambda^\texttt{l}}{\lambda} \right)^{2/3} - \dfrac{1}{(\lambda^\texttt{p})^2} \left( \dfrac{\lambda}{\lambda^\texttt{l}} \right)^{4/3} \right], \\[12pt]
        \tau_{22} = \tau_{33} = J U'(J) + \dfrac{\mu}{3} \left[ \lambda^\texttt{p} \left( \dfrac{\lambda^\texttt{l}}{\lambda} \right)^{2/3} - \dfrac{1}{(\lambda^\texttt{p})^2} \left( \dfrac{\lambda}{\lambda^\texttt{l}} \right)^{4/3}\right],
    \end{array} \right.
\end{equation*}
where $J = \lambda (\lambda^\texttt{l})^2$ is the total volume change. 

\begin{figure}[!htbp]
    \centering
    \includegraphics[width=17cm]{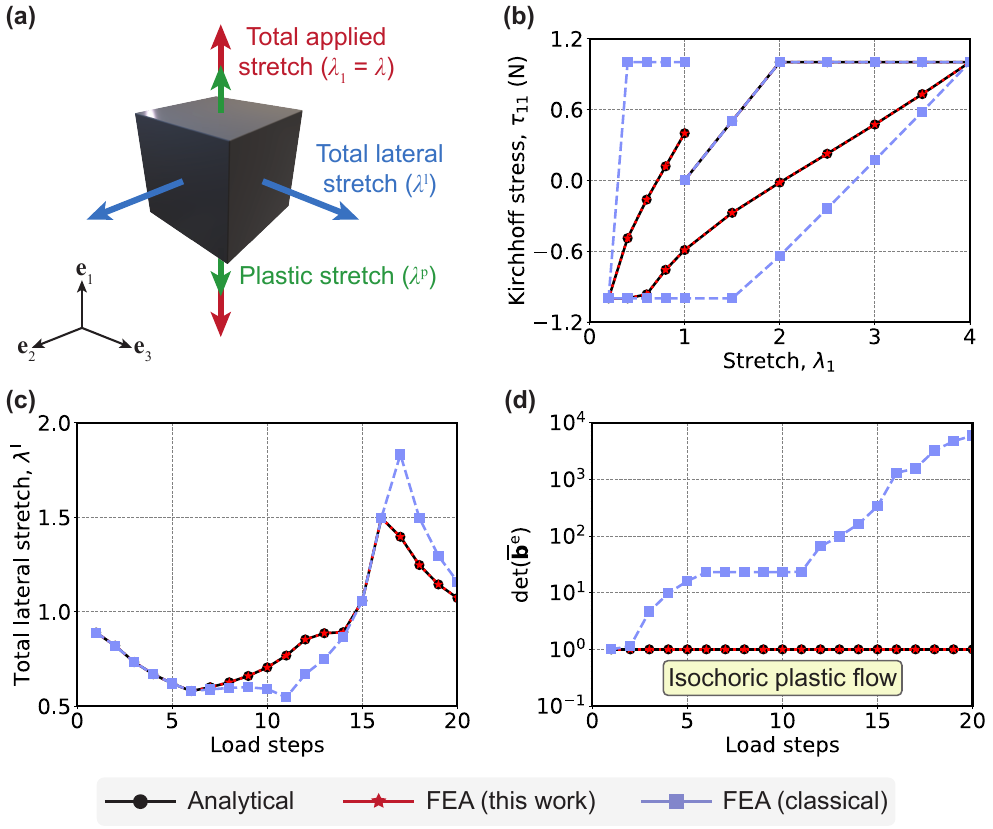}
    \caption{Comparison between the FEA and semi-analytical solutions. (a) Tested mechanical setups. (b) Kirchhoff stress--stretch ($\tau_{11}$--$\lambda_1$) curves. (c) Evolution of the total lateral stretch ($\lambda^\texttt{l}$). (d) Evolution of the determinant of $\overline{\mathbf{b}}^\texttt{e}$.}
    \label{Fig: Isochoric Plastic Flow}
\end{figure}

During the elastic loading/unloading process, the Kirchhoff stress conforms to $\tau_{22} = \tau_{33} = 0$ for given $\lambda$ and $\lambda_p$. We can then solve $\lambda_l$ from
\begin{equation} \label{Solve for lambda_l for Elastic Process}
   \dfrac{\kappa}{2} \left[ \lambda^2 (\lambda^\texttt{l})^4-1 \right] + \dfrac{\mu}{3} \left[ \lambda^\texttt{p} \left( \dfrac{\lambda^\texttt{l}}{\lambda} \right)^{2/3} - \dfrac{1}{(\lambda^\texttt{p})^2} \left( \dfrac{\lambda}{\lambda^\texttt{l}} \right)^{4/3}\right] = 0
\end{equation}
with the choice of $U(J)$ in \eqref{Volumetric and Deviatoric Engery}. When the material yields, the Kirchhoff stress satisfies $\tau_{11} = \pm \sigma_y$ (the sign depends on the loading direction) and $\tau_{22} = \tau_{33} = 0$ for given $\lambda$. We can then analytically compute $\lambda^\texttt{l}$ as
\begin{equation} \label{Solve for lambda_l for Plastic Process}
    \lambda^\texttt{l} = \dfrac{1}{\sqrt{\lambda}} \left( 1 \pm \dfrac{2 \sigma_y}{3 \kappa} \right)^{1/4}
\end{equation}
and solve for $\lambda^\texttt{p}$ from
\begin{equation} \label{Solve for lambda_p for Plastic Process}
    \left( \dfrac{\lambda^\texttt{l}}{\lambda} \right)^{2/3} (\lambda^\texttt{p})^3 \pm \dfrac{\sigma_y}{\mu} (\lambda^\texttt{p})^2 - \left( \dfrac{\lambda}{\lambda^\texttt{l}} \right)^{4/3} = 0.
\end{equation}
Note that the expressions in \eqref{Solve for lambda_l for Elastic Process}, \eqref{Solve for lambda_l for Plastic Process}, and \eqref{Solve for lambda_p for Plastic Process} construct the semi-analytical solutions for predicting the finite strain elastoplastic responses of materials under uniaxial cyclic loadings. 

\subsection{Comparison between the FEA and semi-analytical solutions}
\label{Sec: Comparison with Analytical Solution}

Next, we compare the semi-analytical solutions described above with the FEA predictions based on the theories in Section \ref{Sec: Finite Strain Elastoplasticity}. As a reference, we also include FEA predictions derived from the classical theories in \citet{simo_framework_1988-1, simo_framework_1988} that utilize $\eqref{Local Governing Equations}_1$ to update $\overline{\mathbf{b}}^\texttt{e}$. Without loss of generality, we adopt dummy material constants, $E=1$ MPa, $\nu=0.3$, and $\sigma_y=0.2$ MPa, where $E$ is the initial Young's modulus and $\nu$ is the initial Poisson's ratio.

The comparison results are shown in Figs. \ref{Fig: Isochoric Plastic Flow}(b)--(d), which illustrate the Kirchhoff stress component ($\tau_{11}$), lateral stretch ($\lambda^\texttt{l}$), and the determinant of $\overline{\mathbf{b}}^\texttt{e}$, respectively. Based on the comparison, the FEA predictions in this work ensure isochoric plastic flow and demonstrate good agreement with the semi-analytical solutions, as evidenced by the small absolute 2-norm errors: $1.56 \times 10^{-10}$ for $\tau_{11}$, $7.76 \times 10^{-11}$ for $\lambda^\texttt{l}$, and $7.95 \times 10^{-10}$ for $\text{det}(\overline{\mathbf{b}}^\texttt{e})$. Conversely, the classical FEA initially aligns well during the loading stage but fails to ensure isochoric plastic flow. It also leads to significant deviations from the analytical solutions after unloading occurs. This limitation renders the classical FEA unsuitable for structures experiencing non-monotonic loadings, such as the dampers under cyclic loadings discussed in Section \ref{Sec: Dampers}.

\section{Convergence, precision, and computational time of FEA}
\label{Sec: FEA Convergence}

In this section, we examine the computational aspects of FEA, including convergence, precision, and computational time. Using the optimized damper under multiple-cycle loadings in Fig. \ref{Fig: Damper-Part 3} as an example, we present the required Newton iterations for convergence in Fig. \ref{Fig: Computational Cost}(a). For most load steps in FEA, whether during loading or unloading, convergence is achieved in fewer than 10 Newton iterations. However, at critical load steps corresponding to transitions between loading and unloading, it requires 10--25 iterations to converge.

To further illustrate the convergence and precision of FEA, we highlight four representative load steps in Fig. \ref{Fig: Computational Cost}(a) and show the evolution of the relative residual (right-hand side of $\eqref{Global Governing Equation}_1$) at these steps in Fig. \ref{Fig: Computational Cost}(b). For most steps, such as steps 60 and 100, the residual converges quadratically due to the algorithmic elastoplastic moduli in \eqref{Second Algorithmic Tangent Moduli Final} and the proposed interpolation schemes in \eqref{Interpolated First PK Stress} and \eqref{Interpolated Left-hand Side}. For transitional steps, like steps 66 and 156, a line search (Algorithm \ref{Algorithm of Line Search}) is activated in the initial iterations to reduce the residual. Once the residual reduction trend is established, quadratic convergence resumes. Notably, as shown in Fig. \ref{Fig: Computational Cost}(b), the residuals for all load steps eventually converge to the prescribed tolerance (set to $10^{-8}$ here).

\begin{figure}[!htbp]
    \centering
    \includegraphics[width=16cm]{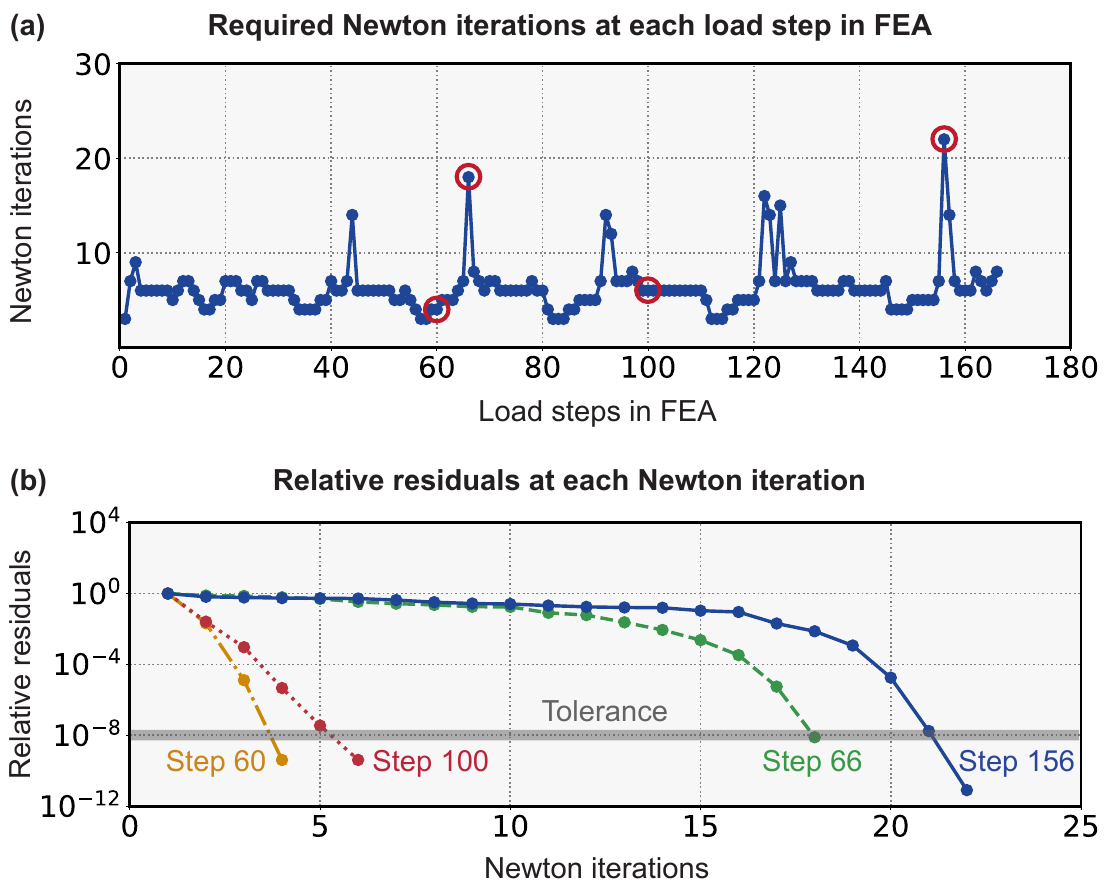}
    \caption{Convergence and precision of FEA. (a) Required Newton iterations for convergence at each load step in FEA. (b) Evolution of the relative residuals at representative load steps marked in (a).}
    \label{Fig: Computational Cost}
\end{figure}

After discussing convergence and precision, we now focus on the computational time of FEA. For representativeness, we analyze the most complex design case in each example presented in Section \ref{Sec: Sample Examples}. Computational performance is tested on a workstation equipped with an AMD Ryzen Threadripper PRO 3995WX CPU featuring 64 cores and 256 GB of memory. The results are summarized in Table \ref{Table: Computational Time}, where we provide the finite element, degree of freedom (DOF), and parallel process counts, along with the average computational time per Newton iteration in the final load step.

It is evident that the 2D designs (damper and beam) require less than 0.1 seconds per Newton iteration due to the parallel computing capabilities provided by FEniTop \citep{jia_fenitop_2024}. The profiled sheet takes around 2 seconds per iteration, while the bumper design takes around 12 seconds per iteration due to a large DOF count. This high computational demand for the bumper can be reduced by exploiting symmetry and analyzing one quarter of the domain.

\begin{table}[!htbp]
    \caption{Computational time of FEA for representative design cases}
    \label{Table: Computational Time}
    \centering
    \footnotesize
    \begin{tabular}{lllll}
        \hline
        \textbf{Optimized designs} & \textbf{Element counts} & \textbf{DOF counts} & \textbf{Process counts} & \textbf{Time per iteration} \\
        \hline
        Damper in Fig. \ref{Fig: Damper-Part 3} & 15,000 & 30,502 & 8 & 0.07 s\\
        Beam (Dsg. 4) in Fig. \ref{Fig: Beam} & 14,400 & 29,402 & 8 & 0.06 s\\
        Bumper in Fig. \ref{Fig: Bumper-Part 3} & 138,400 & 469,491 & 16 & 11.82 s\\
        Profiled sheet (Dsg. 4) in Fig. \ref{Fig: Sheet-Part 2} & 28,800 & 110,715 & 16 & 1.94 s\\
        \hline
    \end{tabular}
\end{table}

\section{History-dependent sensitivity analysis and verification}
\label{Sec: Sensitivity Analysis and Verification}

In this section, we present the sensitivity analysis for finite strain elastoplasticity. Following the blueprint provided in \citet{jia_multimaterial_2025}, we begin by defining the global vectors for the design and state variables as well as the residuals. After that, we derive the sensitivity expressions with the reversed adjoint method and automatic differentiation. Finally, we verify the derived sensitivity expressions through a comparison with the forward finite difference scheme.

\subsection{Global vectors of design variables, state variables, and residuals}

In this subsection, we define the residual vectors used in the sensitivity analysis. Based on the local and global governing equations in Section \ref{Sec: Finite Strain Elastoplasticity}, we identify five groups of independent state variables as $\overline{\mathbf{b}}_1^\texttt{e}, \ldots, \overline{\mathbf{b}}_N^\texttt{e}$, $\overline{\boldsymbol{\beta}}_1, \ldots, \overline{\boldsymbol{\beta}}_N$, $\alpha_1, \ldots, \alpha_N$, $\widehat{\gamma}_1, \ldots, \widehat{\gamma}_N$, and $\mathbf{u}_1, \ldots, \mathbf{u}_N$. Correspondingly, we define five groups of residual expressions as
\begin{equation} \label{Residual Expresssions}
    \left\{ \begin{array}{l}
        \mathbf{r}_{n+1}^{\overline{\mathbf{b}}^\texttt{e}}(\{\overline{\zeta}\}, \overline{\mathbf{b}}_{n+1}^\texttt{e}, \widehat{\gamma}_{n+1}, \mathbf{u}_{n+1}, \overline{\mathbf{b}}_n^\texttt{e}, \overline{\boldsymbol{\beta}}_n, \mathbf{u}_n) = \overline{\mathbf{b}}_{n+1}^\texttt{e} - \text{dev} \left( \overline{\mathbf{b}}_{n+1}^\texttt{e,tr} \right)
        + \dfrac{2 \overline{\overline{\mu}}_{n+1}^\texttt{tr}}{\mu} \widehat{\gamma}_{n+1} \mathbf{n}_{n+1} - \dfrac{1}{3} \mathcal{I}_1 \mathbf{I} = \mathbf{0}, \\[12pt]

        \mathbf{r}_{n+1}^{\overline{\boldsymbol{\beta}}}(\{\overline{\zeta}\}, \overline{\boldsymbol{\beta}}_{n+1}, \widehat{\gamma}_{n+1}, \mathbf{u}_{n+1}, \overline{\mathbf{b}}_n^\texttt{e}, \overline{\boldsymbol{\beta}}_n, \mathbf{u}_n)
        = \overline{\boldsymbol{\beta}}_{n+1} - \overline{\boldsymbol{\beta}}_{n+1}^\texttt{tr}
        - \dfrac{2 h \overline{\overline{\mu}}_{n+1}^\texttt{tr}}{3 \mu} \widehat{\gamma}_{n+1} \mathbf{n}_{n+1} = \mathbf{0}, \\[12pt]

        r_{n+1}^{\alpha}(\alpha_{n+1}, \widehat{\gamma}_{n+1}, \alpha_n) = \alpha_{n+1} - \alpha_{n+1}^\texttt{tr} - \sqrt{\dfrac{2}{3}} \widehat{\gamma}_{n+1} = 0, \\[12pt]
        
        r_{n+1}^{\widehat{\gamma}}(\{\overline{\zeta}\}, \widehat{\gamma}_{n+1}, \mathbf{u}_{n+1}, \overline{\mathbf{b}}_n^\texttt{e}, \overline{\boldsymbol{\beta}}_n, \alpha_n, \mathbf{u}_n) \\[12pt]
        \quad = \widehat{\gamma}_{n+1} \mathcal{G}(\widehat{\gamma}_{n+1})
        = \widehat{\gamma}_{n+1} \left[ \lVert \boldsymbol{\xi}_{n+1}^\texttt{tr} \rVert
        - 2 \overline{\overline{\mu}}_{n+1}^\texttt{tr}
        \left( 1 + \dfrac{h}{3\mu} \right) \widehat{\gamma}_{n+1}
        - \sqrt{\dfrac{2}{3}} k \left( \alpha_{n+1}^\texttt{tr} + \sqrt{\dfrac{2}{3}} \widehat{\gamma}_{n+1} \right) \right] = 0, \\[12pt]
        
        \displaystyle r_{n+1}^\mathbf{u}(\{\overline{\zeta}\}, \widehat{\gamma}_{n+1}, \mathbf{u}_{n+1}, \overline{\mathbf{b}}_n^\texttt{e}, \overline{\boldsymbol{\beta}}_n, \mathbf{u}_n) = \int_{\Omega_0} \mathbf{P}_{n+1} : \nabla \mathbf{v}\ \text{d} \mathbf{X} - \int_{\Omega_0} \overline{\mathbf{q}}_{n+1} \cdot \mathbf{v}\ \text{d} \mathbf{X} 
        - \int_{\partial \Omega_0^\mathcal{N}} \overline{\mathbf{t}}_{n+1} \cdot \mathbf{v}\ \text{d} \mathbf{X} = 0,
    \end{array} \right.
\end{equation}
where $\{\overline{\zeta}\} = \{ \overline{\rho}, \overline{\xi}_1, \ldots, \overline{\xi}_{N^\texttt{mat}} \}$ is a collection of the physical design variables. We remark that, in contrast to \citet{jia_multimaterial_2025}, here we treat $\widehat{\gamma}_1, \ldots, \widehat{\gamma}_N$ as independent state variables and incorporate its residual in $\eqref{Residual Expresssions}_4$. This treatment allows us to consider nonlinear isotropic hardening laws such as \eqref{Nonlinear Isotropic Hardening Law} where $\widehat{\gamma}_1, \ldots, \widehat{\gamma}_N$ typically have no analytical expressions.

In the context of FEA, we reformulate \eqref{Residual Expresssions} into global residual vectors expressed as 
\begin{equation*}
    \left\{ \begin{array}{l}
        \mathbf{R}_{n+1}^{\overline{\mathbf{b}}^\texttt{e}}
        (\{\overline{\boldsymbol{\zeta}}\}, \mathbf{V}_{n+1}^{\overline{\mathbf{b}}^\texttt{e}}, \mathbf{V}_{n+1}^{\widehat{\gamma}}, \mathbf{U}_{n+1}, \mathbf{V}_n^{\overline{\mathbf{b}}^\texttt{e}}, \mathbf{V}_n^{\overline{\boldsymbol{\beta}}}, \mathbf{U}_n) = \mathbf{0}, \\[8pt]
        
        \mathbf{R}_{n+1}^{\overline{\boldsymbol{\beta}}}(\{\overline{\boldsymbol{\zeta}}\}, \mathbf{V}_{n+1}^{\overline{\boldsymbol{\beta}}}, \mathbf{V}_{n+1}^{\widehat{\gamma}}, \mathbf{U}_{n+1}, \mathbf{V}_n^{\overline{\mathbf{b}}^\texttt{e}}, \mathbf{V}_n^{\overline{\boldsymbol{\beta}}}, \mathbf{U}_n) = \mathbf{0}, \\[8pt]
        
        \mathbf{R}_{n+1}^{\alpha}(\mathbf{V}_{n+1}^\alpha, \mathbf{V}_{n+1}^{\widehat{\gamma}}, \mathbf{V}_n^\alpha) = \mathbf{0}, \\[8pt]
        
        \mathbf{R}_{n+1}^{\widehat{\gamma}}(\{\overline{\boldsymbol{\zeta}}\}, \mathbf{V}_{n+1}^{\widehat{\gamma}}, \mathbf{U}_{n+1}, \mathbf{V}_n^{\overline{\mathbf{b}}^\texttt{e}}, \mathbf{V}_n^{\overline{\boldsymbol{\beta}}}, \mathbf{V}_n^\alpha, \mathbf{U}_n) = \mathbf{0}, \\[8pt]
        
        \mathbf{R}_{n+1}^\mathbf{u}(\{\overline{\boldsymbol{\zeta}}\}, \mathbf{V}_{n+1}^{\widehat{\gamma}}, \mathbf{U}_{n+1}, \mathbf{V}_n^{\overline{\mathbf{b}}^\texttt{e}}, \mathbf{V}_n^{\overline{\boldsymbol{\beta}}}, \mathbf{U}_n) = \mathbf{0}.
    \end{array} \right.
\end{equation*}
In these expressions, we define $\{\overline{\boldsymbol{\zeta}}\} = \{ \overline{\boldsymbol{\rho}}, \overline{\boldsymbol{\xi}}_1, \ldots, \overline{\boldsymbol{\xi}}_{N^\texttt{mat}} \}$ as a collection of global vectors of physical design variables, where $\overline{\zeta}_e$ is the physical design variable of finite element $e$. Additionally, the variable $\mathbf{V}_n^{\overline{\mathbf{b}}^\texttt{e}}$ for $n=1,\ldots,N$ is the global vector of $\overline{\mathbf{b}}_n^\texttt{e}$ defined as
\begin{equation*}
    \mathbf{V}_n^{\overline{\mathbf{b}}^\texttt{e}} = \left[ \left(\overline{\mathbf{b}}_{n,1}^\texttt{e,v}\right)^\top, \ldots, \left(\overline{\mathbf{b}}_{n,N^\texttt{pt}}^\texttt{e,v}\right)^\top \right]^\top
\end{equation*}
where $\overline{\mathbf{b}}_{n,i}^\texttt{e,v}$ for $i = 1, \ldots, N^\texttt{pt}$ is the vector form of $\overline{\mathbf{b}}_{n,i}^\texttt{e}$ expressed as
\begin{equation*}
    \overline{\mathbf{b}}_{n,i}^\texttt{e,v} = \left[ \left( \overline{\mathbf{b}}_{n,i}^\texttt{e} \right)_{11}, \left( \overline{\mathbf{b}}_{n,i}^\texttt{e} \right)_{12}, \left( \overline{\mathbf{b}}_{n,i}^\texttt{e} \right)_{13}, \left( \overline{\mathbf{b}}_{n,i}^\texttt{e} \right)_{22}, \left( \overline{\mathbf{b}}_{n,i}^\texttt{e} \right)_{23}, \left( \overline{\mathbf{b}}_{n,i}^\texttt{e} \right)_{33} \right]^\top.
\end{equation*}
Here, subscript $i$ is the count of integration points used in FEA, and $N^\texttt{pt}$ is the number of all integration points. Similarly, the variable $\mathbf{V}_n^{\overline{\boldsymbol{\beta}}}$ for $n=1,\ldots,N$ is the global vector of $\overline{\boldsymbol{\beta}}_n$ defined as
\begin{equation*}
    \mathbf{V}_n^{\overline{\boldsymbol{\beta}}} = \left[ \left(\overline{\boldsymbol{\beta}}_{n,1}^\texttt{v}\right)^\top, \ldots, \left(\overline{\boldsymbol{\beta}}_{n,N^\texttt{pt}}^\texttt{v}\right)^\top \right]^\top
\end{equation*}
where $\overline{\boldsymbol{\beta}}_{n,i}^\texttt{v}$ for $i = 1, \ldots, N^\texttt{pt}$ is the vector form of $\overline{\boldsymbol{\beta}}_{n,i}$ expressed as
\begin{equation*}
    \overline{\boldsymbol{\beta}}_{n,i}^\texttt{v} = \left[ \left( \overline{\boldsymbol{\beta}}_{n,i} \right)_{11},
    \left( \overline{\boldsymbol{\beta}}_{n,i} \right)_{12},
    \left( \overline{\boldsymbol{\beta}}_{n,i} \right)_{13},
    \left( \overline{\boldsymbol{\beta}}_{n,i} \right)_{22},
    \left( \overline{\boldsymbol{\beta}}_{n,i} \right)_{23},
    \left( \overline{\boldsymbol{\beta}}_{n,i} \right)_{33} \right]^\top.
\end{equation*}
The variables $\mathbf{V}_n^\alpha$ and $\mathbf{V}_n^{\widehat{\gamma}}$ for $n=1,\ldots,N$ are the global vectors of $\alpha_n$ and $\widehat{\gamma}_n$, respectively, which are defined as
\begin{equation*}
    \mathbf{V}_n^\alpha = \left[ \alpha_{n,1}, \ldots, \alpha_{n,N^\texttt{pt}} \right]^\top
    \quad \text{and} \quad
    \mathbf{V}_n^{\widehat{\gamma}} = \left[ \widehat{\gamma}_{n,1}, \ldots, \widehat{\gamma}_{n,N^\texttt{pt}} \right]^\top.
\end{equation*}
Finally, the variable $\mathbf{U}_n$ for $n=1,\ldots,N$ is the global vector of $\mathbf{u}_n$ defined as
\begin{equation*}
    \mathbf{U}_n = \left[ (\mathbf{u}_n)_{1,x}, (\mathbf{u}_n)_{1,y}, (\mathbf{u}_n)_{1,z}, \ldots, (\mathbf{u}_n)_{N^\texttt{node},x}, (\mathbf{u}_n)_{N^\texttt{node},y}, (\mathbf{u}_n)_{N^\texttt{node},z} \right]^\top
\end{equation*}
where subscripts $x$, $y$, and $z$ represent the three components of $(\mathbf{u}_n)_j$ for $j=1,\ldots,N^\texttt{node}$, and $N^\texttt{node}$ is the number of nodes in FEA.

\subsection{Sensitivity analysis}
\label{Sec: Sensitivity Analysis}

After defining the global vectors of design variables ($\overline{\boldsymbol{\rho}}$ and $\overline{\boldsymbol{\xi}}_1, \ldots, \overline{\boldsymbol{\xi}}_{N^\texttt{mat}}$), state variables ($\mathbf{V}_n^{\overline{\mathbf{b}}^\texttt{e}}$, $\mathbf{V}_n^{\overline{\boldsymbol{\beta}}}$, $\mathbf{V}_n^\alpha$, $\mathbf{V}_n^{\widehat{\gamma}}$, and $\mathbf{U}_n$), and residuals ($\mathbf{R}_{n+1}^{\overline{\mathbf{b}}^\texttt{e}}$, $\mathbf{R}_{n+1}^{\overline{\boldsymbol{\beta}}}$, $\mathbf{R}_{n+1}^{\alpha}$, $\mathbf{R}_{n+1}^{\widehat{\gamma}}$, and $\mathbf{R}_{n+1}^\mathbf{u}$), we proceed with the sensitivity analysis for finite strain elastoplasticity. We consider a general function
\begin{equation*}
    \mathcal{F}(\overline{\boldsymbol{\rho}}, \overline{\boldsymbol{\xi}}_1, \ldots, \overline{\boldsymbol{\xi}}_{N^\texttt{mat}}; \mathbf{V}_1^{\overline{\mathbf{b}}^\texttt{e}}, \ldots, \mathbf{V}_N^{\overline{\mathbf{b}}^\texttt{e}}; \mathbf{V}_1^{\overline{\boldsymbol{\beta}}}, \ldots, \mathbf{V}_N^{\overline{\boldsymbol{\beta}}}; \mathbf{V}_1^\alpha, \ldots, \mathbf{V}_n^\alpha; \mathbf{V}_N^{\widehat{\gamma}}, \ldots, \mathbf{V}_1^{\widehat{\gamma}}; \mathbf{U}_1, \ldots, \mathbf{U}_N)
\end{equation*}
and write out its Lagrangian expression as
\begin{equation*}
    \widehat{\mathcal{F}} = \mathcal{F} + \sum_{n=1}^N \left( \boldsymbol{\lambda}_n^{\overline{\mathbf{b}}^\texttt{e}} \cdot \mathbf{R}_n^{\overline{\mathbf{b}}^\texttt{e}}
    + \boldsymbol{\lambda}_n^{\overline{\boldsymbol{\beta}}} \cdot \mathbf{R}_n^{\overline{\boldsymbol{\beta}}}
    + \boldsymbol{\lambda}_n^\alpha \cdot \mathbf{R}_n^\alpha
    + \boldsymbol{\lambda}_n^{\widehat{\gamma}} \cdot \mathbf{R}_n^{\widehat{\gamma}}
    + \boldsymbol{\lambda}_n^\mathbf{u} \cdot \mathbf{R}_n^\mathbf{u} \right)
    = \mathcal{F} + \sum_{n=1}^N \boldsymbol{\lambda}_n \cdot \mathbf{R}_n,
\end{equation*}
where we define
\begin{equation*}
    \boldsymbol{\lambda}_n = \left[
    \left( \boldsymbol{\lambda}_n^{\overline{\mathbf{b}}^\texttt{e}} \right)^\top,
    \left( \boldsymbol{\lambda}_n^{\overline{\boldsymbol{\beta}}} \right)^\top,
    \left( \boldsymbol{\lambda}_n^\alpha \right)^\top,
    \left( \boldsymbol{\lambda}_n^{\widehat{\gamma}} \right)^\top,
    \left( \boldsymbol{\lambda}_n^\mathbf{u} \right)^\top
    \right]^\top
    \quad \text{for} \quad
    n = 1,\ldots,N
\end{equation*}
and
\begin{equation*}
    \mathbf{R}_n = \left[
    \left( \mathbf{R}_n^{\overline{\mathbf{b}}^\texttt{e}} \right)^\top,
    \left( \mathbf{R}_n^{\overline{\boldsymbol{\beta}}} \right)^\top,
    \left( \mathbf{R}_n^\alpha \right)^\top,
    \left( \mathbf{R}_n^{\widehat{\gamma}} \right)^\top,
    \left( \mathbf{R}_n^\mathbf{u} \right)^\top
    \right]^\top
    \quad \text{for} \quad
    n = 1,\ldots,N.
\end{equation*}
The variables $\boldsymbol{\lambda}_n^{\overline{\mathbf{b}}^\texttt{e}}$, $\boldsymbol{\lambda}_n^{\overline{\boldsymbol{\beta}}}$, $\boldsymbol{\lambda}_n^\alpha$, $\boldsymbol{\lambda}_n^{\widehat{\gamma}}$, and $\boldsymbol{\lambda}_n^\mathbf{u}$ for $n=1,\ldots,N$ are the adjoint vectors that take arbitrary real values for now.

Taking the derivative of $\widehat{\mathcal{F}}$ with respect to $\overline{\boldsymbol{\zeta}} \in \{ \overline{\boldsymbol{\rho}}, \overline{\boldsymbol{\xi}}_1, \ldots, \overline{\boldsymbol{\xi}}_{N^\texttt{mat}} \}$ derives
\begin{equation*}
    \dfrac{\text{d} \widehat{\mathcal{F}}}{\text{d} \overline{\boldsymbol{\zeta}}}
    = \dfrac{\partial \mathcal{F}}{\partial \overline{\boldsymbol{\zeta}}}
    + \sum_{n=1}^N \left( \dfrac{\partial \mathbf{V}_n}{\partial \overline{\boldsymbol{\zeta}}} \right)^\top \dfrac{\partial \mathcal{F}}{\partial \mathbf{V}_n}
    + \sum_{n=1}^N \left[ \left( \dfrac{\partial \mathbf{R}_n}{\partial \overline{\boldsymbol{\zeta}}} \right)^\top
    + \left( \dfrac{\partial \mathbf{V}_{n-1}}{\partial \overline{\boldsymbol{\zeta}}} \right)^\top \left( \dfrac{\partial \mathbf{R}_n}{\partial \mathbf{V}_{n-1}} \right)^\top
    + \left( \dfrac{\partial \mathbf{V}_n}{\partial \overline{\boldsymbol{\zeta}}} \right)^\top \left( \dfrac{\partial \mathbf{R}_n}{\partial \mathbf{V}_n} \right)^\top
    \right] \boldsymbol{\lambda}_n.
\end{equation*}
Here we define
\begin{equation*}
    \mathbf{V}_n = \left[
    \left( \mathbf{V}_n^{\overline{\mathbf{b}}^\texttt{e}} \right)^\top,
    \left( \mathbf{V}_n^{\overline{\boldsymbol{\beta}}} \right)^\top,
    \left( \mathbf{V}_n^\alpha \right)^\top,
    \left( \mathbf{V}_n^{\widehat{\gamma}} \right)^\top,
    \left( \mathbf{U}_n \right)^\top
    \right]^\top
    \quad \text{for} \quad
    n = 0,\ldots,N,
\end{equation*}
and note that $\mathbf{V}_0$ represents the initial conditions.

Regrouping the terms in $\text{d} \widehat{\mathcal{F}}/\text{d} \overline{\boldsymbol{\zeta}}$ renders
\begin{equation} \label{Regrouping Terms for Sensitivity}
    \begin{array}{ll}
        \dfrac{\text{d} \widehat{\mathcal{F}}}{\text{d} \overline{\boldsymbol{\zeta}}}
        = &\displaystyle \dfrac{\partial \mathcal{F}}{\partial \overline{\boldsymbol{\zeta}}}
        + \sum_{n=1}^N \left( \dfrac{\partial \mathbf{R}_n}{\partial \overline{\boldsymbol{\zeta}}} \right)^\top \boldsymbol{\lambda}_n
        + \sum_{n=1}^{N-1} \left( \dfrac{\partial \mathbf{V}_n}{\partial \overline{\boldsymbol{\zeta}}} \right)^\top 
        \left[ \left( \dfrac{\partial \mathbf{R}_n}{\partial \mathbf{V}_n} \right)^\top \boldsymbol{\lambda}_n
        + \left( \dfrac{\partial \mathbf{R}_{n+1}}{\partial \mathbf{V}_n} \right)^\top \boldsymbol{\lambda}_{n+1}
        + \dfrac{\partial \mathcal{F}}{\partial \mathbf{V}_n} \right] \\[15pt]
        &+ \displaystyle \left( \dfrac{\partial \mathbf{V}_N}{\partial \overline{\boldsymbol{\zeta}}} \right)^\top
        \left[ \left( \dfrac{\partial \mathbf{R}_N}{\partial \mathbf{V}_N} \right)^\top \boldsymbol{\lambda}_N + \dfrac{\partial \mathcal{F}}{\partial \mathbf{V}_N} \right],
    \end{array}
\end{equation}
where we use $\partial \mathbf{V}_0/\partial \overline{\boldsymbol{\zeta}} = \mathbf{0}$. Note again that $\boldsymbol{\lambda}_1, \ldots, \boldsymbol{\lambda}_N$ are arbitrary, and we select their values to eliminate computationally expensive term, $\partial \mathbf{V}_n/\partial \overline{\boldsymbol{\zeta}} = \mathbf{0}$ for $n = 1,\ldots,N$. This step leads to the adjoint equations expressed as
\begin{equation} \label{Adjoint Equations}
    \left\{ \begin{array}{l}
        \left( \dfrac{\partial \mathbf{R}_n}{\partial \mathbf{V}_n} \right)^\top \boldsymbol{\lambda}_n
        + \left( \dfrac{\partial \mathbf{R}_{n+1}}{\partial \mathbf{V}_n} \right)^\top \boldsymbol{\lambda}_{n+1}
        = - \dfrac{\partial \mathcal{F}}{\partial \mathbf{V}_n} \quad \text{for} \quad n = 1,\ldots,N-1, \\[12pt]
        
        \left( \dfrac{\partial \mathbf{R}_N}{\partial \mathbf{V}_N} \right)^\top \boldsymbol{\lambda}_N = - \dfrac{\partial \mathcal{F}}{\partial \mathbf{V}_N},
    \end{array} \right.
\end{equation}
where
\begin{equation*}
    \dfrac{\partial \mathbf{R}_n}{\partial \mathbf{V}_n} = \left[ \begin{array}{ccccc}
        \mathbf{I} & \mathbf{0} & \mathbf{0} & \dfrac{\partial \mathbf{R}_n^{\overline{\mathbf{b}}^\texttt{e}}}{\partial \mathbf{V}_n^{\widehat{\gamma}}} & \dfrac{\partial \mathbf{R}_n^{\overline{\mathbf{b}}^\texttt{e}}}{\partial \mathbf{U}_n} \\[12pt]
        
        \mathbf{0} & \mathbf{I} & \mathbf{0} & \dfrac{\partial \mathbf{R}_n^{\overline{\boldsymbol{\beta}}}}{\partial \mathbf{V}_n^{\widehat{\gamma}}} & \dfrac{\partial \mathbf{R}_n^{\overline{\boldsymbol{\beta}}}}{\partial \mathbf{U}_n} \\[12pt]

        \mathbf{0} & \mathbf{0} & \mathbf{I} & \dfrac{\partial \mathbf{R}_n^\alpha}{\partial \mathbf{V}_n^{\widehat{\gamma}}} & \mathbf{0} \\[12pt]

        \mathbf{0} & \mathbf{0} & \mathbf{0} & \dfrac{\partial \mathbf{R}_n^{\widehat{\gamma}}}{\partial \mathbf{V}_n^{\widehat{\gamma}}} & \dfrac{\partial \mathbf{R}_n^{\widehat{\gamma}}}{\partial \mathbf{U}_n} \\[12pt]

        \mathbf{0} & \mathbf{0} & \mathbf{0} & \dfrac{\partial \mathbf{R}_n^\mathbf{u}}{\partial \mathbf{V}_n^{\widehat{\gamma}}} & \dfrac{\partial \mathbf{R}_n^\mathbf{u}}{\partial \mathbf{U}_n}
    \end{array} \right]
    \quad \text{and} \quad
        \dfrac{\partial \mathbf{R}_{n+1}}{\partial \mathbf{V}_n} = \left[ \begin{array}{ccccc}
        \dfrac{\partial \mathbf{R}_{n+1}^{\overline{\mathbf{b}}^\texttt{e}}}{\partial \mathbf{V}_n^{\overline{\mathbf{b}}^\texttt{e}}} & \dfrac{\partial \mathbf{R}_{n+1}^{\overline{\mathbf{b}}^\texttt{e}}}{\partial \mathbf{V}_n^{\overline{\boldsymbol{\beta}}}} & \mathbf{0} & \mathbf{0} & \dfrac{\partial \mathbf{R}_{n+1}^{\overline{\mathbf{b}}^\texttt{e}}}{\partial \mathbf{U}_n} \\[12pt]
        
        \dfrac{\partial \mathbf{R}_{n+1}^{\overline{\boldsymbol{\beta}}}}{\partial \mathbf{V}_n^{\overline{\mathbf{b}}^\texttt{e}}} & \dfrac{\partial \mathbf{R}_{n+1}^{\overline{\boldsymbol{\beta}}}}{\partial \mathbf{V}_n^{\overline{\boldsymbol{\beta}}}} & \mathbf{0} & \mathbf{0} & \dfrac{\partial \mathbf{R}_{n+1}^{\overline{\boldsymbol{\beta}}}}{\partial \mathbf{U}_n} \\[12pt]

        \mathbf{0} & \mathbf{0} & \dfrac{\partial \mathbf{R}_{n+1}^\alpha}{\partial \mathbf{V}_n^\alpha} & \mathbf{0} & \mathbf{0} \\[12pt]

        \dfrac{\partial \mathbf{R}_{n+1}^{\widehat{\gamma}}}{\partial \mathbf{V}_n^{\overline{\mathbf{b}}^\texttt{e}}} & \dfrac{\partial \mathbf{R}_{n+1}^{\widehat{\gamma}}}{\partial \mathbf{V}_n^{\overline{\boldsymbol{\beta}}}} & \dfrac{\partial \mathbf{R}_{n+1}^{\widehat{\gamma}}}{\partial \mathbf{V}_n^\alpha} & \mathbf{0} & \dfrac{\partial \mathbf{R}_{n+1}^{\widehat{\gamma}}}{\partial \mathbf{U}_n} \\[12pt]

        \dfrac{\partial \mathbf{R}_{n+1}^\mathbf{u}}{\partial \mathbf{V}_n^{\overline{\mathbf{b}}^\texttt{e}}} & \dfrac{\partial \mathbf{R}_{n+1}^\mathbf{u}}{\partial \mathbf{V}_n^{\overline{\boldsymbol{\beta}}}} & \mathbf{0} & \mathbf{0} & \dfrac{\partial \mathbf{R}_{n+1}^\mathbf{u}}{\partial \mathbf{U}_n}
    \end{array} \right].
\end{equation*}

Based on \eqref{Adjoint Equations}, we compute the adjoint variables ($\boldsymbol{\lambda}_n^{\overline{\mathbf{b}}^\texttt{e}}$, $\boldsymbol{\lambda}_n^{\overline{\boldsymbol{\beta}}}$, $\boldsymbol{\lambda}_n^\alpha$, $\boldsymbol{\lambda}_n^{\widehat{\gamma}}$, and $\boldsymbol{\lambda}_n^\mathbf{u}$ for $n = 1,\ldots,N$) in a reversed order as follows. At load step $N$, we directly write out $\boldsymbol{\lambda}_N^{\overline{\mathbf{b}}^\texttt{e}}$, $\boldsymbol{\lambda}_N^{\overline{\boldsymbol{\beta}}}$, and $\boldsymbol{\lambda}_N^\alpha$ as
\begin{equation*}
    \boldsymbol{\lambda}_N^{\overline{\mathbf{b}}^\texttt{e}} = - \dfrac{\partial \mathcal{F}}{\partial \mathbf{V}_N^{\overline{\mathbf{b}}^\texttt{e}}}, \quad
    \boldsymbol{\lambda}_N^{\overline{\boldsymbol{\beta}}} = - \dfrac{\partial \mathcal{F}}{\partial \mathbf{V}_N^{\overline{\boldsymbol{\beta}}}}, \quad \text{and} \quad
    \boldsymbol{\lambda}_N^\alpha = - \dfrac{\partial \mathcal{F}}{\partial \mathbf{V}_N^\alpha}.
\end{equation*}
We then solve for $\boldsymbol{\lambda}_N^{\widehat{\gamma}}$ and $\boldsymbol{\lambda}_N^\mathbf{u}$ from 
\begin{equation} \label{Reduced Adjoint Equation at Final Step}
    \left[ \begin{array}{cc}
        \left( \dfrac{\partial \mathbf{R}_N^{\widehat{\gamma}}}{\partial \mathbf{V}_N^{\widehat{\gamma}}} \right)^\top
        & \left( \dfrac{\partial \mathbf{R}_N^\mathbf{u}}{\partial \mathbf{V}_N^{\widehat{\gamma}}} \right)^\top \\[12pt]
        \left( \dfrac{\partial \mathbf{R}_N^{\widehat{\gamma}}}{\partial \mathbf{U}_N} \right)^\top
        & \left( \dfrac{\partial \mathbf{R}_N^\mathbf{u}}{\partial \mathbf{U}_N} \right)^\top
    \end{array} \right]
    \left[ \begin{array}{c}
        \boldsymbol{\lambda}_N^{\widehat{\gamma}} \\[12pt]
        \boldsymbol{\lambda}_N^\mathbf{u}
    \end{array} \right]
    = - \left[ \begin{array}{c}
         \dfrac{\partial \mathcal{F}}{\partial \mathbf{V}_N^{\widehat{\gamma}}}
         + \left( \dfrac{\partial \mathbf{R}_N^{\overline{\mathbf{b}}^\texttt{e}}}{\partial \mathbf{V}_N^{\widehat{\gamma}}} \right)^\top \boldsymbol{\lambda}_N^{\overline{\mathbf{b}}^\texttt{e}}
         + \left( \dfrac{\partial \mathbf{R}_N^{\overline{\boldsymbol{\beta}}}}{\partial \mathbf{V}_N^{\widehat{\gamma}}} \right)^\top \boldsymbol{\lambda}_N^{\overline{\boldsymbol{\beta}}}
         + \left( \dfrac{\partial \mathbf{R}_N^\alpha}{\partial \mathbf{V}_N^{\widehat{\gamma}}} \right)^\top \boldsymbol{\lambda}_N^\alpha \\[12pt]
         \dfrac{\partial \mathcal{F}}{\partial \mathbf{U}_N}
         + \left( \dfrac{\partial \mathbf{R}_N^{\overline{\mathbf{b}}^\texttt{e}}}{\partial \mathbf{U}_N} \right)^\top \boldsymbol{\lambda}_N^{\overline{\mathbf{b}}^\texttt{e}}
         + \left( \dfrac{\partial \mathbf{R}_N^{\overline{\boldsymbol{\beta}}}}{\partial \mathbf{U}_N} \right)^\top \boldsymbol{\lambda}_N^{\overline{\boldsymbol{\beta}}}
    \end{array} \right].
\end{equation}
At the remaining load steps, $n = N-1, \ldots, 1$, we write out $\boldsymbol{\lambda}_n^{\overline{\mathbf{b}}^\texttt{e}}$, $\boldsymbol{\lambda}_n^{\overline{\boldsymbol{\beta}}}$, and $\boldsymbol{\lambda}_n^\alpha$ as
\begin{equation*}
    \left\{ \begin{array}{l}
        \boldsymbol{\lambda}_n^{\overline{\mathbf{b}}^\texttt{e}}
        = - \left[ \dfrac{\partial \mathcal{F}}{\partial \mathbf{V}_n^{\overline{\mathbf{b}}^\texttt{e}}} 
        + \left( \dfrac{\partial \mathbf{R}_{n+1}^{\overline{\mathbf{b}}^\texttt{e}}}{\partial \mathbf{V}_n^{\overline{\mathbf{b}}^\texttt{e}}} \right)^\top \boldsymbol{\lambda}_{n+1}^{\overline{\mathbf{b}}^\texttt{e}}
        + \left( \dfrac{\partial \mathbf{R}_{n+1}^{\overline{\boldsymbol{\beta}}}}{\partial \mathbf{V}_n^{\overline{\mathbf{b}}^\texttt{e}}} \right)^\top \boldsymbol{\lambda}_{n+1}^{\overline{\boldsymbol{\beta}}}
        + \left( \dfrac{\partial \mathbf{R}_{n+1}^{\widehat{\gamma}}}{\partial \mathbf{V}_n^{\overline{\mathbf{b}}^\texttt{e}}} \right)^\top
        \boldsymbol{\lambda}_{n+1}^{\widehat{\gamma}}
        + \left( \dfrac{\partial \mathbf{R}_{n+1}^\mathbf{u}}{\partial \mathbf{V}_n^{\overline{\mathbf{b}}^\texttt{e}}} \right)^\top \boldsymbol{\lambda}_{n+1}^\mathbf{u}
        \right], \\[18pt]

        \boldsymbol{\lambda}_n^{\overline{\boldsymbol{\beta}}}
        = - \left[ \dfrac{\partial \mathcal{F}}{\partial \mathbf{V}_n^{\overline{\boldsymbol{\beta}}}}
        + \left( \dfrac{\partial \mathbf{R}_{n+1}^{\overline{\mathbf{b}}^\texttt{e}}}{\partial \mathbf{V}_n^{\overline{\boldsymbol{\beta}}}} \right)^\top \boldsymbol{\lambda}_{n+1}^{\overline{\mathbf{b}}^\texttt{e}}
        + \left( \dfrac{\partial \mathbf{R}_{n+1}^{\overline{\boldsymbol{\beta}}}}{\partial \mathbf{V}_n^{\overline{\boldsymbol{\beta}}}} \right)^\top \boldsymbol{\lambda}_{n+1}^{\overline{\boldsymbol{\beta}}}
        + \left( \dfrac{\partial \mathbf{R}_{n+1}^{\widehat{\gamma}}}{\partial \mathbf{V}_n^{\overline{\boldsymbol{\beta}}}} \right)^\top \boldsymbol{\lambda}_{n+1}^{\widehat{\gamma}}
        + \left( \dfrac{\partial \mathbf{R}_{n+1}^\mathbf{u}}{\partial \mathbf{V}_n^{\overline{\boldsymbol{\beta}}}} \right)^\top \boldsymbol{\lambda}_{n+1}^\mathbf{u}
        \right], \\[18pt]

        \boldsymbol{\lambda}_n^\alpha = - \left[ \dfrac{\partial \mathcal{F}}{\partial \mathbf{V}_n^\alpha}
        + \left( \dfrac{\partial \mathbf{R}_{n+1}^\alpha}{\partial \mathbf{V}_n^\alpha} \right)^\top \boldsymbol{\lambda}_{n+1}^\alpha
        + \left( \dfrac{\partial \mathbf{R}_{n+1}^{\widehat{\gamma}}}{\partial \mathbf{V}_n^\alpha} \right)^\top \boldsymbol{\lambda}_{n+1}^{\widehat{\gamma}}
        \right],
    \end{array} \right.
\end{equation*}
and then solve for $\boldsymbol{\lambda}_n^{\widehat{\gamma}}$ and $\boldsymbol{\lambda}_n^\mathbf{u}$ from 
\begin{equation} \label{Reduced Adjoint Equation at Remaining Steps}
    \left[ \begin{array}{cc}
        \left( \dfrac{\partial \mathbf{R}_n^{\widehat{\gamma}}}{\partial \mathbf{V}_n^{\widehat{\gamma}}} \right)^\top
        & \left( \dfrac{\partial \mathbf{R}_n^\mathbf{u}}{\partial \mathbf{V}_n^{\widehat{\gamma}}} \right)^\top \\[12pt]
        \left( \dfrac{\partial \mathbf{R}_n^{\widehat{\gamma}}}{\partial \mathbf{U}_n} \right)^\top
        & \left( \dfrac{\partial \mathbf{R}_n^\mathbf{u}}{\partial \mathbf{U}_n} \right)^\top
    \end{array} \right]
    \left[ \begin{array}{c}
        \boldsymbol{\lambda}_n^{\widehat{\gamma}} \\[12pt]
        \boldsymbol{\lambda}_n^\mathbf{u}
    \end{array} \right]
    = - \left[ \begin{array}{c}
         \dfrac{\partial \mathcal{F}}{\partial \mathbf{V}_n^{\widehat{\gamma}}}
         + \left( \dfrac{\partial \mathbf{R}_n^{\overline{\mathbf{b}}^\texttt{e}}}{\partial \mathbf{V}_n^{\widehat{\gamma}}} \right)^\top \boldsymbol{\lambda}_n^{\overline{\mathbf{b}}^\texttt{e}}
         + \left( \dfrac{\partial \mathbf{R}_n^{\overline{\boldsymbol{\beta}}}}{\partial \mathbf{V}_n^{\widehat{\gamma}}} \right)^\top \boldsymbol{\lambda}_n^{\overline{\boldsymbol{\beta}}}
         + \left( \dfrac{\partial \mathbf{R}_n^\alpha}{\partial \mathbf{V}_n^{\widehat{\gamma}}} \right)^\top \boldsymbol{\lambda}_n^\alpha \\[12pt]
         \dfrac{\partial \mathcal{F}}{\partial \mathbf{U}_n}
         + \dfrac{\partial \widecheck{\mathcal{F}}}{\partial \mathbf{U}_n}
         + \left( \dfrac{\partial \mathbf{R}_n^{\overline{\mathbf{b}}^\texttt{e}}}{\partial \mathbf{U}_n} \right)^\top \boldsymbol{\lambda}_n^{\overline{\mathbf{b}}^\texttt{e}}
         + \left( \dfrac{\partial \mathbf{R}_n^{\overline{\boldsymbol{\beta}}}}{\partial \mathbf{U}_n} \right)^\top \boldsymbol{\lambda}_n^{\overline{\boldsymbol{\beta}}}
    \end{array} \right]
\end{equation}
where
\begin{equation*}
    \dfrac{\partial \widecheck{\mathcal{F}}}{\partial \mathbf{U}_n}
    := \left( \dfrac{\partial \mathbf{R}_{n+1}^{\overline{\mathbf{b}}^\texttt{e}}}{\partial \mathbf{U}_n} \right)^\top \boldsymbol{\lambda}_{n+1}^{\overline{\mathbf{b}}^\texttt{e}}
    + \left( \dfrac{\partial \mathbf{R}_{n+1}^{\overline{\boldsymbol{\beta}}}}{\partial \mathbf{U}_n} \right)^\top \boldsymbol{\lambda}_{n+1}^{\overline{\boldsymbol{\beta}}}
    + \left( \dfrac{\partial \mathbf{R}_{n+1}^{\widehat{\gamma}}}{\partial \mathbf{U}_n} \right)^\top  \boldsymbol{\lambda}_{n+1}^{\widehat{\gamma}}
    + \left( \dfrac{\partial \mathbf{R}_{n+1}^\mathbf{u}}{\partial \mathbf{U}_n} \right)^\top \boldsymbol{\lambda}_{n+1}^\mathbf{u}.
\end{equation*}

After solving for the adjoint variables, we substitute them into \eqref{Regrouping Terms for Sensitivity} and derive the sensitivity expression as
\begin{equation*}
    \begin{array}{ll}
        \dfrac{\text{d} \mathcal{F}}{\text{d} \overline{\boldsymbol{\zeta}}}
        = \dfrac{\partial \mathcal{F}}{\partial \overline{\boldsymbol{\zeta}}}
        &+ \displaystyle \sum_{n=1}^N \left[ \left( \dfrac{\partial \mathbf{R}_n^{\overline{\mathbf{b}}^\texttt{e}}}{\partial \overline{\boldsymbol{\zeta}}} \right)^\top \boldsymbol{\lambda}_n^{\overline{\mathbf{b}}^\texttt{e}}
        + \left( \dfrac{\partial \mathbf{R}_n^{\overline{\boldsymbol{\beta}}}}{\partial \overline{\boldsymbol{\zeta}}} \right)^\top \boldsymbol{\lambda}_n^{\overline{\boldsymbol{\beta}}}
        + \left( \dfrac{\partial \mathbf{R}_n^{\widehat{\gamma}}}{\partial \overline{\boldsymbol{\zeta}}} \right)^\top  \boldsymbol{\lambda}_n^{\widehat{\gamma}}
        + \left( \dfrac{\partial \mathbf{R}_n^\mathbf{u}}{\partial \overline{\boldsymbol{\zeta}}} \right)^\top \boldsymbol{\lambda}_n^\mathbf{u} \right],
    \end{array}
\end{equation*}
where we use $\partial \mathbf{R}_n^\alpha / \partial \overline{\boldsymbol{\zeta}} = \mathbf{0}$ and $\text{d} \mathcal{F}/\text{d} \overline{\boldsymbol{\zeta}} = \text{d} \widehat{\mathcal{F}}/\text{d} \overline{\boldsymbol{\zeta}}$ by noticing 
\begin{equation} \label{Sensitivity Expression}
    \boldsymbol{\lambda}_n^{\overline{\mathbf{b}}^\texttt{e}} \cdot \mathbf{R}_n^{\overline{\mathbf{b}}^\texttt{e}}
    = \boldsymbol{\lambda}_n^{\overline{\boldsymbol{\beta}}} \cdot \mathbf{R}_n^{\overline{\boldsymbol{\beta}}}
    = \boldsymbol{\lambda}_n^\alpha \cdot \mathbf{R}_n^\alpha
    = \boldsymbol{\lambda}_n^{\widehat{\gamma}} \cdot \mathbf{R}_n^{\widehat{\gamma}}
    = \boldsymbol{\lambda}_n^\mathbf{u} \cdot \mathbf{R}_n^\mathbf{u} = 0.
\end{equation}
We remark that, by using the modern computer method of automatic differentiation, one can effortlessly compute all the above partial derivatives without deriving their explicit (and typically complex) expressions.

\subsection{Sensitivity verification}

In this subsection, we verify the accuracy of the sensitivity analysis described in \ref{Sec: Sensitivity Analysis} by comparing it against the forward finite difference method. Using the finite difference scheme, the sensitivity of any function $\mathcal{F}$ with respect to the physical design variable $\overline{\zeta}_e$ of element $e$ for $\overline{\zeta} \in \{ \overline{\rho}, \overline{\xi}_1, \ldots, \overline{\xi}_{N^\xi} \}$ is computed as
\begin{equation*}
    \dfrac{\text{d} \mathcal{F}}{\text{d} \overline{\zeta}_e} = \dfrac{\mathcal{F}(\overline{\zeta}_e + \varepsilon) - \mathcal{F}(\overline{\zeta}_e)}{\varepsilon}
\end{equation*}
where $\varepsilon \in [10^{-6}, 10^{-4}]$ is a small perturbation parameter.

According to this formula, we compare the sensitivities of all objective functions ($J_\texttt{stiff}$, $J_\texttt{force}$, and $J_\texttt{energy}$) and constraint functions ($g_{V0}$, $g_{V1}$, $g_{V2}$, $g_{V3}$, $g_{V4}$, $g_P$, $g_M$, and $g_C$) employed in Section \ref{Sec: Sample Examples}. For representativeness, the comparisons are focused on the most complex cases: $J_\texttt{energy}$ for the optimized damper in Fig. \ref{Fig: Damper-Part 3}; $J_\texttt{stiff}$ and $g_{V0}$ for Dsg. 4 in Fig. \ref{Fig: Beam}; $g_{V1}$, $g_{V2}$, $g_{V3}$, and $g_{V4}$ for the optimized bumper in Fig. \ref{Fig: Bumper-Part 3}; and $J_\texttt{force}$, $g_P$, $g_M$, and $g_C$ for Dsg. 4 in Fig. \ref{Fig: Sheet-Part 2}.

To perform the comparison, we use stratified sampling of all design variables and present the results in Fig. \ref{Fig: Sensitivity Verification}. The proposed sensitivity analysis in \ref{Sec: Sensitivity Analysis} shows good agreement with the finite difference scheme. The observed absolute errors ($e_\texttt{abs}$) are in the range of $[10^{-11}, 10^{-6}]$, while the relative errors ($e_\texttt{rel}$) are in the range of $[10^{-8}, 10^{-4}]$. This verification demonstrates the reliability of the proposed sensitivity analysis, irrespective of the complexity introduced by the history dependence of finite strain elastoplasticity.

\begin{figure}[!htbp]
    \centering
    \includegraphics[width=18cm]{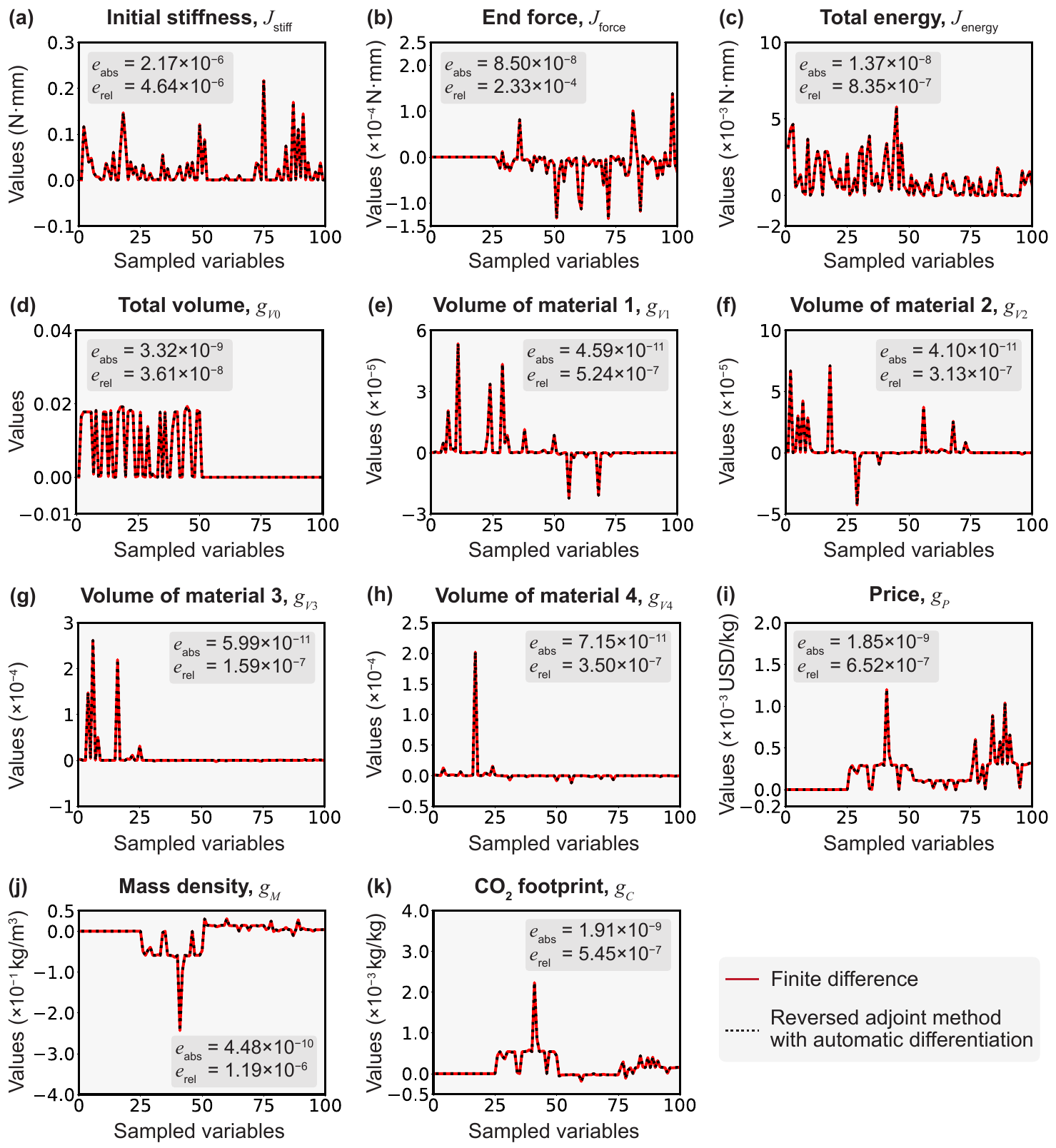}
    \caption{Sensitivity verification for various functions used in Section \ref{Sec: Sample Examples}. (a) Initial energy of Dsg. 4 in Fig. \ref{Fig: Beam}. (b) End compliance of Dsg. 4 in Fig. \ref{Fig: Sheet-Part 2}. (c) Total energy of the optimized damper in Fig. \ref{Fig: Damper-Part 3}. (d) Total material volume of Dsg. 4 in Fig. \ref{Fig: Beam}. (e)--(h) Volumes of materials 1--4 of the optimized bumper in Fig. \ref{Fig: Bumper-Part 3}. (i)--(k) Price, mass density, and CO$_2$ footprint of Dsg. 4 in Fig. \ref{Fig: Sheet-Part 2}.}
    \label{Fig: Sensitivity Verification}
\end{figure}

\section{Comparisons among various mesh and element combinations}
\label{Sec: Comparison of Meshes}

In this section, we compare the elastoplastic responses of structures defined with various meshes and elements. The proposed topology optimization framework favors a voxel-based mesh, where the grid in the undeformed configurations of structures is fixed during optimization, and solid and void elements coexist. However, for practical manufacturing purposes, a body-fitted mesh is more suitable for generating printable files (e.g., the STereoLithography format), as the mesh boundaries closely align with structural geometries and only contain solid elements.

To evaluate the accuracy loss when converting optimized designs from voxel-based to body-fitted meshes, we take Dsg. 5 in Fig. \ref{Fig: Beam} as an example and compare the elastoplastic responses in Fig. \ref{Fig: Mesh Comparisons} and Table \ref{Table: Mesh Comparisons}. The results show that, compared to the design defined on a voxel-based mesh with quadrilateral elements, the structure defined on the body-fitted mesh with the same elements exhibits negligible relative errors: 4.59\% in initial stiffness, 1.20\% in end force, and 0.84\% in total energy. These results justify the use of voxel-based meshes in the proposed framework.

\begin{figure}[!htbp]
    \centering
    \includegraphics[width=18cm]{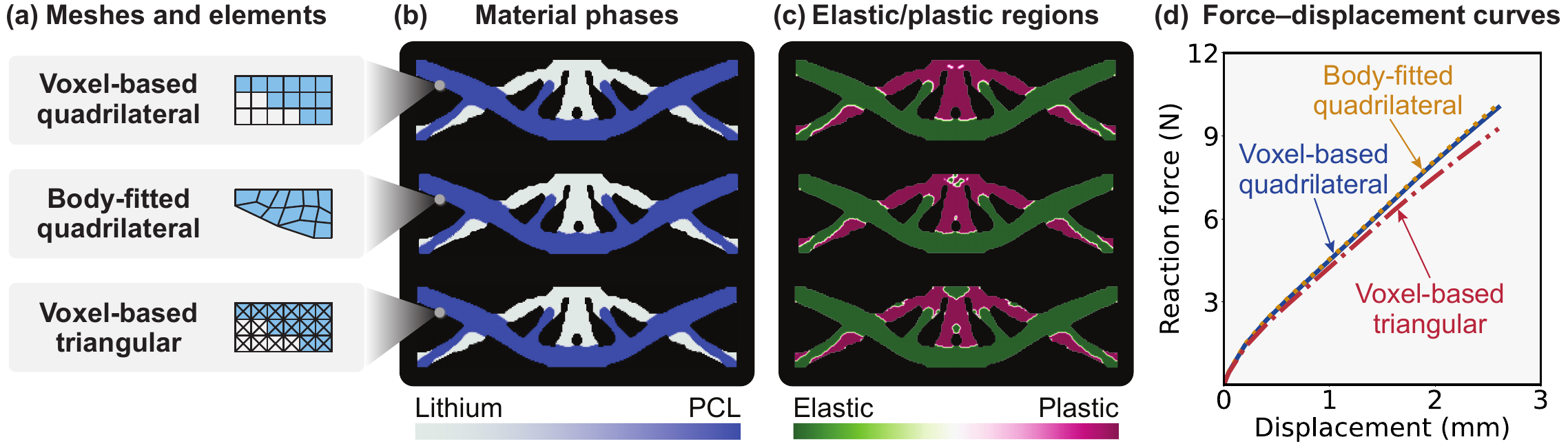}
    \caption{Comparisons among various meshes and elements. (a) Compared meshes and the associated finite elements. (b) Material phases. (c) Elastic/plastic regions. (d) Force--displacement ($F$--$u$) curves.}
    \label{Fig: Mesh Comparisons}
\end{figure}

\begin{table}[!htbp]
    \caption{Comparisons of the performance metrics among various meshes and elements}
    \label{Table: Mesh Comparisons}
    \centering
    \footnotesize
    \begin{tabular}{llllll}
        \hline
        \textbf{Meshes} & \textbf{Elements} & \textbf{Usage}
        & \begin{tabular}[x]{@{}l@{}}\textbf{Initial stiffness}\\(N/mm)\end{tabular}
        & \begin{tabular}[x]{@{}l@{}}\textbf{End force}\\(N)\end{tabular}
        & \begin{tabular}[x]{@{}l@{}}\textbf{Total energy}\\(N$\cdot$mm)\end{tabular} \\
        \hline
        Voxel-based, structured & Quadrilateral & Topology optimization & 10.25 & 10.03 & 14.31 \\
        Body-fitted, unstructured & Quadrilateral & Manufacturing & 10.72 (4.59\%) & 10.15 (1.20\%) & 14.43 (0.84\%) \\
        Voxel-based, structured & Triangular & Checking stiffening & 9.88 ($-$3.61\%) & 9.26 ($-$7.68\%) & 13.39 ($-$6.43\%) \\
        \hline
    \end{tabular}
\end{table}

In the context of finite strain elastoplasticity, we also need to examine the element types, as certain mesh and element combinations may exhibit over-stiffening behaviors in structures \citep{sloan_numerical_1982, wells_p-adaptive_2002} due to the isochoric plastic flow, despite the limited compressibility from elasticity. To mitigate these over-stiffening behaviors, higher-order and crossed linear triangular elements \citep{wells_p-adaptive_2002} have been identified as simple and effective strategies. However, both strategies are computationally expensive.

In this study, we opt for first-order quadrilateral and hexahedral elements to reduce computational time. As shown in Fig. \ref{Fig: Mesh Comparisons} and Table \ref{Table: Mesh Comparisons}, the performance metrics of the current setup remain close to the ``golden setup" of crossed linear triangular elements for plasticity \citep{wells_p-adaptive_2002}, with relative errors below 10\%. Additionally, we emphasize that the proposed design framework is independent of mesh and element types, especially with the versatile implementation of FEniTop \citep{jia_fenitop_2024}. Users have the flexibility to balance solution accuracy and computational cost based on their specific requirements.

\bibliographystyle{elsarticle-harv}
\biboptions{authoryear}
\bibliography{References}

\end{document}